\documentclass{article}
\usepackage[utf8]{inputenc}
\usepackage[margin=1in]{geometry}
\usepackage{complexity}
\usepackage{mathtools,amsmath,amssymb}
\usepackage{scrextend}
\usepackage{color,subfig}

\usepackage{footnote}
\newcommand{\astfootnote}[1]{
\let\oldthefootnote=\thefootnote
\setcounter{footnote}{0}
\renewcommand{\thefootnote}{\fnsymbol{footnote}}
\footnote{#1}
\let\thefootnote=\oldthefootnote
}

\usepackage{comment}
\usepackage{todonotes}
\usepackage{xspace}
\usepackage{multirow}
\usepackage{xcolor}
\RequirePackage[colorlinks=true]{hyperref}
\hypersetup{
	linkcolor=[rgb]{0,0,0.5},
	citecolor=[rgb]{0, 0.5, 0},
	urlcolor=[rgb]{0.5, 0, 0}
}
\usepackage{dsfont}

\usepackage{tikz}
\usetikzlibrary{positioning}
\usetikzlibrary{shapes,arrows}
\usetikzlibrary{shadows,arrows}
\newcommand*\circled[1]{\scalebox{.7}{\tikz[baseline=(char.base)]{\node[shape=circle,draw,inner sep=1.5pt] (char) {#1};}}}
\usepackage{tikzit}

\tikzstyle{Theorem}=[fill={red!5}, draw=red, shape=rectangle, text width=3.2cm, align=center, text=red]
\tikzstyle{Step}=[fill=white, shape=rectangle, text width=3cm, align=center]
\tikzstyle{Lemma}=[fill={blue!5}, draw=blue, shape=rectangle, text width=3.2cm, align=center, text=blue]

\usepackage{amsthm}
\usepackage{thmtools,thm-restate}

\numberwithin{equation}{section}
\declaretheoremstyle[bodyfont=\it,qed=\qedsymbol]{noproofstyle}

\declaretheorem[name=Observation,numbered=no]{observation*}

\declaretheorem[numberlike=equation]{problem}

\declaretheorem[numberlike=equation]{theorem}

\declaretheorem[name=Theorem,numbered=no]{theorem*}

\declaretheorem[numberlike=equation]{lemma}
\declaretheorem[name=Lemma,numbered=no]{lemma*}

\declaretheorem[numberlike=equation]{corollary}
\declaretheorem[name=Corollary,numbered=no]{corollary*}

\declaretheorem[name=Proposition,numbered=no]{proposition*}

\declaretheorem[numberlike=equation]{claim}
\declaretheorem[name=Claim,numbered=no]{claim*}

\declaretheorem[name=Conjecture,numbered=no]{conjecture*}

\declaretheorem[name=Question,numbered=no]{question*}

\declaretheoremstyle[bodyfont=\it]{defstyle} 

\declaretheorem[numberlike=equation,style=defstyle]{definition}
\declaretheorem[unnumbered,name=Definition,style=defstyle]{definition*}

\declaretheorem[unnumbered,name=Example,style=defstyle]{example*}

\declaretheorem[unnumbered,name=Notation=defstyle]{notation*}

\declaretheorem[unnumbered,name=Construction,style=defstyle]{construction*}

\declaretheoremstyle[]{rmkstyle}

\newcommand{\N}{\mathbb{N}}
\newcommand{\Nz}{\mathbb{N}_{\geq0}}
\newcommand{\Real}{\mathbb{R}}
\newcommand{\Exp}{\mathbb{E}}
\newcommand{\Var}{\mathsf{Var}}

\newcommand{\gap}{\textsf{gap}}
\newcommand{\diag}{\textsf{diag}}

\newcommand{\bA}{\mathbf{A}}

\newcommand{\be}{\mathbf{e}}
\newcommand{\bff}{\mathbf{f}}
\newcommand{\bw}{\mathbf{w}}

\newcommand{\bv}{\mathbf{v}}
\newcommand{\bx}{\mathbf{x}}
\newcommand{\bz}{\mathbf{z}}

\makeatletter
\def\moverlay{\mathpalette\mov@rlay}
\def\mov@rlay#1#2{\leavevmode\vtop{%
   \baselineskip\z@skip \lineskiplimit-\maxdimen
   \ialign{\hfil$\m@th#1##$\hfil\cr#2\crcr}}}
\newcommand{\charfusion}[3][\mathord]{
    #1{\ifx#1\mathop\vphantom{#2}\fi
        \mathpalette\mov@rlay{#2\cr#3}
      }
    \ifx#1\mathop\expandafter\displaylimits\fi}
\makeatother

\def\mycoauthor{1} 

\ifx\mycoauthor\undefined
\author{
Chi-Ning Chou\thanks{School of Engineering and Applied Sciences, Harvard University, Cambridge, Massachusetts, USA. Supported by NSF awards CCF 1565264 and CNS 1618026. Email: \texttt{chiningchou@g.harvard.edu}.}
\and Mien Brabeeba Wang\thanks{CSAIL, MIT, Cambridge, Massachusetts, USA. Email: \texttt{brabeeba@mit.edu}.}
}
\else
\author{
Chi-Ning Chou~\footnotemark[2]~\footnotemark[3]
\and Mien Brabeeba Wang~\footnotemark[2]~\footnotemark[4]
}  
\fi

\title{ODE-Inspired Analysis for the Biological Version of Oja's Rule in Solving Streaming PCA\footnote{Accepted for presentation at the Conference on Learning Theory (COLT) 2020.}}

\begin{document}

\maketitle

\ifx\mycoauthor\undefined
\else
\renewcommand{\thefootnote}{\fnsymbol{footnote}}
\footnotetext[2]{The two authors contribute equally to this paper.}
\footnotetext[3]{School of Engineering and Applied Sciences, Harvard University, Cambridge, Massachusetts, USA. Supported by NSF awards CCF 1565264 and CNS 1618026. Email: \texttt{chiningchou@g.harvard.edu}.}
\footnotetext[4]{CSAIL, MIT, Cambridge, Massachusetts, USA. Supported by NSF Awards CCF-1810758, CCF-0939370, CCF-1461559 and Akamai Presidential Fellowship. Email: \texttt{brabeeba@mit.edu}.}
\renewcommand{\thefootnote}{\arabic{footnote}}
\fi

\begin{abstract}
Oja's rule [Oja, Journal of mathematical biology 1982] is a well-known biologically-plausible algorithm using a Hebbian-type synaptic update rule to solve streaming principal component analysis (PCA). Computational neuroscientists have known that this biological version of Oja's rule converges to the top eigenvector of the covariance matrix of the input in the limit. 
However, prior to this work, it was open to prove any convergence rate guarantee.

In this work, we give the first convergence rate analysis for the biological version of Oja's rule in solving streaming PCA. Moreover, our convergence rate matches the information theoretical lower bound up to logarithmic factors and outperforms the state-of-the-art upper bound for streaming PCA.
Furthermore, we develop a novel framework inspired by ordinary differential equations (ODE) to analyze general stochastic dynamics. The framework abandons the traditional \textit{step-by-step} analysis and instead analyzes a stochastic dynamic in \textit{one-shot} by giving a closed-form solution to the entire dynamic. The one-shot framework allows us to apply stopping time and martingale techniques to have a flexible and precise control on the dynamic. We believe that this general framework is powerful and should lead to effective yet simple analysis for a large class of problems with stochastic dynamics.
\end{abstract}

\newpage
\tableofcontents
\newpage

\section{Introduction}
Brains processes high dimensional visual inputs constantly. In our eyes, 100 millions photoreceptors in the retina receive gigabytes of information per second \cite{Wandell1995, Snyder1977}. In addition, the retina is a highly convergent pathway: $100$ million photoreceptors converges the visual information onto one million retina ganglion cells in optical nerves~\cite{Ganguli2012}. Therefore, it is important to understand a neural implementation of the dimensionality reduction in the retina. Furthermore, many works in theoretical neuroscience~\cite{Atick1990, Atick1992, Hemmen1998} demonstrated from the efficient coding principle that the retina might implement Principal Component Analysis (PCA). Specifically, they showed that under natural image statistics, PCA-like solution recovers the center-surround receptive fields in the retina. However, their work only proposed PCA as an potential solution to the pathway and did not provide a dynamic to explain the learning process of PCA.

On the other hand, in the seminal work of~\cite{Oja82}, he proposed a mathematical model for the biological neural network that solves streaming PCA with several biologically-plausible properties: the network not only updates its synaptic weights locally but also normalizes the strength of synapses. This rule, now known as the biological version of Oja's rule (\textit{biological Oja's rule}\footnote{Also known as \textit{Oja's rule} in the literature. However, many works in the machine learning community use the name ``Oja's rule'' for \textit{non-biologically-plausible} variants of the original Oja's rule. Thus, in this paper we emphasize the term ``\textit{biological}'' to distinguish the two. See~\autoref{sec:related works} for more discussions on their differences.} in abbreviation), has been the subject of extensive theoretical~\cite{Oja82,OK85,Sanger89,HKP91,Oja92,Plumbley1995,DK96,Zufiria02,YYLT05,Duflo13,ACS13} and experimental~\cite{CL94,KDT94,HP94,Karayiannis96,CKS96,SLY06,SA06,LTYH09,ACS13} studies aimed at understanding its performance. Despite its popularity, the theoretical understanding of the biological Oja's rule cannot account for the biologically-realistic time scale in the retina-optical nerve pathway because the state-of-the-art theoretical analysis only provides a guarantee on convergence in the limit \cite{Duflo13}.

In practice, the retina can change its receptive field to adapt to environments with different illumination \cite{Shapley1984}, contrast \cite{Shapley1984, Baccus2002, Smirnakis1997}, spatial frequency \cite{Smirnakis1997, Hosoya2005}, orientation and temporal correlation \cite{Hosoya2005} in the time scale of seconds. This suggests that a plausible dynamic for explaining the retina-optical nerve pathway should have little or no dependency on the dimension, \textit{i.e.,} the number of neurons, which in this case is on the order of $100$ million.
Meanwhile, researchers have observed that the biological Oja's rule (and its variants) has fast convergence rates~\cite{HP94,Karayiannis96,SLY06,SA06,LTYH09} in simulations. 
Thus, to further our understanding in the retina-optical nerve pathway, it is important to give a theoretical analysis to show that the biological Oja's rule solves streaming PCA in a biologically-realistic time scale. This is nevertheless a challenging task and has remained elusive for almost 40 years~\cite{Oja82}.\\

In this work, we provide the first convergence rate analysis for the biological Oja's rule in solving streaming PCA.

\begin{theorem}[informal]
The biological Oja's rule efficiently solves streaming PCA with (nearly) optimal convergence rate. Specifically, the convergence rate we obtain matches the information theoretical lower bound up to logarithmic factors. 

Furthermore, the convergence rate has no dependency on the dimension when the initial weight vector is close to the top eigenvector or has a logarithmic dependency on the dimension when the initial vector is random. Therefore, the biological Oja's rule solves streaming PCA in a biologically-realistic time scale.
\end{theorem}

To show the (nearly) optimal convergence rate of biological Oja's rule in solving streaming PCA, we develop an ODE-inspired framework to analyze stochastic dynamics. 
Concretely, instead of the traditional \textit{step-by-step} analysis, our framework analyzes a dynamical system in \textit{one-shot} by giving a closed-form solution for the entire dynamic. The framework borrows ideas from ordinary differential equations (ODE) and stochastic differential equations (SDE) to obtain a closed-form characterization of the dynamic and uses stopping time and martingale techniques to precisely control the dynamic. This framework provides a more elegant and more general analysis compared with the previous step-by-step approaches. We believe that this novel framework can provide simple and effective analysis on other problems with stochastic dynamics.\\

We organize the rest of the introduction as follows. We first formally define biological Oja's rule and streaming PCA in~\autoref{sec:def} and state the main results and their biological relevance in~\autoref{sec:result}. In~\autoref{sec:technical overview}, we provide a technical overview on the proof and the analysis framework. Finally, we conclude the introduction with  a survey and comparison of related works in~\autoref{sec:related works}.

\subsection{Biological Oja's rule and streaming PCA}\label{sec:def}
In a biological neural network, two neurons primarily interact with each other via action potentials or instantaneous signals, \textit{a.k.a.}, "spikes", through \textit{synapses} between them.
The strength of a synapse might vary from time to time and is called the \textit{synaptic weight}.
The ability of a synaptic weight to strengthen or weaken over time is considered as a source for learning and long term memory in our brains.
While generally the update of a synaptic weight could depend on the \textit{spiking patterns} of the end neurons, it is common for neuroscientists to focus on the averaging behaviors of a spiking dynamic. Namely, they simplify the model by only considering the \textit{firing rate}, which is defined as the average number of spikes. This is known as the \textit{rate-based model}~\cite{WC72,WC73} and since the biological Oja's rule was defined on a rate-based model, this setting is going to be the focus of this work.

To understand how the biological Oja's rule works, consider the following baby example with two neurons $x$ and $y$.
Let $x_t,y_t\in\Real$ be the firing rates of neurons $x,y$ at time $t\in\N$ and let $w_t\in\Real$ be the synaptic weight from $x$ to $y$ at time $t$. In a biological neural network, $w_t$ could change over time and the dynamic is defined \textit{locally} on the previous synaptic weight as well as the firing rates of the end neurons. Namely, the synaptic weight from the neuron $x$ to $y$ has the following dynamic
\[
w_t = w_{t-1} + \eta_tF_t(w_{t-1},x_t,y_t)
\]
where $F_t$ is an update function and $\eta_t$ is the \textit{plasticity coefficient}, \textit{a.k.a.,} the \textit{learning rate}. Biologically, the update function should further follow the Hebb postulate,``cells that fire together wire together"~\cite{Hebb49}. One naive way to implement Hebbian learning is to set the update function as $F_t(w_{t-1},x_t,y_t)=x_ty_t$. However, the values of $w_t$ can grow unboundedly. The biological Oja's rule is a self-normalizing Hebbian rule with the following synaptic updates.
\[
w_t = w_{t-1} + \eta_ty_t\left(x_t-y_tw_{t-1}\right) \, .
\]
Using the above synaptic update rule, Oja~\cite{Oja82} configured a network that solves streaming PCA while keeping the norm of the weights stable. Before introducing the network, let us formally define the streaming PCA problem.

\paragraph{Streaming PCA.}
Principal component analysis (PCA)~\cite{Pearson01,Hotelling33} is a problem to find the top eigenvector of a covariance matrix of a dataset. Let $n$ be the dimension of the data. In the offline setting, one can compute the covariance matrix in $O(n^2)$ space and use the power method to approximate the top eigenvector.
As for its variant, the streaming PCA (a.k.a. the stochastic online PCA, see~\cite{CG90} for a survey on the literature), the input data arrives in a stream and the algorithm/dynamic only has limited amount of space, \textit{e.g.,} $O(n)$ space.
Streaming PCA is important for biological system because the information inherently arrives in a stream in a living system. On the other hand, it is also much more challenging than offline PCA (see for example~\cite{AL17}).
In the following, we formally define the streaming PCA problem.~\footnote{In related works, some (\textit{e.g.,}~\cite{AL17}) measure the error using $1-\langle\bw,\bv_1\rangle^2$, some (\textit{e.g.,}~\cite{Shamir16}) use $1-\bw^\top A\bw/\|A\|$, and some (\textit{e.g.,}~\cite{JJKNS16}) use $\sin^2(\bw,\bv_1)$. We remark that all of these error measures (including ours) are the same up to a constant multiplicative factor.

Also, some works emphasize other convergence notions such as the gap-free convergence~\cite{Shamir16}. Though we do not explicitly study the convergence of biological Oja's rule under these notions, we believe that our results could be easily extended to other convergence notions with comparable convergence rate and leave this for future work.}

\begin{problem}[Streaming PCA]\label{prob:streaming PCA}
Let $n,T\in\N$ and $\mathcal{D}$ be a distribution over the unit sphere of $\Real^n$. Suppose the input data $\bx_1,\bx_2,\dots,\bx_{T}\stackrel{\text{i.i.d.}}{\sim}\mathcal{D}$ are given one by one in a stream.
Let $A = \Exp_{\bx\sim\mathcal{D}}[\bx\bx^\top]$ be the covariance matrix and $\lambda_1\geq\lambda_2\geq\cdots\geq\lambda_n\geq0$ be the eigenvalues of $A$. Assume $\lambda_1>\lambda_2$ and let $\bv_1$ be the top eigenvector of $A$ of unit length. Then the goal of the streaming PCA problem is to output $\bw\in\Real^n$ such that $\frac{\langle\bw,\bv_1\rangle^2}{\|\bw\|_2^2}\geq1-\epsilon$.
\end{problem}

Since the inputs arrive in a stream, usually a streaming PCA algorithm/dynamic would maintain a solution $\bw_t\in\Real^n$ at each time $t\in\N$. Thus, the goal for a streaming PCA algorithm/dynamic would be achieving $\Pr\left[\frac{\langle\bw_T,\bv_1\rangle^2}{\|\bw\|_2^2}\geq1-\epsilon\right]\geq1-\delta$ with small $T$.

\paragraph{Biological Oja's rule in solving streaming PCA.}
Oja~\cite{Oja82} proposed a streaming PCA algorithm using $n$ input neurons and one output neuron. The firing rates of the input neurons at time $t$ are denoted by a vector $\bx_t\in\Real^n$ and the firing rate of the output neuron is denoted by a scalar $y_t\in\Real$. The synaptic weights at time $t$ from the input neurons to the output neuron are denoted by a vector $\bw_t\in\Real^n$. Note that the weight vector will be the output and ideally it will converge to the top eigenvector $\bv_1$.

The input stream $\bx_1,\bx_2,\dots,\bx_{T}$ arrives in the form of firing rates of the input neurons.
The firing rate of the output neuron is simply the inner product of the synaptic weight vector and the firing rate vector of the input neurons, \textit{i.e.,} $y_t=\bx_t^\top\bw_{t-1}$.
Now, from the biological Oja's rule, the dynamic of the synaptic weight vector is described by the following equation.

\begin{definition}[Biological Oja's rule]\label{def:bio oja}
For any initial vector $\bw_0\in\Real^n$ such that $\|\bw_0\|_2=1$, the dynamic of the biological Oja's rule is the following. For any $t\in\N$, define 
\begin{equation}\label{eq:oja}
\bw_t = \bw_{t-1} + \eta_ty_t\left(\bx_{t}-y_t\bw_{t-1}\right)
\end{equation}
where $y_t=\bx_t^\top\bw_{t-1}$ and $\bx_t$ is the input at time $t$. See also in~\autoref{fig:bio oja} for a pictorial definition of biological Oja's rule in solving streaming PCA.
\end{definition}

\begin{figure}[ht]
\centering
 \includegraphics[width=8cm]{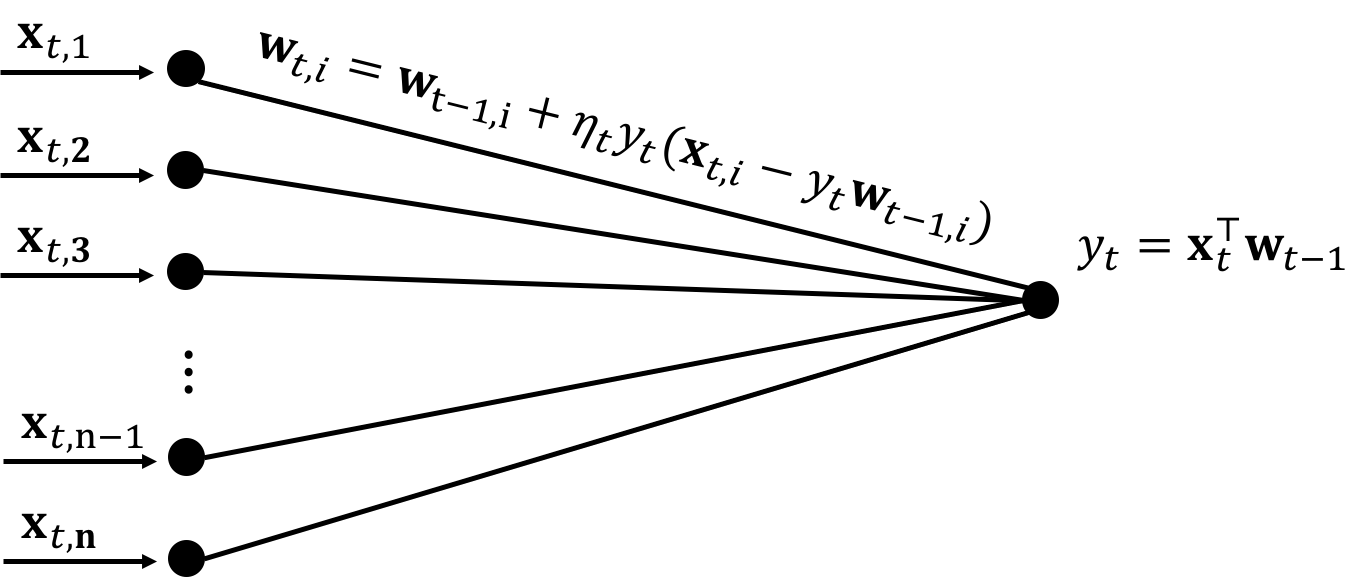}
\caption{A neural network that uses biological Oja's rule to solve streaming PCA. The firing rate vector $\bx_t$ is the input and the weight vector $\bw_t$ is the output at time $t$.}
\label{fig:bio oja}
\end{figure}

Follow from the definition, the biological Oja's rule is automatically \textit{biologically-plausible} in the following sense. First, the synaptic update rule is \textit{local}. Namely, each synapse only depends on the previous synaptic weight and the firing rates of the two end neurons. Second, with some simple calculations (\textit{e.g.,}~\autoref{lem:oja w norm ub}), biological Oja's rule achieves the \textit{synaptic scaling guarantee}~\cite{Abbott2000}, \textit{i.e.,} $\bw_{t,i}$ being bounded for all $t\in\N$ and $i\in[n]$. 
Thus, one can then interpret the convergence results of this work as showing further biological-plausibilities of the biological Oja's rule in the retina-optical nerve pathway. See~\autoref{sec:result} for more discussions.

\paragraph{Oja's derivation for the biological Oja's rule.}
Before going into more technical contents, it would be helpful to take a look at the original derivation for the biological Oja's rule. Initially, Oja wanted to use the following update rule with normalization\footnote{This update rule is doing a variant of power method with normalization. It is widely used in the machine learning community to solve streaming PCA. See~\autoref{sec:related works} for more discussion.} to solve the streaming PCA problem.
\begin{equation}\label{eq:oja normalized}
\bw_t = \frac{\left(I+\eta_t\bx_t\bx_t^\top\right)\bw_{t-1}}{\|\left(I+\eta_t\bx_t\bx_t^\top\right)\bw_{t-1}\|_2} \, .
\end{equation}
However, the normalization term $\|\left(I+\eta_t\bx_t^\top\bx_t\right)\bw_{t-1}\|_2^{-1}$ is \textit{global}\footnote{It is global because computing the $\ell_2$ norm requires the information from \textit{every} neurons.} and does not seem to have a biologically-plausible implementation. To bypass this issue, Oja applied \textit{Taylor's expansion} on the normalization term and truncated the second order terms of $\eta_t$. This exactly results in the biological Oja's rule (\textit{i.e.,}~\autoref{eq:oja}). See~\autoref{sec:oja derivation} for more details on the derivation.

Also, to see why intuitively biological Oja's rule could solve streaming PCA, one can check that any eigenvector $\bv$ of $A$ of unit length with eigenvalue $\lambda$ is a fixed point of the biological Oja's rule in expectation. 
Specifically, the expectation of the update term $y_t(\bx_t-y_t\bw_{t-1})$ with $\bw_{t-1}=\bv$ is the following.
\[
\Exp\left[\bx_t^\top\bv\bx_t-(\bx_t^\top\bv)^2\bv\right] = A\bv- \bv^\top A\bv\bv = \lambda\bv-\lambda\|\bv\|_2^2\bv=0 \, .
\]
The first equality follows from for all $i,j\in[n]$, $\Exp[\bx_{t,i}\bx_{t,j}]=\lambda_i\cdot\mathbf{1}_{i=j}$, and the second equality follows from $A\bv=\lambda\bv$. By checking the Hessian at the top eigenvector $\bv_1$, one can even see that $\bv_1$ is a \textit{stable} fixed point.

\paragraph{Previous works: Convergence in the limit results.}
There were many previous works on analyzing the convergence of biological Oja's rule in solving streaming PCA~\cite{Oja82,OK85,Sanger89,HKP91,Oja92,Plumbley1995,DK96,Zufiria02,YYLT05,Duflo13}. However, their works only proved guarantee on convergence in the limit. For example, Duflo~\cite{Duflo13} showed that $\bw_t$ converges to the top eigenvector of $A$ in the limit under some constraints on the learning rates.
\begin{theorem}[\cite{Duflo13}, informal]
Let $\bw_0$ be a random unit vector in $\Real^n$. If $\eta_{t}\leq\frac{1}{2}$ for all $t\in\N$, $\sum_{t=0}^\infty\eta_t=\infty$, and $\sum_{t=0}^\infty\eta_t^2<\infty$, then $\lim_{t\rightarrow\infty}\langle\bw_t,\bv_1\rangle^2=1$ almost surely.
\end{theorem}
The proofs of these previous analyses are usually based on tools from dynamical system such as the Kushner-Clark method or Lyapunov theory. Note that these proof techniques are not sufficient for providing convergence rate guarantee.

Prior to this work, there had been no efficiency guarantee for the biological Oja's rule. 
The main technical barrier is due to the non-linear terms in the update rule which introduces correlations in the traditional step-by-step analysis and thus naive analysis would not work.
We explain the difficulty further in~\autoref{sec:technical overview} and~\autoref{sec:compare bio ML}.
Given this situation, natural questions on the frontier would then be:
\begin{quote}
	\textbf{Question}: What is the convergence rate of biological Oja's rule in solving streaming PCA? Is the convergence rate biologically-realistic?
\end{quote}

\subsection{Our results}\label{sec:result}

In this paper, we answer the above questions by giving the first convergence rate guarantee for the biological Oja's rule in solving streaming PCA.
Furthermore, the convergence rate matches the information-theoretic lower bound for streaming PCA up to logarithmic factors.
In terms of the techniques, we develop an ODE-inspired framework to analyze stochastic dynamics.
We believe this general framework of using tools and insights from ODE and SDE in analyzing stochastic dynamics is elegant and powerful. We provide more details and intuitions on the ODE-inspired framework in the section on the technical overview (see~\autoref{sec:technical overview}). Also, as a byproduct, our convergence rate guarantee for biological Oja's rule outperforms the state-of-the-art upper bound for streaming PCA (using other variants of Oja's rule).

There are two common convergence notions in the streaming PCA literature. The \textit{global convergence} requires the algorithm/dynamic to start from a random initial vector while the \textit{local convergence} allows the algorithm/dynamic to start from an initial vector that is highly correlated to the top eigenvector of the covariance matrix.
Now, we are ready to state our main theorem as follows.

\begin{theorem}[Global and local convergence]\label{thm:informal}
With the setting in~\autoref{prob:streaming PCA} and dynamic in~\autoref{def:bio oja}, let $\gap:=\lambda_1-\lambda_2>0$.
For any $\epsilon,\delta\in(0,1)$, we have the following results.\\
$\bullet$ (Local Convergence) Suppose $\frac{\langle\bw_0, \bv_1\rangle^2}{\|\bw_0\|_2^2} = \Omega(1)$. For any $n\in\N$, $\delta,\epsilon\in(0,1))$, let
\[
\eta = \tilde{\Theta}\left(\frac{\epsilon\gap}{\lambda_1}\right),\ T = \Theta\left(\frac{\lambda_1}{\epsilon\gap^2}\cdot \log^2\left(\frac{1}{\epsilon},\, \frac{1}{\delta}\right)\right).
\]
Then, we have 
\[
\Pr\left[\frac{\langle\bw_T, \bv_1\rangle^2}{\|\bw_T\|_2^2} < 1-\epsilon\right] < \delta \,.
\]
$\bullet$ (Global Convergence) Suppose $\bw_0$ is uniformly sampled from the unit sphere of $\Real^n$. For any $n\in\N$, $\delta,\epsilon\in(0,1)$, let
\[
\eta = \tilde{\Theta}\left(\frac{(\epsilon\wedge\delta^2)\gap}{\lambda_1}\right),\,T = \Theta\left(\frac{\lambda_1}{(\epsilon\wedge\delta^2)\gap^2}\cdot\log^3\left(\frac{1}{\epsilon},\, \frac{1}{\delta},\, \frac{1}{\gap}, n\right)\right)
\]
Then, we have
\[
\Pr\left[\frac{\langle\bw_T, \bv_1\rangle^2}{\|\bw_T\|_2^2} < 1-\epsilon\right] < \delta \,.
\]
The notation $a\wedge b$ stands for $\min\lbrace a, b\rbrace$ and $\tilde{\Theta}$ hides the poly-logrithmic factors with respect to $\epsilon^{-1}, \delta^{-1}, \gap^{-1}, n$.
\end{theorem}

\paragraph{Biological perspectives.}
Our results provide further theoretical evidences for the biological plausibility of biological Oja's rule to be a likely candidate of the dimensionality reduction in the retina-optical nerve pathway. Specifically, we show that \textit{``biological Oja's rule is a local Hebbian learning rule with bounded synaptic weights that functions in a biologically-realistic time scale."} In particular, in this work we demonstrate that biological Oja's rule  does not have any dependency on the dimension (\textit{i.e.,} $n$, the number of neurons) in the local convergence setting while the dependency is logarithmic in the global convergence setting. 
Moreover, in the local convergence setting, the dependency of the convergence rate on the failure probability $\delta$ is inverse-logarithmic in stead of $O(1/\delta)$.

Furthermore, we prove the \textit{for-all-time} guarantee of the biological Oja's rule as an corollary of the techniques used in the proof for the main theorems. By for-all-time guarantee we refer to the behavior of a dynamic that \textit{always} stays around the optimal solution after convergence. Especially, the dynamic would not temporarily leave the neighborhood of the optimal solution. The for-all-time guarantee is of biological importance because a biological system constantly adapts and functions, and it is not enough for a mechanism to hold for only a brief moment. 
We state the theorem for the for-all-time guarantee as follows.
\begin{theorem}[For-all-time guarantee with slowly diminishing rate]\label{thm:for all time informal}
With the setting in~\autoref{prob:streaming PCA} and dynamic in~\autoref{def:bio oja}, let $\gap:=\lambda_1-\lambda_2>0$.
For any $\epsilon,\delta\in(0,1)$, suppose $\frac{\langle\bw_{0},\bv_1\rangle^2}{\|\bw_0\|_2^2}\geq1-\epsilon/2$. For any $t\in\N$, there exists $\eta_t\geq \Theta\left(\frac{\epsilon\cdot\gap}{\lambda_1\log(t/\delta)}\right)$ such that
\[
\Pr\left[\forall t\in\N,\ \frac{\langle\bw_t,\bv_1\rangle^2}{\|\bw_t\|_2^2}\geq1-\epsilon\right]\geq1-\delta \, .
\]
\end{theorem}

We should further notice that the learning rate is slowly-diminishing, \textit{i.e.,} $\eta_t=\Theta(1/\log t)$ instead of the commonly used $\eta_t=O(1/t)$, in the for-all-time guarantee (\textit{i.e.,}~\autoref{thm:for all time informal}). Because our learning rate is slowly diminishing, when the environment changes, the learning rate is still large enough to do efficient learning. This allows the sensory system to continuously adapt to changing environments without taking a long time to adapt or reset the learning rate. This suggests the capability of \textit{continual adaptation}, which is crucial in the biological scenario. For example, if a person walks into a new environment, the retina cells need to quickly adapt to the new environment and this cannot be achieved if the learning rate already diminished too fast in the previous environment. 

We remark that prior to this work, the for-all-time guarantee with slowly diminishing learning rates was even unknown to any streaming PCA algorithms. The convergence in the limit result for biological Oja's rule requires $\eta_t=o(1/\sqrt{t})$~\cite{Duflo13} and the convergence rate analysis for non-biologically-plausible variants of Oja's rule requires $\eta_t=\tilde{O}(1/t)$~\cite{JJKNS16,AL17,LWLZ18} or $\eta_t=O(1/\sqrt{t})$~\cite{Shamir16}. In particular, all previous works satisfy $\sum_{t}\eta_t^2 < \infty$ while in this work we can achieve for-all-time convergence with much weaker assumptions $\eta_t=\Theta(1/\log t)$ (hence $\sum_t \eta_t^2 = \infty$) for the biological Oja's rule.

\subsection{Technical overview}\label{sec:technical overview}

In this work, we give the first efficiency guarantee for the biological Oja's rule in solving streaming PCA with an (nearly) optimal convergence rate.
In this subsection, we highlight three technical insights of our analysis which lead us to a clean understanding in how the biological Oja's rule solves streaming PCA. In short, our high-level strategy is to first consider the \textit{continuous} version of the Oja's rule where the learning rate $\eta$ is set to be infinitesimal.
In the continuous setting, the dynamic can be fully understood by tools from the theory of ordinary differential equations (ODE) or stochastic differential equations (SDE). With the inspiration from the continuous analysis, we are able to identify the right tools (\textit{e.g.,} linearization at two different centers, etc.) to tackle the discrete dynamic.

Before we start, let us recall the problem setting and the goal.
For simplicity, here we consider the \textit{diagonal case} where the covariance matrix $A$ is a diagonal matrix, \textit{i.e.,} $A=\diag(\lambda_1\dots,\lambda_n)$ with $\lambda_1>\lambda_2\geq\lambda_3\geq\cdots\geq\lambda_n\geq0$. Thus the top eigenvector of $A$ is $\mathbf{e}_1$, \textit{i.e.,} the indicator vector for the first coordinate, and the goal becomes showing that $\bw_{t,1}^2$ efficiently converges to $1$ when $t\rightarrow\infty$. A reduction from the general case to the diagonal case is provided in~\autoref{sec:diagonal}.

\paragraph{Insight 1: Inspiration from the continuous dynamics.}
The first insight is to analyze the biological Oja's rule in a way inspired by its continuous analog. 
The advantage to consider the continuous dynamics is that not only it captures the inherent dynamics but also we can apply the theory of ODE and SDE to obtain \textit{closed-form} solutions. 
Thus, the continuous dynamic would serve as a hint on how to derive a tight and closed-form analysis for the discrete dynamic.

Interestingly, the continuous SDE of the biological Oja's rule degenerates into a simple deterministic ODE almost surely (see~\autoref{sec:continuous Oja} for a derivation). Specifically, for any $t\geq0$, we have
\begin{equation}\label{eq:intro continuous}
\frac{d\bw_{t, 1}}{dt} \geq (\lambda_1 - \lambda_2)\bw_{t, 1}(1 - \bw_{t, 1}^2)\ \ \text{ and }\ \ \|\bw_t\|_2=1
\end{equation}
almost surely. Furthermore, observe that the continuous Oja's rule is non-decreasing and has three fixed points 0 and $\pm1$ for $\bw_{t,1}$ while the first is unstable and the later two are stable. Namely, in the continuous dynamic, $\bw_t$ will eventually converge to $\pm\be_1$, \textit{i.e.,} the top eigenvector of $A$.

Note that in a discrete stochastic dynamic, there are two sources of the noises: (i) the intrinsic stochasticity from its continuous analog and (ii) the noise due to discretization. Thus,~\autoref{eq:intro continuous} suggests that the noise in the biological Oja's rule only comes from discretization since the continuous Oja's rule is deterministic.

In addition to the limiting behavior, one can also read out finer structures of the continuous dynamic from~\autoref{eq:intro continuous} by solving the differential equation using standard tools from dynamical system. The right hand side (RHS) of the inequality in~\autoref{eq:intro continuous} is non-linear which usually does not have a clean solution. A natural idea from dynamical system would then be \textit{linearizing} the differential equation around fixed points and applying the \textit{exact} solution for a linear ordinary differential equation. Moreover, as there are three fixed points in~\autoref{eq:intro continuous}, one can linearize the differential equation with center being either 0 or $\pm1$. For simplicity, we focus on the two fixed points $0$ and $1$ while $-1$ can be analyzed similarly due to symmetry.

For example, we can linearize at 0 by lower bounding the RHS of~\autoref{eq:intro continuous} by $\epsilon(\lambda_1-\lambda_2)\bw_{t,1}$ for any $\bw_{t,1}\in[0,\sqrt{1-\epsilon}]$ (see Figure~\ref{fig:continuous 1 sided 0}). 
Similarly, we can linearize at 1 by using $\bw_{0,1}(\lambda_1-\lambda_2)(1-\bw_{t,1})$ for any $\bw_{t,1}\in[\bw_{0,1},1]$ (see Figure~\ref{fig:continuous 1 sided 1}).
Another choice would be \textit{linearizing at both 0 and 1}. Concretely, we linearize at 0 for $\bw_{t,1}\in[0,2/3]$ and linearize at 1 for $\bw_{t,1}\in[2/3,1]$ (see Figure~\autoref{fig:continuous 2 sided}).

\begin{figure}[ht]
    \centering
    \subfloat[Linearization only at $0$.]{%
		\includegraphics[width=0.3\textwidth]{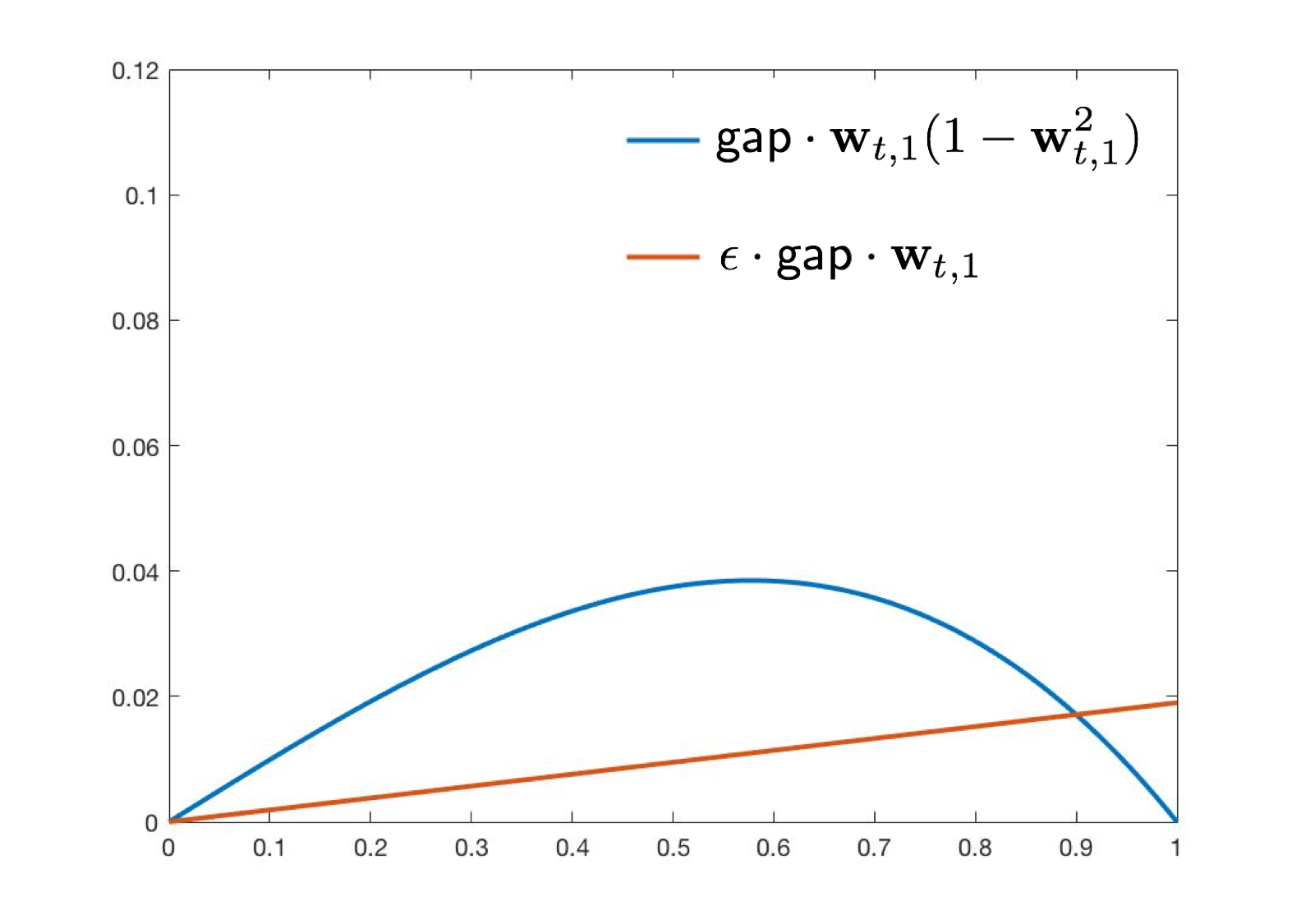}%
		\label{fig:continuous 1 sided 0}%
	}\qquad
	\subfloat[Linearization only at $1$.]{%
		\includegraphics[width=0.3\textwidth]{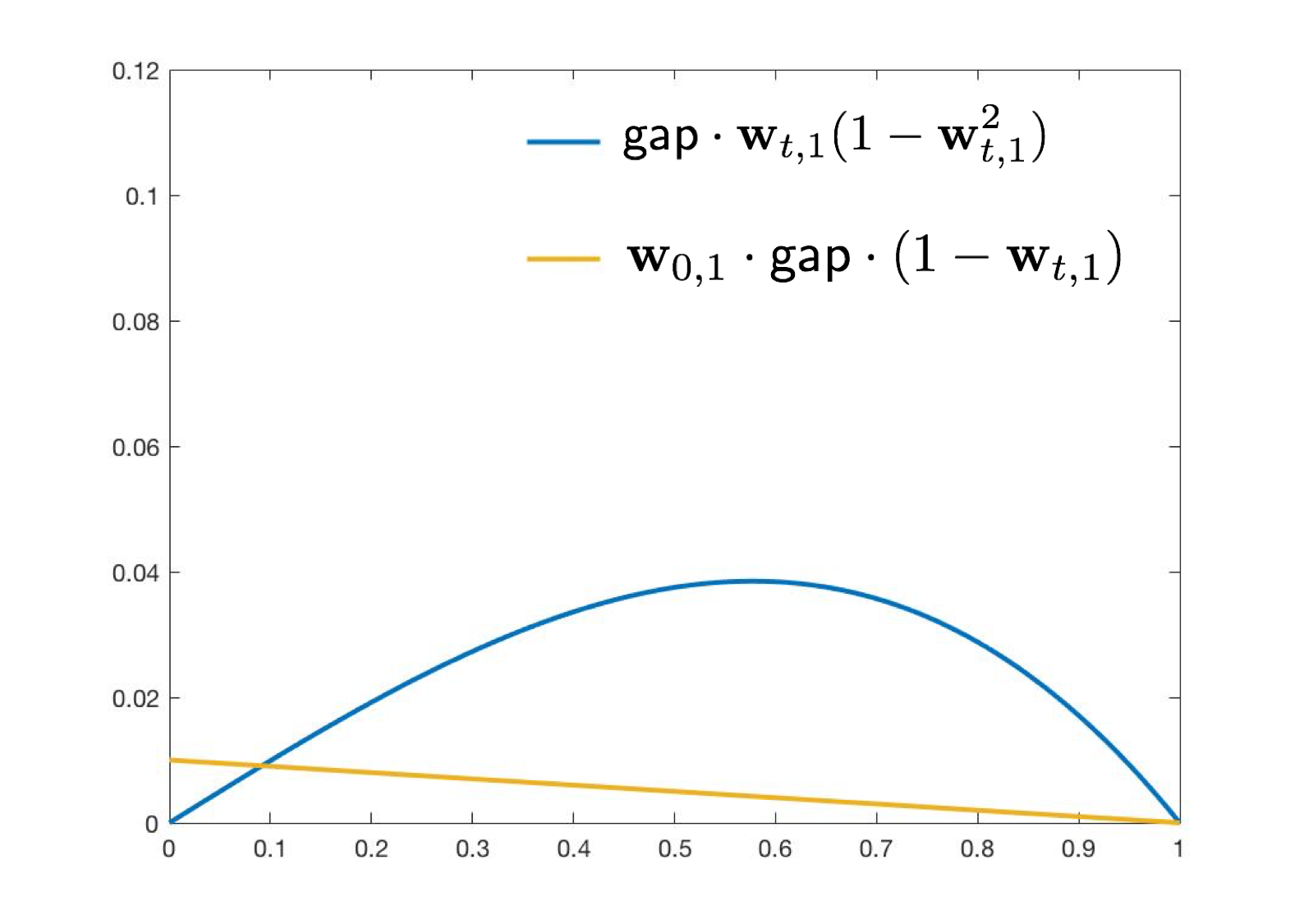}%
		\label{fig:continuous 1 sided 1}%
	}\qquad
	\subfloat[Linearization at both $0$ and $1$.]{%
		\includegraphics[width=0.3\textwidth]{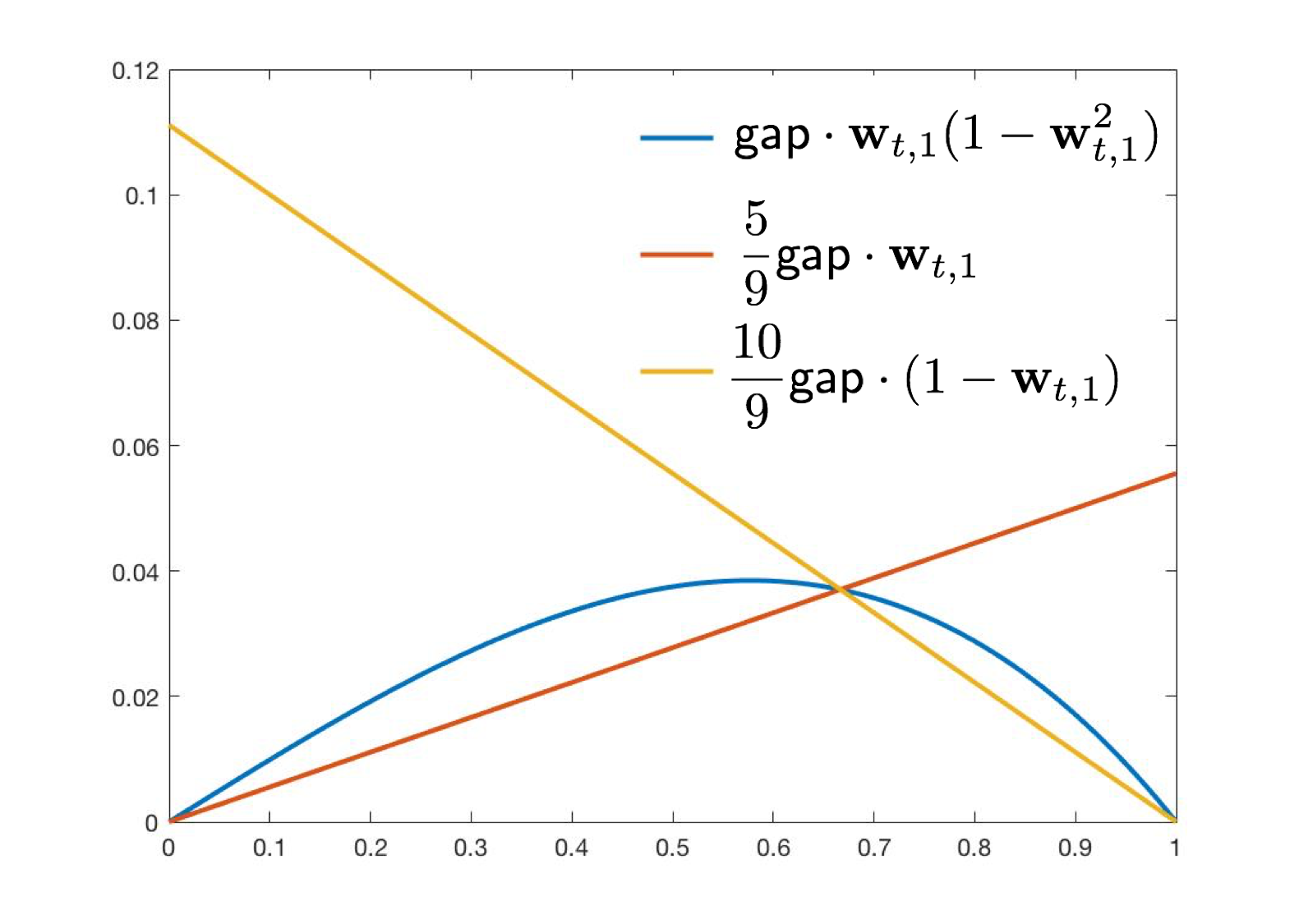}%
		\label{fig:continuous 2 sided}%
	}
    \caption{Comparison between one-sided linearization and two-sided linearization.}
    \label{fig:continuous}
\end{figure}

The main difference between linearizing only at a single fixed point and linearizing at two fixed points is the \textit{slope} in the linearization.
Note that the slopes of the linearizations in Figure~\autoref{fig:continuous 1 sided 0} and Figure~\autoref{fig:continuous 1 sided 1} are $\epsilon(\lambda_1-\lambda_2)$ and $\bw_{0,1}(\lambda_1-\lambda_2)$ respectively while the slope is of the order $\Omega(\lambda_1-\lambda_2)$ in Figure~\ref{fig:continuous 2 sided}.
As the slope corresponds to the \textit{speed} of the convergence, the extra $\epsilon$ or $\bw_{0,1}$ in the slope of linearization at a single fixed point would result in an extra $\epsilon^{-1}$ or $\bw_{0,1}^{-1}$ in the convergence rate. See~\autoref{fig:continuous} for a pictorial explanation.\\

Another key inspiration from the continuous dynamic is the \textit{ODE trick} which provides a close form characterization of the dynamic in terms of the drifting term captured by the continuous dynamic and the noise term originated from the linearization and discretization.
The ODE trick is inspired by the solution to a linear ordinary differential equation (linear ODE). Consider the following simple linear ODE
\[
\frac{d y(t)}{dt} = ay(t)+b(t)
\]
for some constant $a$ and function $b(t)$. To put into the context, one can think of $a$ as the drifting term and $b(t)$ as the noise term in the continuous Oja's rule due to the linearization\footnote{In the biological Oja's rule, the \textit{discretization} also contributes in the noise term.}. By the standard tool for solving linear ODE, the solution of $y(t)$ at $t=T$ is
\begin{equation}\label{eq:intro ODE trick}
y(T) = e^{aT}\cdot\left(y(0)+\int_0^\top e^{-at}b(t)dt\right) \, .
\end{equation}
From the above equation, one can see that the solution of a linear ODE extracts the drifting term into a \textit{multiplier} $e^{aT}$ and decouples the initial condition $y(0)$ with the noise term $\int_{0}^\top e^{-at}b(t)dt$. As a consequence, once we can show that the noise term is much smaller than the initial value, then $y(T)$ is dominated by the drifting term $e^{aT}y(0)$ and thus we are able to analyze the progress of $y(T)$.

To sum up, the continuous dynamic informs us to linearize the biological Oja's rule at different centers in different phases of the analysis. Further, the ODE trick provides us a closed-form approximation to the dynamic.
We are then able to analyze the biological Oja's rule in \textit{one-shot} rather than doing the traditional step-by-step analysis.

\paragraph{Insight 2: One-shot analysis instead of step-by-step analysis.}
The second insight of this work is performing an \textit{one-shot analysis} instead of the traditional step-by-step analysis (\textit{e.g.,}~\cite{AL17}). 
\subparagraph{Traditional step-by-step analysis}
To see the difference, let us illustrate how would the step-by-step analysis on the biological Oja's rule work as follows. Denote the natural filtration as $\{\mathcal{F}_t\}$ where $\mathcal{F}_t$ is the $\sigma$-algebra generated by $\bx_1,\bx_2,\dots,\bx_t$. For any $t\in\N$, we have
\begin{align*}
\Exp\left[\bw_{t,1}\right] &= \Exp\left[\Exp\left[\bw_{t-1,1}+\eta_t(\bx_t^\top\bw_{t-1})\bx_{t,1}-\eta_t(\bx_t^\top\bw_{t-1})^2\bw_{t-1,1}\ |\ \mathcal{F}_{t-1}\right]\right]\\
&=\Exp\left[\bw_{t-1,1}+\eta_t\lambda_1\bw_{t-1,1}-\eta_t\left(\sum_{i=1}^n\lambda_i\bw_{t-1,i}^2\right)\bw_{t-1,1}\right]
\end{align*}
where the second equation is due to the fact that for any $i,j\in[n]$, $\Exp[\bx_{t,i}\bx_{t,j}\ |\ \mathcal{F}_{t-1}]=A_{ij}=\lambda_i\cdot\mathbf{1}_{i=j}$ and for any $i\in[n]$, $\Exp[\bw_{t-1,i}\ |\ \mathcal{F}_{t-1}]=\bw_{t-1,i}$. In a step-by-step analysis, one then argues that the expectation $\Exp[\bw_{t,1}]$ would be improved from $\Exp[\bw_{t-1,1}]$ by a certain factor. Then, an induction on each step followed by showing concentration would give some convergence rate guarantee.
However, there are two difficulties in getting optimal convergence rate
 (these difficulties usually also appear in the step-by-step analysis for other problems). 
\begin{itemize}
\item First, there are some non-linear terms of $\bw_{t-1,1}$ in the update noise. This usually requires some hacks tailored to the specific problem to enable the analysis.
\item Second, the improvement factor at each step can depend on $\bw_{t-1}$ and at worst case, the dynamic can show no improvement or even deteriorate. Taking expectation loses precise controls of the values of $\bw_{t-1}$. This makes naive martingale analysis difficult to work and probably requires more ad hoc tricks.
\end{itemize}
For instance, the first difficulty is exactly what~\cite{AL17} encountered in their analysis for a variant of the biological Oja's rule. They resolved the first difficulty by decomposing the non-linear term through careful variable substitution, but as a result they incur unnecessary logrithmic costs. The biological Oja's rule, in addition to having the first difficulty, also has the second difficulty (see~\autoref{sec:compare bio ML} for more discussions). Therefore, applying the traditional step-by-step analysis on the biological Oja's rule will encounter great obstacles.

\subparagraph{Our one-shot analysis}
In this work, we use an \textit{one-shot} analysis to avoid the complication of a step-by-step analysis. Namely, instead of looking at the process iteratively, we study the entire dynamic at once. Two key ingredients are needed to implement such an one-shot analysis: (i) a closed-form characterization of the dynamic and (ii) stopping time techniques. As discussed in the previous discussion, the continuous dynamic of the biological Oja's rule inspires us to get a closed-form lower bound for $\bw_{t,1}$ by the \textit{ODE trick}. Concretely, as a simplified example\footnote{In general, the multiplier term also varies with respect to time $t$.}, we have
\begin{equation}\label{eq:intro one shot}
\bw_{T,1} = H^\top \cdot\left(\bw_{0,1} + \sum_{t=1}^{T}\frac{N_{t}}{H^{t}}\right)
\end{equation}
where $H>1$ is the multiplier term and $\{N_{t}\}$ is the noise term which forms a martingale on the natural filtration. See~\autoref{cor:discrete phase 1 ODE} and~\autoref{cor:discrete phase 2 ODE} for a precise formulation of $H$ and $\{N_t\}$ in our analysis. Intuitively, one should think of $H^\top \bw_{0,1}$ as the \textit{drifting term} and the other part as the \textit{noise term}. The goal of the ODE trick in the discrete dynamic is to show that the drifting term dominates the noise term.

To show that the noise in~\autoref{eq:intro one shot} is small, Azuma's inequality (see~\autoref{lem:azuma}) would be a natural tool to start with. However, the \textit{bounded difference} condition in Azuma's inequality would immediately cause an issue: the noise at time $t$ is correlated with $\bw_{t-1,1}$ and thus one cannot get a small bounded difference almost surely. For example, suppose the bounded difference of $\{N_t\}$ at time $t$ is at most $\bw_{t-1,1}^2$. Since we do not yet know the behavior of $\bw_{t-1,1}$, we can only upper bound the bounded difference of $\{N_t\}$ in the worst case\footnote{This is because we are able to upper bound $\bw_{t-1,1}$ by $1+o(1)$ almost surely. See~\autoref{sec:oja bounded}. Note that there are ways to get better bounded difference condition in the worst case but this is still not sufficient.} by $1+o(1)$. In the meantime, both $\bw_{t,1}^2$ and the noise are expected to be very small in the early stage of the dynamic with high probability. 

To circumvent this obstacle, we consider the \textit{stopped process} of the original martingale in which the bounded difference is under control. 
For example, consider the above situation where the noise term $\{N_t\}$ is a martingale and a stopping time $\tau$ for the event $\{\bw_{\tau,1}^2\geq0.1\}$. The stopped process, denoted by $\{N_{t\wedge\tau}\}$ where $t\wedge\tau=\min\{t,\tau\}$, is a process that simulates $\{N_t\}$ and \textit{stops} at the first time $t^*$ such that $\bw_{t^*,1}^2\geq0.1$.
It is known that a stopped process of a martingale is also a martingale. 
Furthermore, the bounded difference of the stopped process $\{N_{t\wedge\tau}\}$ would be $0.1$ almost surely by the choice of $\tau$.
It turns out that this improvement in the bounded difference condition drastically increases the quality of Azuma's inequality and gives the desiring concentration for the stopped process.

There is one last missing step before showing the dominance of $\bw_{0,1}$ in~\autoref{eq:intro one shot}: we have to show that the concentration for the stopped process $\{N_{t\wedge\tau}\}$ can be extended to the original process $\{N_{t}\}$.
We achieve this task by developing a \textit{pull-out lemma} which is able to utilize the structure of the martingale and pull out the stopping time from a concentration inequality.

\paragraph{Insight 3: Maximal martingale inequality and pull-out lemma.}

In general, there is no hope to pull out the stopping time from a concentration inequality for the stopped process without blowing up the failure probability.
The naive union bound would give a blow-up of factor $T$ in the failure probability and it is undesirable.

Let $M_t=\sum_{t'=1}^\top H^{-t'}N_{t'}$ be the noise term in the ODE trick (\textit{i.e.,}~\autoref{eq:intro one shot}) and $\tau$ be a stopping time that ensures good bounded difference condition. Note that as $\{N_t\}$ is a martingale, we know that $\{M_{t\wedge\tau}\}$ is also a martingale.
There are two key ingredients to pull out the stopping time from $\{M_{t\wedge\tau}\}$, \textit{i.e.,} the stopped process of the noise term.

First, we use the \textit{maximal} concentration inequality (\textit{e.g.,}~\autoref{lem:maximal azuma}) which gives the following stronger guarantee than the traditional Azuma's inequality.
\begin{equation}\label{eq:intro pull out 1}
\Pr\left[\sup_{1\leq t\leq T}|M_{t\wedge\tau}-M_0|\geq a\right]<\delta
\end{equation}
for some $a>0$, $T\in\N$, and $\delta\in(0,1)$. Note that the maximal concentration inequality gives concentration for any $1\leq t\leq T$ without paying an union bound.

Second, we identify a \textit{chain structure} on the martingale and the stopping time $\tau$ we are working with. 
Concretely, we are able to show that for all $t\in[T]$,
\begin{equation}\label{eq:intro pull out 2}
\Pr\left[\tau\geq t+1\ \Big|\ \sup_{1\leq t'\leq t}|M_{t'}-M_0|<a\right] = 1 \, .
\end{equation}
Namely, if the bad event has not happened, then the martingale would not stop immediately. Intuitively,~\autoref{eq:intro pull out 2} holds due to the ODE trick because $\{\sup_{1\leq t'\leq t}|M_{t'}-M_0|<a\}$ implies the noise term to be small and thus the drifting term dominates. As $\tau$ is properly chosen such that the martingale would not stop if the process $\bw_t$ followed the drifting term, we know that $\tau\geq t+1$.

Combining the above two ingredients (\textit{i.e.,}~\autoref{eq:intro pull out 1} and~\autoref{eq:intro pull out 2}), we are able to show in the pull-out lemma that
\[
\Pr\left[\sup_{1\leq t\leq T}|M_{t}-M_0|\geq a\right]<\delta \, ,
\]
\textit{i.e.,} the stopping time has been \textit{pulled out}.\\

Let us end this subsection with a high-level sketch on the proof for the pull-out lemma. The key idea is to consider another stopping time $\tau'$ for the event $\{|M_{\tau'}-M_0|\geq a\}$ and partition the probability space of the error event $\{\sup_{1\leq t\leq T}|M_{t}-M_0|\geq a\}$ in to two parts $P_1$ and $P_2$ with the following properties. In $P_1$, we can show that 
\begin{equation*}
\Pr\left[\sup_{1\leq t\leq T}|M_{t}-M_0|\geq a,\ P_1\right]=\Pr\left[\sup_{1\leq t\leq T}|M_{t\wedge\tau}-M_0|\geq a,\ P_1\right] \, .
\end{equation*}
As for $P_2$, we use the chain condition in~\autoref{eq:intro pull out 2} to show that the probability of error event is 0 based on a  \textit{diagonal argument}. Thus, we have
\begin{align*}
\Pr\left[\sup_{1\leq t\leq T}|M_{t}-M_0|\geq a\right]&=\Pr\left[\sup_{1\leq t\leq T}|M_{t}-M_0|\geq a,\ P_1\right]+\Pr\left[\sup_{1\leq t\leq T}|M_{t}-M_0|\geq a,\ P_2\right]\\
&=\Pr\left[\sup_{1\leq t\leq T}|M_{t\wedge\tau}-M_0|\geq a,\ P_1\right]+0\\
&\leq\Pr\left[\sup_{1\leq t\leq T}|M_{t\wedge\tau}-M_0|\geq a\right]<\delta \, .
\end{align*}
See~\autoref{sec: local convergence} and~\autoref{fig:intuition pull out} for more details on the chain condition for biological Oja's rule and how to partition the probability space of the error event.

\subsection{Related works}\label{sec:related works}

\paragraph{Related theory work on biological Oja's variants.}
Computational neuroscientists have proposed several variants of the biological Oja's rule to solve streaming PCA~\cite{Oja82, Oja92,Sanger89,Foldiak1989,Leen1991,Rubner1989,KDT94,PHC15}. In the single neuron case, Oja used stochastic approximation theory~\cite{Kushner1978} to prove the global convergence in the limit~\cite{Oja82}. In the multi-neurons case, Hornik and Kuan similarly demonstrated the connection between the discrete dynamics and the associated ODE~\cite{Hornik1992} from Kushner-Clark theorem~\cite{Kushner1978}. However, most existing analysis on the multi-neurons case shows only local convergence~\cite{Sanger89,Foldiak1989,Leen1991,Rubner1989,KDT94,PHC15}. Even for the convergence in the limit, the global convergence for most networks in the multi-neurons case is difficult to show. Yan et al. provided the only global analysis on Oja's multi-neurons subspace network~\cite{Oja92,Yan1994,Yan1998}. Previously there is no work showing the convergence rate on the discrete dynamics. This paper shows the first convergence rate bound on the biological Oja's rule.

\paragraph{Oja's rule in machine learning.}
Unlike the situation in the biological Oja's rule, a line of recent exciting results~\cite{HP14,DOR15,BDWY16,Shamir16,JJKNS16,AL17} showed convergence rate analysis for variants of Oja's rule in the machine learning community. Since the update rules of these works are not biologically-plausible, we call them \textit{ML Oja's rules} to distinguish from the biological Oja's rule.

To see the difference between the biological Oja's rule and the ML Oja's rule, let us take the update rule from~\cite{Shamir16,JJKNS16,AL17} as an example. Note that the other variants of ML Oja's rule also have the similar fundamental difference to the biological Oja's rule as illustrated by the following example. Let $\bw_t\in\Real^n$ be the output vector at time $t=0,1,\dots,T$, the update rule is
\[
\bw_t = \prod_{t'=1}^{t}\left(1+\eta_{t'}\bx_{t'}\bx_{t'}^\top\right)\bw_0
\]
and the output is $\bw_T/\|\bw_T\|_2$. Note that the above update rule is equivalent to~\autoref{eq:oja normalized}, \textit{i.e.,} applying Taylor's expansion on the ML Oja's rule and truncating the higher-order terms would result in biological Oja's rule.

A natural idea would be trying to \textit{couple} the biological Oja's rule with the ML Oja's rule by showing that for all $t\in\N$, the weight vectors from the two dynamics would be close to each other. However, this seems to be more difficult than direct analysis and we leave it as an interesting open problem to investigate whether this is the case.
Moreover, the corresponding continuous dynamics suggest an intrinsic difference between the two: the continuous version of the ML Oja's rule can be tightly characterized by a single linear ODE while that of the biological Oja's rule requires two linear ODEs in different regimes for tight analysis. See~\autoref{sec:continuous Oja} and~\autoref{sec:compare bio ML} for more details. 

To sum up, the biological Oja's rule and the ML Oja's rule are similar but the analysis of the later cannot be directly applied to the former. While following the proof idea for the ML Oja's rule might give some hints on how to analyze the biological Oja's rule, in this work we develop a completely different framework (as briefly discussed in~\autoref{sec:technical overview}). This framework not only gives the first and nearly optimal convergence rate guarantee for the biological Oja's rule, but also could improve the convergence rate of the ML Oja's rule with better logarithmic dependencies and we leave it as a future work.

\paragraph{Comparing with other streaming PCA algorithms.}
Streaming PCA is a well-studied and challenging computational problem. Many works~\cite{DOR15,Shamir16,LWLZ18,JJKNS16,AL17} provided theoretical guarantees for streaming PCA algorithms. Interestingly, all of the streaming PCA algorithms in these works are some variants of the biological Oja's rule.

Recall that there are two standard convergence notions: the global convergence where $\bw_0$ is an uniformly random unit vector and the local convergence where $\bw_0$ is constantly correlated with the top eigenvector. There are 5 parameters of interest: the dimension $n\in\N$, the eigenvalue gap $\gap:=\lambda_1-\lambda_2\in(0,1)$, the top eigenvalue $\lambda_1\in(0,1)$, the error parameter $\epsilon\in(0,1)$, and the failure probability $\delta\in(0,1)$. Ideally, the goal is to achieve the information-theoretic lower bound $\Omega(\lambda_1\gap^{-2}\epsilon^{-1}\log(\delta^{-1}))$ given by~\cite{AL17}. Prior to this work, the state-of-the-art for both global and local convergences are achieved by~\cite{AL17} using ML Oja's rule (see the second to last row of~\autoref{table:compare}). In this work, as a byproduct, the convergence rate we get for the biological Oja's rule outperforms~\cite{AL17} by a logarithmic factor in both settings. 
See~\autoref{table:compare} for a summary.

\begin{table}[ht]
\centering
\def\arraystretch{2}
\begin{tabular}{|c|c|c|l|c|l|c|}
\hline
\multirow{2}{*}{\small{Algorithm}}           & \multirow{2}{*}{Reference}   & \multirow{2}{*}{$\substack{\text{Any}\\\text{Input}}$}  &  \multicolumn{2}{|c|}{Global Convergence} & \multicolumn{2}{|c|}{Local Convergence}                             \\ 
\cline{4-7} & & & \multicolumn{1}{|c|}{$\substack{\text{Convergence}\\\text{Rate}}$} & $\substack{\text{Degree in}\\\text{Log Terms}\footref{foot:log}}$ & \multicolumn{1}{|c|}{$\substack{\text{Convergence}\\\text{Rate}}$} & $\substack{\text{Degree in}\\\text{Log Terms}\footref{foot:log}}$ \\
\hline
$\substack{\text{Biological}\\\text{Oja's Rule}}$ &  This Work  & {\color{red}Y}  &  $\tilde{O}\left(\frac{\lambda_1}{\gap^2}\cdot\frac{1}{\epsilon\wedge\delta^2}\right)$ & $3$ & $\tilde{O}\left(\frac{\lambda_1}{\gap^2}\cdot\frac{1}{\epsilon}\right)$ & $2$ \\ \hline
\multirow{5}{*}{$\substack{\text{ML}\\\text{Oja's Rule}}$} & \cite{DOR15} & N  &  $\tilde{O}\left(\frac{n}{\gap^2}\cdot\frac{1}{\epsilon}\right)$~\footref{foot:delta} & -~\footref{foot:log degree} & $\tilde{O}\left(\frac{n}{\gap^2}\cdot\frac{1}{\epsilon}\right)$~\footref{foot:delta} & -~\footref{foot:log degree} \\ 
\cline{2-7} & \cite{Shamir16} & {\color{red}Y} &  $\tilde{O}\left(\frac{n}{\gap^2}\cdot\frac{1}{\epsilon}\right)$~\footref{foot:delta} & -~\footref{foot:log degree} & $\tilde{O}\left(\frac{n}{\gap^2}\cdot\frac{1}{\epsilon}\right)$~\footref{foot:delta} & -~\footref{foot:log degree} \\
\cline{2-7} & \cite{LWLZ18} & N &  $\tilde{O}\left(\frac{\lambda_1n}{\gap^2}\cdot\frac{1}{\epsilon\delta^6}\right)$ & -~\footref{foot:log degree} & $\tilde{O}\left(\frac{\lambda_1n}{\gap^2}\cdot\frac{1}{\epsilon\delta^4}\right)$ & -~\footref{foot:log degree} \\
\cline{2-7} & \cite{JJKNS16} & {\color{red}Y} &  $\tilde{O}\left(\frac{\lambda_1}{\gap^2}\cdot\frac{1}{\epsilon\delta^3}\right)$ & $2$ &     $\tilde{O}\left(\frac{\lambda_1}{\gap^2}\cdot\frac{1}{\epsilon\delta^3}\right)$ & $2$  \\
\cline{2-7}   & \cite{AL17} & {\color{red}Y} &  $\tilde{O}\left(\frac{\lambda_1}{\gap^2}\cdot\frac{1}{\epsilon\wedge\delta^2}\right)$ & $\geq 4$ & $\tilde{O}\left(\frac{\lambda_1}{\gap^2}\cdot\frac{1}{\epsilon}\right)$ & $\geq3$ \\ \hline
$\substack{\text{Any}\\\text{Algorithm}}$    & \cite{AL17}         & \multicolumn{5}{|c|}{$\Omega\left(\frac{\lambda_1}{\gap^2}\cdot\frac{\log\frac{1}{\delta}}{\epsilon}\right)$~\footref{foot:lower bound}\ \ \ \ \ \ \ }               \\ \hline
\end{tabular}
\caption{Convergence rate for biological Oja's rule and ML Oja's rule in solving streaming PCA. The ``Any Input'' column indicates that whether the analysis has higher moment conditions on the unknown distribution $\mathcal{D}$. Note that having higher moment conditions would drastically simplify the problem because the non-linear terms in the update rule can then be non-trivially replaced with the first order term.}
\label{table:compare}
\end{table}

\renewcommand{\thefootnote}{\fnsymbol{footnote}}
\footnotetext[1]{\label{foot:log}Let $f(\log n, \log(1/\epsilon),\log(1/\delta),\log(1/\gap))$ be the polynomial of the logarithmic dependencies in the convergence rate. We compare the maximum degree of $f$ among different analyses. Note that this measure makes sense when $n,1/\epsilon,1/\delta,1/\gap$ are polynomially related.}
\footnotetext[2]{\label{foot:delta}Both~\cite{DOR15} and~\cite{Shamir16} cannot handle arbitrary failure probability so we ignore their $\delta$ dependency on the table.}
\footnotetext[3]{\label{foot:log degree}In~\cite{DOR15,Shamir16,LWLZ18}, their convergence rates are far from the information-theoretic lower bound. So we do not trace down their logarithmic dependencies.}
\footnotetext[4]{\label{foot:lower bound}In~\cite{AL17}, they only stated $\Omega(\frac{\lambda_1}{\gap}\cdot\frac{1}{\epsilon})$ lower bound. We observe that their lower bound can be improved by a $\log(1/\delta)$ factor using the fact that distinguishing a fair coin from a biased coin with probability at least $\delta$ requires $\Omega(\log(1/\delta))$ samples.}

\renewcommand{\thefootnote}{\arabic{footnote}}

\paragraph{Algorithms inspired by biological neural networks.}
In recent years, the study of the algorithmic aspect of mathematical models for biological neural networks is an emerging field in theoretical CS.
For example, the efficiency of spiking neural networks in solving the \textit{winner-take-all (WTA)} problem~\cite{LMP17BDA,LMP17ITCS,LMP17DISC,LM18,SCL19}, the efficiency of spiking neural networks in storing temporal information~\cite{WL19,Merav2019}, assemblies~\cite{LMPV18,PV19}, spiking neural networks in solving optimization problems~\cite{CCL19,Pehlevan19} and learning hierarchically structured concepts~\cite{LM19}. Under this context, this work provides an algorithmic insight in a plasticity learning rule that solves streaming PCA.

\section{Preliminaries}\label{sec:prelim}
In this section, we introduce the mathematical notations and tools that we use in this work.
\subsection{Notations}
We use $\N=\{1,2,\dots\}$ and $\N_{\geq0}=\{0,1,\dots\}$. For each $n\in\N$, denote $[n]=\{1,2,\dots,n\}$ and $[n]_{\geq0}=\{0,1,\dots,n\}$. For a vector indexed by time $t$, \textit{e.g.,} $\bw_t$, its $i^\text{th}$ coordinate is denoted by $\bw_{t,i}$. The notation $\tilde{O}$ (similarly, $\tilde{\Omega}$ and $\tilde{\Theta}$) is the same as the big-O notation by ignoring extra poly-logarithmic term. 
$\mathbf{1}_E$ stands for the indicator function for a probability event $E$. We sometimes abuse the big O notation by $y=O(x)$ meaning $|y|=O(x)$ and this will be clear in the context. Throughout the paper, $\lambda$ is used to denote the vector $(\lambda_1,\lambda_2,\dots,\lambda_n)$ where $\lambda_1>\lambda_2\geq\cdots\geq\lambda_n\geq0$ are the eigenvalues of the covariance matrix $A$. $\diag(\lambda)$ denotes the diagonal matrix with $\lambda$ on the diagonal. We will follow the convention of stochastic process and denote $\min\lbrace a, b\rbrace$ as $a\wedge b$. We say an event happens {\it almost surely} if it happens with probability one.

\subsection{Probability toolbox}
\paragraph{Random process and concentration inequality.}
Random process is a central tool in this paper. Let us start with the most general definition on adapted random process.

\begin{definition}[Adapted random process]
Let $\{X_t\}_{t\in\Nz}$ be a sequence of random variables and $\{\mathcal{F}_t\}_{t\in\Nz}$ be a filtration. We say $\{X_t\}_{t\in\Nz}$ is an adapted random process with respect to $\{\mathcal{F}_t\}_{t\in\Nz}$ if for each $t\in\Nz$, the $\sigma$-algebra generated by $X_0,X_1,\dots,X_t$ is contained in $\mathcal{F}_t$.
\end{definition}

In most of the situation, we use $\mathcal{F}_t$ to denote the \textit{natural filtration} of $\{X_t\}_{t\in\Nz}$, namely, $\mathcal{F}_t$ is defined as the $\sigma$-algebra generated by $X_0,X_1,\dots,X_t$. One of the most common adapted processes is the martingale.

\begin{definition}[Martingale]
Let $\{M_t\}_{t\in\Nz}$ be a sequence of random variables and let $\{\mathcal{F}_t\}_{\N}$ be its natural filtration. We say $\{M_t\}_{t\in\Nz}$ is a martingale if for each $t\in\N$, $\Exp[M_{t+1}\ |\ \mathcal{F}_t]=M_t$.
\end{definition}

Note that for any adapted random process $\{X_t\}_{t\in\Nz}$, one can always turn it into a martingale by defining $M_0=X_0$ and for any $t\in\N$, let $M_t=X_t-\Exp[X_t\ |\ \mathcal{F}_{t-1}]$.
When the difference of a martingale can be bounded almost surely, \textit{the Azuma's inequality} provides an useful concentration inequality with exponential tail.

\begin{lemma}[Azuma's inequality~\cite{Azuma67}]\label{lem:azuma}
Let $\{M_t\}_{t\in\Nz}$ be a martingale. Let $T\in\N$ and $a,c\geq0$ be some constants. Suppose for each $t=1,2,\dots,T$, $|M_{t}-M_{t-1}|\leq c$ almost surely, then we have
\[
\Pr\left[|M_T-M_0|\geq a\right]< \exp\left(-\Omega\left(\frac{a^2}{c^2T}\right)\right) \, .
\]
\end{lemma}

The following maximal Azuma's inequality shows that one can even get union bound for free with the help of Doob's inequality.

\begin{lemma}[Maximal Azuma's inequality~{\cite[Section~3]{HMRR13}}]\label{lem:maximal azuma}
Let $\{M_t\}_{t\in\Nz}$ be a martingale. Let $T\in\N$ and $a,c\geq0$ be some constants. Suppose for each $t=1,2,\dots,T$, $|M_{t}-M_{t-1}|\leq c$ almost surely, then we have
\[
\Pr\left[\sup_{0\leq t\leq T}|M_t-M_0|\geq a\right]< \exp\left(-\Omega\left(\frac{a^2}{c^2T}\right)\right) \, .
\]
\end{lemma}

The Azuma's inequality can be strengthen by considering the conditional variance. This is known as the Freedman's inequality.

\begin{lemma}[Freedman's inequality~{\cite{F75}}]\label{lem:freedman}
Let $\{M_t\}_{t\in\Nz}$ be a martingale. Let $T\in\N$ and $a,c,\sigma_t\geq0$ be some constants for all $t\in[T]$. Suppose for each $t=1,2,\dots,T$, $|M_{t}-M_{t-1}|\leq c$ almost surely and $\Var[M_{t} - M_{t-1}\ |\ \mathcal{F}_{t-1}]\leq\sigma_t^2$, then we have
\[
\Pr\left[\sup_{0\leq t\leq T}|M_t-M_0|\geq a\right]< \exp\left(-\Omega\left(\frac{a^2}{\sum_{t=1}^{T}\sigma_t^2+ca}\right)\right) \, .
\]
\end{lemma}

Finally, we state a corollary of Freedman's inequality for adapted random process with small conditional expectation.
\begin{corollary}\label{cor:freedman}
Let $\{M_t\}_{t\in\Nz}$ be a random process. Let $T\in\N$ and $a,c,\sigma_t,\mu_t\geq0$ be some constants for all $t\in[T]$. Suppose for each $t=1,2,\dots,T$, $|M_{t}-M_{t-1}|\leq c$ almost surely, $\Var[M_{t} - M_{t-1}\ |\ \mathcal{F}_{t-1}]\leq\sigma_t^2$, and $|\Exp[M_{t}-M_{t-1}\ |\ \mathcal{F}_{t-1}]|\leq\mu_t$, then we have
\[
\Pr\left[\sup_{0\leq t\leq T}|M_t-M_0|\geq a+\max_{1\leq t\leq T}\sum_{t=1}^T \mu_{t}\right]< \exp\left(-\Omega\left(\frac{a^2}{\sum_{t=1}^{T}\sigma_t^2+ca}\right)\right) \, .
\]
\end{corollary}

\paragraph{Stopping time.}
One powerful technique for studying martingale is the notion of \textit{stopping time} defined as follows.
\begin{definition}[Stopping time]\label{def:stopping time}
Let $\{X_t\}_{t\in\Nz}$ be an adapted random process associated with filtration $\{\mathcal{F}_t\}_{t\in\Nz}$. An integer-valued random variable $\tau$ is a stopping time for $\{X_t\}_{t\in\Nz}$ if for all $t\in\N$, $\{\tau=t\}\in\mathcal{F}_t$.
\end{definition}

Let $\{M_t\}_{t\in\Nz}$ be a martingale, the most common stopping time for $\{M_t\}_{t\in\Nz}$ is of the following form. For any $a\in\Real$, let
\[
\tau := \min_{M_t\geq a}t \, .
\]
Namely, $\tau$ is the first time when the martingale becomes at least $a$. For convenience, in the rest of the paper, we would define stopping time of this form by saying ``$\tau$ is the stopping time for the event $\{M_t\geq a\}$''.

Given a martingale $\{M_t\}_{t\in\Nz}$ and a stopping time $\tau$, it is then natural to consider the corresponding \textit{stopped process} $\{M_{t\wedge\tau}\}_{t\in\Nz}$ where $t\wedge\tau=\min\{t,\tau\}$ is also a random variable. An useful and powerful fact here is that the stopped process of a martingale is also a martingale. See~\cite[Theorem~10.9]{W91} for a proof for this classic result. 

We have the following identity for the stopped process.
\begin{lemma}[The difference of a stopped process]\label{lem: stop difference}
Given a stochastic process $M_t$ and a stopping time $\tau$. We have
	\[
	M_{t\wedge\tau} - M_{(t-1)\wedge\tau} =  \mathbf{1}_{\tau \geq t}(M_t - M_{t-1}) \, .
	\]
\end{lemma}
\begin{proof}
We have
\begin{align*}
M_{t\wedge\tau} - M_{(t-1)\wedge\tau}&= \mathbf{1}_{\tau \geq t}M_t + \mathbf{1}_{\tau < t}M_{\tau} - \mathbf{1}_{\tau \geq t-1}M_{t-1} - \mathbf{1}_{\tau < t-1}M_{\tau}\\
&= \mathbf{1}_{\tau \geq t}M_t - \mathbf{1}_{\tau \geq t-1}M_{t-1} +  \mathbf{1}_{\tau = t-1}M_{\tau} \, . \\
\intertext{Since $\tau = t-1$ at the last term, we can combine the last two terms to have }
&=\mathbf{1}_{\tau \geq t}M_t - \mathbf{1}_{\tau \geq t}M_{t-1} \\
&= \mathbf{1}_{\tau \geq t}(M_t - M_{t-1})
\end{align*}
as desired. 
\end{proof}

Here we also define a \textit{shifted stopping process} which is an useful variant used in more complicated situations.
\begin{restatable}[The shifted stopped process]{definition}{starstopping}\label{def:star stopped}
Given an adapted stochastic process $M_t$ with respect to filtration $\mathcal{F}_t$ and a stopping time $\tau$, we define a new adapted process $M_{t\star \tau}$ with respect to $\mathcal{F}_t$ to be 
\[
M_{t\star \tau} = \mathbf{1}_{\tau > t}M_t + \mathbf{1}_{\tau \leq t}M_{\tau - 1} \, .
\]
Given $t\in \mathbb{N}$, we define a random variable $t\star \tau$ as 
\[
t\star \tau = \mathbf{1}_{\tau > t}t + \mathbf{1}_{\tau \leq t}(\tau - 1) \, .
\]
\end{restatable}
Intuitively, a shifted stopped process is the original process which moves one step back if the stopping time stops.

\begin{lemma}[The difference of a shifted stopped process]\label{lem: star difference}
Given a stochastic process $M_t$ and a stopping time $\tau$. We have
	\[
	M_{t\star\tau} - M_{(t-1)\star\tau} =  \mathbf{1}_{\tau > t}(M_t - M_{t-1}) \, .
	\]
\end{lemma}
\begin{proof}
We have
\begin{align*}
M_{t\star\tau} - M_{(t-1)\star\tau}&= \mathbf{1}_{\tau > t}M_t + \mathbf{1}_{\tau \leq t}M_{\tau - 1} - \mathbf{1}_{\tau > t-1}M_{t-1} - \mathbf{1}_{\tau \leq t-1}M_{\tau - 1}\\
&= \mathbf{1}_{\tau > t}M_t - \mathbf{1}_{\tau > t-1}M_{t-1} +  \mathbf{1}_{\tau = t}M_{\tau - 1} \, . \\
\intertext{Since $\tau = t$ at the last term, we can combine the last two terms to have }
&=\mathbf{1}_{\tau > t}M_t - \mathbf{1}_{\tau > t}M_{t-1} \\
&= \mathbf{1}_{\tau > t}(M_t - M_{t-1})
\end{align*}
as desired. 
\end{proof}

\paragraph{Brownian motion.}

In~\autoref{sec:continuous Oja}, we consider a continuous version of biological Oja's rule by modeling the input stream as a Brownian motion. Here, we provide background that is sufficient for the readers to understand the discussion there.

First, we introduce the 1-dimensional Brownian motion using the following operational definition. In the following, we use $N(\mu,\sigma^2)$ to denote the Gaussian distribution with mean $\mu$ and variance $\sigma^2$.
\begin{definition}[1-dimensional Brownian motion]\label{def:1 dim BM}
Let $\{\beta_t\}_{t\geq0}$ be a real-valued random process. We say $\{\beta_t\}_{t\geq0}$ is a 1-dimensional Brownian motion if the following holds.
\begin{itemize}
\item $\beta_0 = 0$ and $\beta_t$ is almost surely continuous.
\item For any $t_1, t_2, t_3, t_4$ such that $0\leq t_1 < t_2\leq t_3 < t_4$, $\beta_{t_{2}}-\beta_{t_{1}}$ is independent from $\beta_{t_{4}}-\beta_{t_{3}}$.
\item For any $t_1, t_2$ such that $0\leq t_1<t_2$, $\beta_{t_2}-\beta_{t_1}\sim N(0,t_2-t_1)$.
\end{itemize}
\end{definition}

With the above definition, it is then natural to consider some variants such as putting $n$ independent copies of 1-dimensional Brownian motion into a vector, \textit{i.e.,} the $n$-dimensional Brownian motion, or applying linear operations on an $n$-dimensional Brownian motion, or considering the \textit{calculus} on Brownian motion by looking at $d\beta_t=\lim_{\Delta\rightarrow0}\beta_{t+\Delta}-\beta_t$. The role of Brownian motion in the study of continuous random process is similar to Gaussian random variance in discrete random process and many properties in the discrete world directly extend to the continuous world. One property of Brownian motion though obviously does not hold in the discrete setting and might be counter-intuitive for people who see this for the first time. This is the \textit{quadratic variation} of Brownian motion as stated below.
\begin{lemma}[Quadratic variation of Brownian motion]\label{lem:quadratic variation}
Let $\{\beta_t\}_{t\geq0}$ and $\{\beta'_t\}_{t\geq0}$ be two independent 1-dimensional Brownian motions. The following holds almost surely.
\[
d\beta_t^2=dt\ \ \text{and}\ \ d\beta_td\beta'_t=0 \, .
\]
\end{lemma}
We omit the proof of~\autoref{lem:quadratic variation} here and refer the interested readers to standard textbook such as~\cite{L16} for more details on Brownian motion.

\subsection{ODE toolbox}

\begin{lemma}[ODE trick for scalar]\label{lem:ODE trick}
Let $\{X_t\}_{t\geq\Nz}$, $\{A_t\}_{t\in\N}$, and $\{H_t\}_{t\in\N}$ be sequences of random variables with the following dynamic
\begin{equation}\label{eq:ODE trick 1}
X_{t} = H_tX_{t-1}+A_t
\end{equation}
for all $t\in\N$. Then for all $t_0,t\in\Nz$ such that $t_0<t$, we have
\[
X_t = \prod_{i=t_0+1}^T H_i\cdot\left(X_{t_0}+\sum_{i=t_0+1}^T \frac{A_i}{\prod_{j=t_0+1}^iH_{j}}\right) \, .
\]
\end{lemma}
\begin{proof}[Proof of~\autoref{lem:ODE trick}]
For each $t_0<i\leq t$, dividing~\autoref{eq:ODE trick 1} with $\prod_{j=t_0+1}^iH_{j}$ on both sides, we have
\[
\frac{X_i}{\prod_{j=t_0+1}^iH_{j}} = \frac{X_{i-1}}{\prod_{j=t_0+1}^{i-1}H_{j}} + \frac{A_i}{\prod_{j=t_0+1}^iH_{j}} \, .
\]
By telescoping the above equation from $t=t_0+1$ to $t$, we get the desiring expression.
\end{proof}

\begin{lemma}[ODE trick for vector]\label{lem:ODE trick matrix}
Let $\{X_t\}_{t\in\Nz}$, $\{A_t\}_{t\in\N}$ be sequences of $m_t$-dimensional random variables and $\{H_t\}_{t\in\N}$ be a sequence of random $m_t\times m_{t-1}$ matrices with the following dynamic
\begin{equation}\label{eq:ODE trick 1 matrix}
X_{t} = H_tX_{t-1}+A_t
\end{equation}
for all $t\in\N$. Then for all $t_0,t\in\Nz$ such that $t_0<t$, we have
\[
X_t = \prod_{i=t_0+1}^T H_iX_{t_0}+\sum_{i=t_0+1}^T \prod_{j=i+1}^T H_{t-j}A_i \, .
\]
\end{lemma}
\begin{proof}[Proof of~\autoref{lem:ODE trick matrix}]
The proof is a direct induction.
\end{proof}

\subsection{Approximation toolbox}
Here we state some useful inequalities. Since some are quite standard, the proofs are omitted.
\begin{lemma}\label{lem:approx 1}
For any $x\in(-0.5,1)$,
$$
1+x\leq e^x\leq1+x+x^2\leq1+2x \, .
$$
In fact for all $x\geq0$, the first inequality holds.
\end{lemma}
\begin{lemma}\label{lem:approx 2}
For any $x\in(0,0.5)$ and $t\in\N$,
$$
1+\frac{xt}{2}\leq e^{\frac{xt}{2}}\leq(1+x)^\top \leq e^{xt} \, .
$$
\end{lemma}

\begin{lemma}
\label{lem:approx log}
For any $\epsilon\in (0, 1)$, we have
\[
\left(\frac{\epsilon}{8}\right)^{1 - \frac{1}{\log\frac{8}{\epsilon}}} =  \frac{\epsilon}{4}.
\]
\end{lemma}
\begin{proof}
Rewrite the expression as the follows.
\[
\left(\frac{\epsilon}{8}\right)^{1 - \frac{1}{\log\frac{8}{\epsilon}}} = \epsilon\cdot \left(\frac{8}{\epsilon}\right)^{\frac{1}{\log\frac{8}{\epsilon}}}\cdot \frac{1}{8}.
\]
It suffices to show that  $\left(\frac{8}{\epsilon}\right)^{\frac{1}{\log\frac{8}{\epsilon}}}\cdot \frac{1}{8} = \frac{1}{4}$. Consider the logarithm of the term, we have
\begin{align*}
\log\left(\left(\frac{8}{\epsilon}\right)^{\frac{1}{\log\frac{8}{\epsilon}}}\cdot \frac{1}{8}\right) &= \frac{1}{\log\frac{8}{\epsilon}}\left(3 + \log\frac{1}{\epsilon}\right) - 3 = 1-3 = -2
\end{align*}
as desired.
\end{proof}

\section{Analyzing the Continuous Version of Oja's Rule}
\label{sec:continuous Oja}
In this section, we introduce the continuous version of Oja's rule and analyze its convergence rate.
The analysis here serves as an inspiration for attacking the discrete dynamic.
To model the continuous dynamics, we use \textit{Brownian motion} to capture the continuous stream of inputs. Surprisingly, it turns out that this continuous version of Oja's rule is \textit{deterministic}. Thus, we are able to use the tools from ODE to easily give an exact characterization of how it converges to the top eigenvector of the covariance matrix.
As a disclaimer, since the analysis for continuous Oja's rule is mainly for intuition, we would omit some mathematical details and point the interested readers to the corresponding resources.

\subsection{Continuous Oja's rule is deterministic}
In the rest of the section we are going to focus on the \textit{diagonal case} where the covariance matrix $A=\diag(\lambda)$ and the goal is showing that $\bw_{t,1}$ goes to 1. This is sufficient since there is a reduction from the general case to the diagonal case as explained in~\autoref{sec:diagonal}.

Intuitively, the continuous dynamic is the limiting process of biological Oja's rule with learning rate $\eta$ going to $0$. Formally, for each $i\in[n]$, let $(\beta_t^{(i)})_{t\geq0}$ be an independent 1-dimensional Brownian motion and let $(B_t)_{t\geq0}$ be an $n$-dimensional random process with the $i^\text{th}$ entry being $B_{t,i}=\sqrt{\lambda_i}\beta_t^{(i)}$ for each $t\geq0$. Now, the difference of $B_t$ should then be thought of as $\eta\bx_t$.

Concretely, to see why $(B_t)_{t\geq0}$ captures the input behavior of streaming PCA in the continuous setting, let us first discretize $(B_t)_{t\geq0}$ using constant step size $\Delta>0$. Now, observe that for each $t\geq0$, $B_{t+\Delta}-B_t$ is an isotropic Gaussian vector with the variance of the $i^\text{th}$ entry being $\lambda_i\cdot\Delta$. Namely,
\begin{equation}\label{eq:discrete quadratic variation}
\frac{1}{\Delta}\Exp\left[\left(B_{t+\Delta}-B_t\right)\left(B_{t+\Delta}-B_t\right)^\top\right] = \diag(\lambda) \, .
\end{equation}
Thus, by discretizing $B_t$ into intervals of length $\Delta>0$, $\left\{\frac{1}{\sqrt{\Delta}}\left(B_{j\cdot\Delta}-B_{(j-1)\cdot\Delta}\right)\right\}_{j\in\N}$ naturally forms a stream of i.i.d. input\footnote{Though here is a caveat that the length of the input vector might not be $1$. Nevertheless, the point of continuous dynamic is not to exactly characterize the limiting behavior of discrete Oja's rule. Instead, the goal here is to capture the intrinsic properties of the biological Oja's rule.} with covariance matrix being $A$. To put this into the context of biological Oja's rule, one should think of $\eta=\Delta$, $\bx_j=\frac{1}{\sqrt{\Delta}}\Delta B_j$, and $y_j=\bx_j^\top\bw_{j-1}$ for each $j\in\N$ where $\Delta B_j=\left(B_{j\cdot\Delta}-B_{(j-1)\cdot\Delta}\right)$~\footnote{Here we abuse the notation of $\Delta$. When we write $\Delta B_j$, the $\Delta$ is regarded as an \textit{operator} instead of the interval length.}.
Then, we get the following dynamic.
\begin{align*}
\bw_j &= \bw_{j-1} + \eta\cdot y_j\left(\bx_j - y_j\bw_{j-1}\right)\\
&= \bw_{j-1} + \Delta B_j^\top\bw_{j-1}\Delta B_j-\left[\Delta B_j^\top\bw_{j-1}\right]^2\bw_{j-1} \, .
\end{align*}

The above dynamics becomes continuous once we let $\Delta\rightarrow0$. Formally, we replace\footnote{This replacement might look weird for those who have not seen Brownian motion before. But this is standard and can be found in textbook such as~\cite{L16}.} $B_{t+\Delta}-B_t$ with $dB_t$ and index the weight vector by $t\geq0$, \textit{i.e.,} $(\bw_t)_{\geq0}$. The above dynamic becomes the following SDE.
\begin{equation}\label{eq:continuous Oja 1}
d\bw_t = dB_t^\top\bw_tdB_t - \left(dB_t^\top\bw_t\right)^2\bw_t \, .
\end{equation}
It might look absurd at first glance (for those who have not seen stochastic calculus before) that there is a quadratic term of $dB_t$ in~\autoref{eq:continuous Oja 1}. Nevertheless, it is in fact mathematically well-defined and we recommend standard resource such as~\cite{L16} for more details. Intuitively, the quadratic term (which is formally called the \textit{quadratic variation}) of a Brownian motion should be thought of as a \textit{deterministic} quantity. Concretely, let $(\beta_t)_{\geq0}$ be a Brownian motion, we have $d\beta_t^2=dt$ almost surely (see~\autoref{lem:quadratic variation}). Thus, for the $(B_t)_{t\geq0}$ defined here, we would have
\[
dB_{t,i}dB_{t,j} = \left\{\begin{array}{ll}
\lambda_idt & , i=j \\
0   & , i\neq j
\end{array}\right.
\]
for each $i,j\in[n]$.
As a consequence, the continuous Oja's rule defined in~\autoref{eq:continuous Oja 1} can be rewritten as the following \textit{deterministic} process almost surely.
\begin{equation}\label{eq:continuous Oja 2}
d\bw_t = \left[\diag(\lambda)\bw_t - \bw_t^\top\diag(\lambda)\bw_t\bw_t\right]dt \, .
\end{equation}
With the continuous Oja's rule being deterministic as in~\autoref{eq:continuous Oja 2}, it is then not difficult to have a tight analysis on its convergence using tools from ODE as explained in the next subsection.

\subsection{One-sided versus two-sided linearization}\label{sec:continuous 1 sided}

In this subsection, we analyze~\autoref{eq:continuous Oja 2} by linearizing the dynamic at $0$ and $1$ respectively and get two incomparable convergence rates (\autoref{thm:continuous lin at 0} and~\autoref{thm:continuous lin at 1}).

\begin{restatable}[Linearization at $0$]{theorem}{contlinzero}\label{thm:continuous lin at 0}Suppose $\bw_{0,1}>0$. For any $\epsilon\in(0,1)$, when $t\geq \Omega\left(\frac{\log(1/\bw_{0,1}^2)}{\epsilon (\lambda_1 - \lambda_2)}\right)$, we have
$\bw_{t, 1}^2 > 1 - \epsilon$.
\end{restatable}

\begin{restatable}[Linearization at $1$]{theorem}{contlinone}\label{thm:continuous lin at 1} Suppose $\bw_{0,1}>0$.
For any $\epsilon\in(0,1)$, when $t\geq \Omega\left(\frac{\log(1/\epsilon)}{\bw_{0,1}(\lambda_1 - \lambda_2)}\right)$, we have
$\bw_{t, 1}^2 > 1 - \epsilon$.
\end{restatable}

The proofs for~\autoref{thm:continuous lin at 0} and~\autoref{thm:continuous lin at 1} are based on applying Taylor's expansion on~\autoref{eq:continuous Oja 2} with center either being $0$ or $1$. Then, we approximate the dynamics with linear differential equations and use tools from ODE to get an tight analysis.
See~\autoref{sec:continuous oja linearization} for the details on the linearizations of continuous Oja's rule.

When starting with a random vector, \textit{i.e.,} $\bw_{0,1}=\Omega(1/\sqrt{n})$ with high probability, the above convergence rates become $O(\frac{\log n}{\epsilon (\lambda_1 - \lambda_2)})$ and $O(\frac{\sqrt{n}\log(1/\epsilon)}{\lambda_1 - \lambda_2})$ respectively.
This indicates that linearizing only on one side (either at $0$ or at $1$) would not give tight analysis. Nevertheless, if we invoke~\autoref{thm:continuous lin at 0} with the error parameter being $0.5$, then for some $t_1=O(\frac{\log n}{\lambda_1 - \lambda_2})$, we have $\bw_{t_1,1}>0.5$. Next, we invoke~\autoref{thm:continuous lin at 1} starting from $\bw_{t_1}$ and with the error parameter being $\epsilon$, then for some $t_2=O(\frac{\log(1/\epsilon)}{\lambda_1 - \lambda_2})$, we have $\bw_{t_1+t_2,1}>1-\epsilon$. Putting these together, we have the following theorem combining the linearizations on both sides.

\begin{theorem}[Linearization at both $0$ and $1$]\label{thm:continuous 2 side}Suppose $\bw_{0,1}>0$. For any $\epsilon\in(0,1)$, when \[t \geq \Omega\left(\frac{\log\frac{1}{\bw_{0,1}^2}+\log\frac{1}{\epsilon}}{\lambda_1-\lambda_2}\right),\] we have $\bw_{t,1}^2>1-\epsilon$.
\end{theorem}

The above theorem for the convergence rate of the continuous Oja's rule gives three key insights. First, it suggests that one should linearize at 0 in the beginning of the process and switch to linearizing at 1 when $\bw_{t,1}$ becomes $\Omega(1)$. Second, after the linearization, using linear ODE to give exact characterization of the dynamic would give tight analysis. Finally, the continuous dynamic is deterministic and will stay around the optimal region for all time after certain point. This suggests that the \textit{for-all-time} guarantee could potentially happen in the original discrete setting.

\section{Main Result}\label{sec:main}
Now, let us state the formal version of the main theorem for the biological Oja's rule. In the following, all of the theorems and lemmas are stated with respect to the setting of~\autoref{prob:streaming PCA} and~\autoref{def:bio oja}. Thus, for simplicity, we would not repeat the setup in their statements.

In~\autoref{thm:main}, we show that both the local and the global convergence of Oja's rule are efficient. We remind readers that in the local convergence setting, the weight vector is correlated with the top eigenvector by a constant while in the global convergence setting, the weight vector is randomly initiated. In~\autoref{thm:main stay} we show that once $\bw_{t}$ becomes $\epsilon$-close to the top eigenvector $\bv_1$, it will stay in the neighborhood of $\bv_1$ for a long time without decreasing the learning rate too much. This demonstrates the capacity of Oja's rule as a continual learning mechanism in a living system. 

\begin{theorem}[Main Theorem]\label{thm:main}We have the following results on the local and global convergence of Oja's rule.\\
$\bullet$ (Local Convergence) Let $n\in\N$, $\delta\in(0,1), \epsilon\in (0, \frac{1}{8})$. Suppose $\frac{\langle\bw_0, \bv_1\rangle^2}{\|\bw_0\|_2^2}\geq2/3$. Let
\[
\eta = \Theta\left(\frac{\epsilon(\lambda_1 - \lambda_2)}{\lambda_1\log\frac{\log\log\frac{1}{\epsilon}}{\delta}}\right),\ T = \Theta\left(\frac{\log \frac{1}{\epsilon}}{\eta (\lambda_1 - \lambda_2)}\right).
\]
Then, we have 
\[
\Pr\left[\frac{\langle\bw_T, \bv_1\rangle^2}{\|\bw_T\|_2^2} < 1-\epsilon\right] < \delta \,.
\]
Namely, the convergence rate is of order
\[
\Theta\left(\frac{\lambda_1\log\frac{1}{\epsilon}\left(\log\log\log\frac{1}{\epsilon}+\log\frac{1}{\delta}\right)}{\epsilon(\lambda_1-\lambda_2)^2}\right)
\]
with probability at least $1-\delta$.\\
$\bullet$ (Global Convergence) Let $n\in\N$, $\delta\in(0,1), \epsilon\in (0, \frac{1}{4})$. Suppose $\bw_0$ is uniformly sampled from the unit sphere of $\Real^n$. Let
\[
\eta = \Theta\left(\frac{\lambda_1 - \lambda_2}{\lambda_1}\cdot\left(\frac{\epsilon}{\log\frac{\log\frac{n}{\epsilon}}{\delta}} \bigwedge \frac{\delta^2}{\log^2 \frac{\lambda_1 n}{\delta(\lambda_1 - \lambda_2)^2}}\right)\right),\,T = \Theta\left(\frac{\log \frac{1}{\epsilon} + \log\frac{n}{\delta}}{\eta (\lambda_1 - \lambda_2)}\right).
\]
Then, we have
\[
\Pr\left[\frac{\langle\bw_T, \bv_1\rangle^2}{\|\bw_T\|_2^2} < 1-\epsilon\right] < \delta \,.
\]
Namely, the convergence rate is of order
\[
\Theta\left(\frac{\lambda_1\left(\log \frac{1}{\epsilon} + \log\frac{n}{\delta}\right)}{(\lambda_1 - \lambda_2)^2}\cdot\max\left\lbrace \frac{\log\frac{\log\frac{n}{\epsilon}}{\delta}}{\epsilon},\, \frac{\log^2 \frac{\lambda_1 n}{\delta(\lambda_1 - \lambda_2)^2}}{\delta^2}  \right\rbrace\right)
\]
with probability at least $1-\delta$.
\end{theorem}

\paragraph{Proof structure of~\autoref{thm:main}.}
To prove~\autoref{thm:main}, we first reduce the general setting where the covariance matrix $A$ is PSD to the special case where $A=\diag(\lambda)$ in~\autoref{sec: preprocessing}. For local convergence, we show that starting from constant correlation, Oja's rule can efficiently converge to the top eigenvector up to arbitrarily small error in~\autoref{thm:discrete phase 2} of~\autoref{sec: local convergence}. For global convergence, we show that starting from random initialization, Oja's rule can efficiently converge to the top eigenvector up to arbitrarily small error in~\autoref{thm:discrete global main} of~\autoref{sec: global convergence}. To get tight analysis for global convergence, we need to take an extra care on a \textit{cross term} in~\autoref{thm: D main} of~\autoref{sec: D main}. See~\autoref{fig:framework} for the proof structure of these theorems.

\begin{figure}[h]
    \centering
    \begin{tikzpicture}
	\begin{pgfonlayer}{nodelayer}
		\node [style=Theorem] (0) at (0, 2.5) {Local Convergence\\ \autoref{thm:discrete phase 2}};
		\node [style=Theorem] (1) at (8, 2.5) {Global Convergence\\ \autoref{thm:discrete global main}};
		\node [style=Theorem] (2) at (16, 2.5) {Cross Term\\ \autoref{thm: D main}};
		\node [style=Step] (3) at (-8, -1.75) {Step 1\\ (Linearization \& Moment Analysis)};
		\node [style=Step] (5) at (-8, -11.75) {Step 3\\ (Interval Analysis)};
		\node [style=Lemma] (6) at (0, -0.5) {\autoref{lem:discrete phase 2 linearization}\\(Linearization)};
		\node [style=Lemma] (7) at (8, -0.5) {\autoref{lem:discrete phase 1 linearization}\\(Linearization)};
		\node [style=Lemma] (8) at (16, -0.5) {\autoref{lem:D linearization}\\(Linearization)};
		\node [style=Lemma] (9) at (0, -3) {\autoref{lem: local bound}\\(Moment)};
		\node [style=Lemma] (10) at (8, -3) {\autoref{lem: global bound}\\(Moment)};
		\node [style=Lemma] (11) at (16, -3) {\autoref{lem:D diff var vector}\\(Moment)};
		\node [style=Lemma] (12) at (0, -6.25) {\autoref{lem:discrete phase 2 stopped concentration}\\(Stopped noise)};
		\node [style=Lemma] (13) at (0, -8.75) {\autoref{cor: local concentration}\\(Noise)};
		\node [style=Lemma] (14) at (8, -6.25) {\autoref{lem:discrete phase 1 stopped concentration}\\(Stopped noise)};
		\node [style=Lemma] (15) at (8, -8.75) {\autoref{lem:discrete phase 1 small stopping time}\\(Noise)};
		\node [style=Lemma] (16) at (16, -6.25) {\autoref{lem:D stopped concentration}\\(Stopped noise)};
		\node [style=Lemma] (17) at (16, -8.75) {\autoref{lem: D main}\\(Noise)};
		\node [style=Lemma] (18) at (0, -11.5) {\autoref{sec: local interval}};
		\node [style=Lemma] (19) at (8, -11.5) {\autoref{sec:phase 1 wrap up}};
		\node [style=Lemma] (20) at (16, -11.5) {\autoref{sec:D wrap up}};
		\node [style=Step] (21) at (-8, -7.25) {Step 2\\ (Improvement Analysis)};
		\node [style=none] (22) at (-10.75, -4.5) {};
		\node [style=none] (23) at (20.25, -4.5) {};
		\node [style=none] (24) at (-10.75, -10.25) {};
		\node [style=none] (25) at (20.25, -10.25) {};
		\node [style=none] (26) at (-10.75, 1) {};
		\node [style=none] (27) at (20.25, 1) {};
		\node [style=none] (28) at (-4.5, 3.75) {};
		\node [style=none] (29) at (-4.5, -12.75) {};
		\node [style=none] (30) at (4, 3.75) {};
		\node [style=none] (31) at (4, -12.75) {};
		\node [style=none] (32) at (12, 3.75) {};
		\node [style=none] (33) at (12, -12.75) {};
	\end{pgfonlayer}
	\begin{pgfonlayer}{edgelayer}
		\draw (22.center) to (23.center);
		\draw (24.center) to (25.center);
		\draw (26.center) to (27.center);
		\draw (28.center) to (29.center);
		\draw (30.center) to (31.center);
		\draw (32.center) to (33.center);
	\end{pgfonlayer}
\end{tikzpicture}
    \caption{The proof structure of key theorems. Here we present the structure of the three main theorems using the three-step framework described in a follow-up paper~\cite{Chou2020}.}
    \label{fig:framework}
\end{figure}

\begin{theorem}[Continual Learning]\label{thm:main stay}
We have the following results on the continual learning aspects of Oja's rule.\\
$\bullet$ (Finite continual learning) Let $n, l\in\N$, $\epsilon,\delta\in(0,1)$. Suppose $\frac{\langle\bw_0, \bv_1\rangle^2}{\|\bw_0\|_2^2}\geq 1-\frac{\epsilon}{2}$. Let
\[
\eta = \Theta\left(\frac{\epsilon(\lambda_1 - \lambda_2)}{\lambda_1\log\frac{l}{\delta}}\right).
\]
Then
\[
\Pr\left[\exists 1\leq t\leq \Theta\left(\frac{l}{\eta(\lambda_1 - \lambda_2)}\right),\ \frac{\langle\bw_T, \bv_1\rangle^2}{\|\bw_T\|_2^2} < 1-\epsilon\right] < \delta \, .
\]
$\bullet$ (For-all-time continual learning)
Let $n, t_0\in\N$, $\epsilon,\delta\in(0,1)$. Suppose $\frac{\langle\bw_0, \bv_1\rangle^2}{\|\bw_0\|_2^2}\geq 1-\frac{\epsilon}{2}$. Then there is 
\[
\eta_t \geq \Theta\left(\frac{\epsilon(\lambda_1 - \lambda_2)}{\lambda_1\log\frac{t}{\delta}}\right)
\] 
such that
\[
\Pr\left[\exists t\in\N,\ \frac{\langle\bw_t, \bv_1\rangle^2}{\|\bw_t\|_2^2} < 1-\epsilon\right] < \delta \, .
\]
\end{theorem}

\paragraph{Proof structure of~\autoref{thm:main stay}.}
We first reduce the general setting to the special case where $A=\diag(\lambda)$ in~\autoref{sec: preprocessing}. The proof of finite continual learning is then a direct application of techniques developed in local convergence. By repetitively applying finite continual learning, we can show for-all-time continual learning. The results will be proven in~\autoref{sec: finite continual learning}.

\section{Preprocessing}\label{sec: preprocessing}

Before the main analysis of the biological Oja's rule, we provide two useful observations on the dynamic in this section. Specifically, we show in~\autoref{sec:diagonal} that considering the covariance matrix being \textit{diagonal} is sufficient for the analysis and in~\autoref{sec:oja bounded} that $\|\bw_{t}\|_2^2=1\pm O(\eta)$ almost surely for all $t\in\N$.

\subsection{A reduction to the diagonal case}\label{sec:diagonal}
In this subsection, we show that it suffices to analyze the case where the covariance matrix $A$ is a diagonal matrix $D$.
Recall that $A$ is defined as the expectation of $\bx\bx^\top$ and thus it is positive semidefinite. Namely, there exists an orthonormal matrix $U$ and a diagonal matrix $D$ such that $A=UDU^\top$. Especially, the eigenvalues of $A$, \textit{i.e.,} $1\geq\lambda_1\geq\lambda_2\geq\cdots\geq\lambda_n\geq0$, are the entries of $D$ from top left to bottom right on the diagonal. Thus, by a change of basis, we can focus on the case where $A=D$ without loss of generality.

To see this, consider $\tilde{\bw_t}=U\bw_t$ and $\tilde{\bx_t}=U\bx_t$. As $U^\top U=UU^\top=I$, we have $\tilde{\bx}_t^\top\tilde{\bw}=\bx_t^\top\bw$ and $\Exp[\tilde{\bx}\tilde{\bx}^\top]=D$.
Let $\bv_1$ be the top eigenvector of $A$ (\textit{i.e.,} the first row of $U$), we also have
$$
\|\bw_t-\bv_1\|_2 = \|U\bw_t-U\bv_1\|_2 = \|\tilde{\bw}_t-\be_1\|_2
$$
where $\be_1$ is the indicator vector for the first coordinate. Namely, it suffices to analyze how fast does $\tilde{\bw}_t$ converge to $\be_1$.
Thus, we without loss of generality consider the diagonal case where the goal would be showing that $\bw_{t,1}^2\geq1-\epsilon$.

\subsection{Bounded conditions of Oja's rule}\label{sec:oja bounded}
In this section, we show that the $\ell_2$ norm of the weight vector is always close to $1$.
\begin{lemma}\label{lem:oja w norm ub}
For any $\eta\in(0,0.1)$, if for all $t\in\N$, $\eta_t\leq\eta$, then for all $t\in\Nz$, $1-10\eta\leq\|\bw_t\|_2^2\leq 1+10\eta$ almost surely.
\end{lemma}
\begin{proof}[Proof of~\autoref{lem:oja w norm ub}]
Here we prove only the upper bound while the lower bound can be proved using the same argument.
The proof is based on induction. For the base case where $t=0$, we have $\|\bw_0\|_2^2=1$ from the problem setting. For the induction step, consider any $t\in\N$ such that $\bw_{t-1}$ satisfies the bounds, we have
\begin{align*}
\|\bw_{t}\|_2^2 &= \|\bw_{t-1}\|_2^2 + 2\eta_t\bw_{t-1}^\top\left[y_t\bx_t-y_t^2\bw_{t-1}\right] + \eta^2_t\cdot\|y_t\bx_t-y_t^2\bw_{t-1}\|^2\\
&=\|\bw_{t-1}\|_2^2 - 2\eta_t(y_t)^2\cdot(\|\bw_{t-1}\|_2^2-1) + 2\eta^2_ty_t^2\cdot\max\{\|\bx_t\|_2^2,y_t^2\|\bw_{t-1}\|_2^2\} \, .
\end{align*}
Consider two cases: (i) $\|\bw_{t-1}\|_2^2\leq1+8\eta$ and (ii) $1+8\eta<\|\bw_{t-1}\|_2^2\leq1+10\eta$. Note that $\|\bw_{t}\|_2^2\leq1+10\eta$ in both cases. This completes the induction and the proof.
\end{proof}

\section{Local Convergence: Starting With Correlated Weights}\label{sec: local convergence}
For the local convergence result, the synaptic weight $\bw_0$ is correlated with the top eigenvector by a constant. To be precise, we suppose that $\bw_{0, 1}^2\geq \tfrac{2}{3}$. The goal of this section is to show that $1-\bw_{ t,1}^2\leq\epsilon$ for some $t=O(\tfrac{\lambda_1\log(1/\epsilon)(\log\log\log(1/\epsilon)+\log(1/\delta))}{\epsilon(\lambda_1-\lambda_2)^2})$ for any small $\epsilon>0$. Let us first state the main theorem of this section as follows.

\begin{theorem}[Local convergence of the diagonal case]\label{thm:discrete phase 2}
Suppose $\bw_{0,1}^2\geq2/3$. For any $n\in\N$, $\delta\in(0,1), \epsilon\in (0, \frac{1}{8})$, let
\[
\eta = \Theta\left(\frac{\epsilon(\lambda_1 - \lambda_2)}{\lambda_1\log\frac{\log\log\frac{1}{\epsilon}}{\delta}}\right),\ T = \Theta\left(\frac{\log \frac{1}{\epsilon}}{\eta (\lambda_1 - \lambda_2)}\right).
\]
Then
\[
\Pr\left[\bw_{T, 1}^2 < 1-\epsilon\right] < \delta \, .
\]
Namely, the convergence rate is of order $\Theta\left(\frac{\lambda_1\log\frac{1}{\epsilon}\left(\log\log\log\frac{1}{\epsilon}+\log\frac{1}{\delta}\right)}{\epsilon(\lambda_1-\lambda_2)^2}\right)$ with probability at least $1-\delta$. 
\end{theorem}

\paragraph{Proof overview and organization.}
First note that by applying the diagonal reduction argument in~\autoref{sec:diagonal},~\autoref{thm:discrete phase 2} implies the local convergence part of~\autoref{thm:main} as a corollary. The proof structure of~\autoref{thm:discrete phase 2} is as follows. First, in~\autoref{sec:phase 2 linearization} we derive a linearization of the dynamic using a center at 1 instead of 0 based on the intuition from the continuous dynamic in~\autoref{sec:continuous Oja}. Furthermore, we use the ODE trick to write down the dynamic in a closed form with respect to the linearization.

Next, in~\autoref{sec:phase 2 noise}, we want to show that the noise term is small. However, the difficulty here is that $\bw_{t,1}$ might \textit{go back} to the small region (\textit{e.g.,} $\bw_{t,1}^2<\tfrac{1}{3}$) and thus the bounded difference might become too large to effectively bound the noise with  Freedman's inequality. To deal with this issue, we consider a stopping time where $\bw_{t,1}^2 < 1-a$ to give good control on the bounded difference and subsequently bound the stopped version of the noise term in~\autoref{lem:discrete phase 2 stopped concentration}. After we show that the stopped noise term is small, we want to pull out the stopping time to show the concentration on the original noise term. In general, pulling out the stopping time is impossible without introducing extra failure probability; however, by exploiting the structure of the dynamic, we are able to pull out the stopping time without additional cost in~\autoref{lem:discrete phase 2 stopping time}. 

Finally in~\autoref{sec: local interval}, by combining the small noise and the ODE trick, we are able to prove~\autoref{thm:discrete phase 2} with an interval analysis. As a corollary of~\autoref{cor: local concentration} in the local convergence, we show that biological Oja's rule has the continual learning capacity in~\autoref{sec: finite continual learning}. In a biological system, it is important to function for a long period of time instead of at one time point. In this section, we prove two theorems on continual learning. ~\autoref{thm: finite continual learning} guarantees Oja's rule can maintain the convergence for any finite time length efficiently while~\autoref{thm: continual learning} guarantees Oja's rule can function for all time without sacrificing too much efficiency to adapt to a new environment. 
\subsection{Linearization and ODE trick centered at 1}\label{sec:phase 2 linearization}
In this section, we derive the linearization of Oja's rule with a center at $1$ in~\autoref{lem:discrete phase 2 linearization} and the closed form solution of Oja' rule in~\autoref{cor:discrete phase 2 ODE}. In addition, we show that the bounded differences and moments of the noise can be controlled in~\autoref{lem: local bound}.

In the analysis of the local convergence, we use the linearization with a center at $1$ instead of $0$. The idea is inspired from the analysis of the continuous dynamics as explained in~\autoref{sec:continuous Oja}. To ease the notation, we define $\tilde{\bw}_{t,1}=\bw_{t,1}-1$ and the goal becomes to show that $\tilde{\bw}_{t_0+t_2,1}>-\epsilon$ with probability at least $1-\delta$. The following lemma states the linearization for $\tilde{\bw}_{t,1}$.

\begin{lemma}[Linearization at $1$]\label{lem:discrete phase 2 linearization}
Let $\tilde{\bw}_{t} = \bw_{t, 1}^2 - 1$ and $\bz_t = \bx_ty_t - y_t^2\bw_{t-1}$. For any $t\in\N_{\geq 0}$ and $\eta\in(0,1)$, we have
\[
\tilde{\bw}_{t} \geq  H\cdot\tilde{\bw}_{t-1} + A_{t} + B_{t}
\]
almost surely, where 
\begin{align*}
H &=1-\frac{2}{3}(\lambda_1 - \lambda_2)\eta \, ,\ \\
A_t &= 2\eta\bz_{t,1}\bw_{t-1, 1} + \eta^2\bz_{t,1}^2 - \Exp\left[2\eta\bz_{t,1}\bw_{t-1, 1}\, |\, \bw_{t-1} \right] + 2\eta\lambda_2(1-\|\mathcal{F}_{t-1}\|^2)\bw_{t-1, 1}^2 \, ,\ \\
B_t &= -2\eta(\lambda_1 - \lambda_2)\tilde{\bw}_{t, 1}(\frac{2}{3} + \tilde{\bw}_{t, 1}) \, .
\end{align*}
\end{lemma}
\begin{proof}[Proof of~\autoref{lem:discrete phase 2 linearization}]
By expanding $\bw^2_{t, 1}$ with the Oja's rule (\autoref{eq:oja}), we have
\begin{align}
\bw^2_{t,1} &= \bw^2_{t-1,1} + 2\eta\bz_{t, 1}\bw_{t-1,1} + \eta^2\bz_{t}^2\,. \nonumber
\intertext{Add and subtract $\Exp\left[2\eta\bz_{t,1}\bw_{t-1, 1}\, |\, \mathcal{F}_{t-1} \right] - 2\eta\lambda_2(1-\|\bw_{t-1}\|^2)\bw_{t-1}^2$. We have}
&= \bw_{t-1,1}^2 + 2\eta(\lambda_1\bw_{t-1, 1}^2 - \sum_{i=1}^n\lambda_i\bw_{t-1,i}^2\bw_{t-1,1}^2 - \lambda_2(1-\|\bw_{t-1}\|^2)\bw_{t-1,1}^2 ) + A_{t} \, .\nonumber
\intertext{Upper bound $\sum_{i=2}^n\lambda_i\bw_{t-1,i}^2\bw_{t-1,1}^2$ by $\lambda_2\sum_{i=2}^n\bw_{t-1,i}^2\bw_{t-1,1}^2$, we then have}
&\geq \bw_{t-1,1}^2 + 2\eta(\lambda_1(\bw_{t-1, 1}^2 - \bw_{t-1, 1}^4) - \lambda_2(\bw_{t-1, 1}^2 - \bw_{t-1, 1}^4)) + A_{t}\nonumber\\
&= \bw_{t-1,1}^2  + 2\eta(\lambda_1 - \lambda_2)\bw_{t-1, 1}^2(1 - \bw_{t-1, 1}^2) + A_{t} \, . \label{eq: linear 1}
\end{align}
Based on the intuition from the continuous dynamic in~\autoref{sec:continuous Oja}, since we want to converge from constant error to $\epsilon$ error, we want to linearize at $1$. Hence we rewrite~\autoref{eq: linear 1} in terms of $\tilde{\bw}_{t, 1} = \bw_{t, 1}^2 - 1$ and get
\begin{align*}
\tilde{\bw_{t}} &\geq \tilde{\bw}_{t-1}  - 2\eta(\lambda_1 - \lambda_2)\tilde{\bw}_{t-1}(1 + \tilde{\bw}_{t-1}) + A_{t}\\
&= H\cdot\tilde{\bw}_{t-1} + A_{t} + B_{t}
\end{align*}
as desired.
\end{proof}

We apply the ODE trick (see~\autoref{lem:ODE trick}) on~\autoref{lem:discrete phase 2 linearization} and get the following corollary. 

\begin{corollary}[ODE trick]\label{cor:discrete phase 2 ODE}
For any $t_0\in\Nz$, $t\in\N$, and $\eta\in(0,1)$, we have
\[
\tilde{\bw}_{t_0+t} \geq H^{t} \cdot\left(\tilde{\bw}_{t_0} + \sum_{i=t_0+1}^{t_0+t}\frac{A_i+ B_i}{H^{i-t_0}} \right) \, .
\]
\end{corollary}

To control the noise term, we need to have bounds on the bounded differences and the moments of $A_i, B_i$. 

\begin{lemma}\label{lem: local bound}
Let $A_t, B_t$ be defined as in~\autoref{lem:discrete phase 2 linearization}. For any $t\in\N$, we have $A_t,B_t$ satisfy the following properties:
\begin{itemize}
\item (Bounded difference) $|A_t|=O(\eta|\tilde{\bw}_{t-1}| + \eta|\tilde{\bw}_{t-1}|^{\frac{1}{2}} + \eta^{\frac{3}{2}})$ almost surely. If $\tilde{\bw}_{t-1,1}\geq -\frac{2}{3}$, then $B_t\geq - O(\eta^2)$ almost surely.
\item (Conditional expectation) $\Exp[A_t\ |\ \mathcal{F}_{t-1}]=O(\eta^2\lambda_1)$.
\item (Conditional variance) $\Var\left[A_t\ |\ \mathcal{F}_{t-1}\right]=O\left(\eta^2\lambda_1\left(\left|\tilde{\bw}_{t-1}\right|^2+ |\tilde{\bw}_{t-1}| + \eta\right)\right)$.
\end{itemize}
\end{lemma}
\begin{proof}
First by~\autoref{lem:oja w norm ub}, we have $|\bw_{t, 1}|, |y_t| < \sqrt{1+10\eta} < 1+10\eta < 2$. Now let's bound $|\bz_{t, 1}|$ first. By expanding $|\bz_{t, 1}|$, we have
\begin{align}
|\bz_{t,1}| &= |y_t(\bx_{t, 1} - y_t\bw_{t-1, 1})|\nonumber\\
&=  \left|y_t\left(\bx_{t, 1}(1 -\bw_{t-1, 1}^2) - \sum_{i=2}^n\bx_{t, i}\bw_{t-1, i}\bw_{t-1, 1}\right)\right|\nonumber\\
&\leq |y_t|\cdot\left(\left|\bx_{t, 1}\tilde{\bw}_{t-1}\right| + \left|\sum_{i=2}^n\bx_{t, i}\bw_{t-1, i}\bw_{t-1, 1}\right|\right).\nonumber
\intertext{By Cauchy-Schwarz and the fact that $\|\bx\|_2=1$, we have}
&\leq |y_t|\cdot\left(\left|\tilde{\bw}_{t-1}\right| + \left|\sqrt{\left(\sum_{i=2}^n\bx_{t, i}^2\right)\left(\sum_{i=2}^n\bw_{t-1, i}^2\right)}\bw_{t-1, 1}\right|\right).\nonumber
\intertext{By~\autoref{lem:oja w norm ub} and the definition of $\tilde{\bw}_{t-1}$, we have}
&\leq |y_t|\cdot\left(\left|\tilde{\bw}_{t-1}\right| + \left|\sqrt{-\tilde{\bw}_{t-1} + 10\eta}\right|\right)\nonumber\\
&\leq|y_t|\cdot\left(\left|\tilde{\bw}_{t-1}\right|+ \sqrt{|\tilde{\bw}_{t-1}|} + \sqrt{10\eta}\right).\label{eq: local bound 1}
\intertext{Since $|y_t|\leq 2$, we have}
&\leq 2\left(\left|\tilde{\bw}_{t-1}\right|+ \sqrt{|\tilde{\bw}_{t-1}|} + \sqrt{10\eta}\right).\nonumber
\end{align}
Combining above,~\autoref{lem:oja w norm ub} and the fact that $\bz_{t, 1} = O(1)$, we have
\begin{equation*}
|A_t| = O\left(\eta|\tilde{\bw}_{t-1}| + \eta|\tilde{\bw}_{t-1}|^{\frac{1}{2}} + \eta^{\frac{3}{2}}\right) 
\end{equation*}
and for $\tilde{\bw}_{t-1}\geq -\frac{2}{3}$, we have $B_t \geq -O(\eta^2)$ because $\frac{2}{3} +\tilde{\bw}_{t-1} > 0$ and $\tilde{\bw}_{t-1} \leq O(\eta)$. 

For conditional expectation, notice that $\Exp[y_t^2| \mathcal{F}_{t-1}] = \bw_{t-1}^\top \diag(\lambda)\bw_{t-1} = O(\lambda_1)$. This implies that $\Exp[\bz_{t, 1}^2| \mathcal{F}_{t-1}] = O(\lambda_1)$ and hence
 $\Exp\left[A_t|\mathcal{F}_{t-1}\right] = O(\eta^2\lambda_1)$.
Now the conditional variance is
\begin{align*}
\Var\left[A_t|\mathcal{F}_{t-1}\right] &= O\left(\eta^2\Exp[\bz_{t, 1}^2\, |\, \mathcal{F}_{t-1}]\bw_{t-1,1}^2 + \lambda_1\eta^4\right). 
\intertext{By~\autoref{eq: local bound 1}, we have}
&= O\left(\eta^2\Exp[y_t^2| \mathcal{F}_{t-1}]\left(\left|\tilde{\bw}_{t-1}\right|+ \sqrt{|\tilde{\bw}_{t-1}|} + \sqrt{10\eta}\right)^2 + \lambda_1\eta^4\right)\\
&= O\left(\eta^2\lambda_1\left(\left|\tilde{\bw}_{t-1}\right|^2+ |\tilde{\bw}_{t-1}| + \eta\right)\right)
\end{align*} 
as desired.
\end{proof}

\subsection{Concentration of noise and pulling out the stopping time}\label{sec:phase 2 noise}

In this subsection, we want to show that the noise term in~\autoref{cor:discrete phase 2 ODE} is small. Specifically, we prove the following lemma.

\begin{restatable}[Concentration of the noise term in local convergence]{lemma}{localconcentration}
\label{cor: local concentration}
Let $\epsilon, \delta\in (0, 1), T\in \N_{\geq 0}$. Suppose given $t_0\in \mathbb{N}$, $v_0\in (-\frac{1}{3}, 0)$ and $a\in [0, 1]$, we have $\tilde{\bw}_{t_0} \geq v_0$ and $v_0 = -\Theta(\epsilon^{1-a})$. Let $\eta = \Theta\left(\frac{\epsilon(\lambda_1 - \lambda_2)}{\lambda_1\log\frac{1}{\delta}}\right)$. If $H^{-T} = \Theta(\epsilon^{-\frac{a}{2}})$, then
\[
\Pr\left[\min_{1\leq t \leq T} \sum_{i=t_0 + 1}^{t_0 + t} \frac{A_i+ B_i}{H^{i-t_0}} \leq v_0\right] < \delta \, .
\]
\end{restatable}

The most natural way to prove such a statement is using a martingale concentration inequality. However, the difficulty here is that $\tilde{\bw}_{t}$ might \textit{go back} to the small region (\textit{e.g.,} $\tilde{\bw}_{t}<-2/3$) and thus the bounded difference might become too large to bound the noise effectively with Freedman's inequality. Nevertheless, the continuous dynamic (see~\autoref{sec:continuous Oja}) suggests that this situation should happen with only a small probability because the $\bw_1$ term in the continuous dynamic increases monotonically to $1$. To enforce the analysis, we consider a \textit{stopped process} where the dynamic stops once $\tilde{\bw}_{t}$ is too small. This stopped process satisfies good bounded difference conditions by its construction and thus we can apply Freedman's inequality on it. See~\autoref{lem:discrete phase 2 stopped concentration} for a formal statement of the above intuition.

After obtaining good control of the noise term in the stopped process, we want to remove the stopping time and show the concentration of the original non-stopped process in order to prove~\autoref{cor: local concentration}. This can be done by~\autoref{lem:discrete phase 2 stopping time} which \textit{pulls out} the stopping time from the concentration inequality for the stopped process. In general, pulling out the stopping time is impossible without introducing additional failure probability; however, the following structure of the stochastic process we are looking at allows us to pull out the stopping time. Intuitively, given a stopping time $\tau$ with $\tau\geq t$ for some $t$, with high probability all the noise terms before time $t$ are small (using a maximal martingale inequality). Next, the noise being small at time $t$ would further imply that $\tau\geq t+1$ (using the ODE trick). The above argument forms a chain of implications as pictured in~\autoref{fig:intuition pull out}.

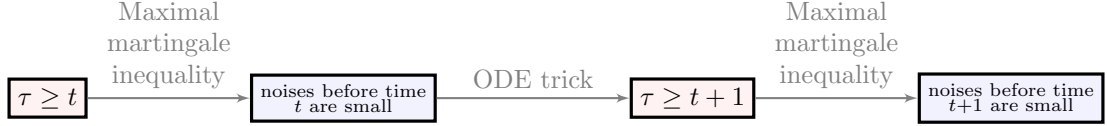
\begin{figure}[ht!]
	\centering
	\begin{tikzpicture}[
	squarednode/.style={rectangle, draw=black!100, fill=red!5, very thick, minimum size=5mm},
	noisenode/.style={rectangle, draw=black!100, fill=blue!5, very thick, minimum size=5mm},
	]
	\tikzstyle{line} = [draw, thick, color=black!50, -latex']
	\path node (time1) [squarednode] {$\tau\geq t$};
	\path (time1.east) + (3.4,0) node (noise1) [noisenode] {$\substack{\text{noises before time}\\t\text{ are small}}$};
	\path (noise1.east) + (3.4,0) node (time2) [squarednode] {$\tau\geq t+1$};
	\path (time2.east) + (3.4,0) node (noise2) [noisenode] {$\substack{\text{noises before time}\\t+1\text{ are small}}$};

	\path [line] (time1) -- node [text width=2cm,midway,above,align=center] {Maximal martingale inequality} (noise1);
    \path [line] (noise1) -- node [text width=2cm,midway,above,align=center] {ODE trick} (time2);
	\path [line] (time2) -- node [text width=2cm,midway,above,align=center] {Maximal martingale inequality} (noise2);
	
	\end{tikzpicture}
	\caption{Intuition on why it is possible to pull out stopping time in Phase 2.}
	\label{fig:intuition pull out}
\end{figure}

With the above \textit{chain} structure in the noise terms, we are then able to pull out the stopping time in~\autoref{lem:discrete phase 2 stopped concentration} by introducing another stopping time to help us properly partition the probability space.
The rest of this subsection is devoted to formalizing the above intuition and completing the proof for~\autoref{cor: local concentration}.

First, let us show the concentration of the stopped process.
\begin{lemma}[Concentration of stopped noise in an interval]\label{lem:discrete phase 2 stopped concentration}
Let $\epsilon, \delta\in (0, 1), T\in \N_{\geq 0}$. Suppose given $t_0\in \mathbb{N}$, $v_0\in (-\frac{1}{3}, 0)$ and $a\in [0, 1]$, we have $\tilde{\bw}_{t_0} \geq v_0$ and $v_0 = -\Theta(\epsilon^{1-a})$. Let $\tau_{v_0}$ to be the stopping time $\lbrace \tilde{\bw}_t < 2v_0 \rbrace$ such that $t > t_0$. Let $\eta = \Theta\left(\frac{\epsilon(\lambda_1 - \lambda_2)}{\lambda_1\log\frac{1}{\delta}}\right)$. If $H^{-T} = \Theta(\epsilon^{-\frac{a}{2}})$, then
\[
\Pr\left[\min_{1\leq t \leq T} \sum_{i=t_0 + 1}^{(t_0 + t)\wedge \tau_{v_0}} \frac{A_i+ B_i}{H^{i-t_0}} \leq v_0\right] < \delta \, .
\]
\end{lemma}
\begin{proof}[Proof of~\autoref{lem:discrete phase 2 stopped concentration}]
We are going to apply Freedman's inequality~\autoref{cor:freedman} on the stopped process $\sum_{i=t_0 + 1}^{(t_0 + t)\wedge \tau_{v_0}} \frac{A_i}{H^{i-t_0}}$. First notice that given a stopping time $\tau$ and an adapted stochastic process $M_t$, the difference of the stopped process can be described as
\[
M_{t\wedge\tau} - M_{(t-1)\wedge\tau} = \mathbf{1}_{\tau \geq t}(M_t - M_{t-1}) \, .
\] 
For notational convenience, we denote $\mathbf{1}_{\tau_{v_0}\geq (t_0 + t)}A_t$ as $\bar{A}_t$. Now by~\autoref{lem: local bound} and geometric series, \textit{i.e.,} $\sum_{i=1}^T H^{-i}\leq O(\frac{H^{-T}}{\eta(\lambda_1-\lambda_2)})$, we have
\[
\forall 1\leq t\leq T,\ \left|\frac{\bar{A}_{t_0 + t}}{H^\top }\right|\leq O\left(\eta \epsilon^{\frac{1-a}{2}}\right) \, ,
\]
\[
\left|\sum_{i=t_0 + 1}^{t_0 + T} \Exp\left[\frac{\bar{A}_i}{H^{i-t_0}}\ \middle|\ \mathcal{F}_{i-1} \right]\right| \leq O\left(\eta^2\frac{H^{-T}}{\eta(\lambda_1 - \lambda_2)}\right) = O\left(\frac{\eta\lambda_1 \epsilon^{-\frac{a}{2}}}{\lambda_1 - \lambda_2}\right) \text{, and}
\]
\[
\left|\sum_{i=t_0 + 1}^{t_0 + T} \Var\left[\frac{\bar{A}_i}{H^{i-t_0}} \ \middle|\ \mathcal{F}_{i-1}\right]\right| \leq O\left(\eta^2\lambda_1\epsilon^{1-a}\frac{H^{-2T}}{\eta(\lambda_1 - \lambda_2)}\right) = O\left(\frac{\eta \lambda_1\epsilon^{1-2a}}{\lambda_1 - \lambda_2}\right) \, .
\]
By applying the above bounds to~\autoref{lem:freedman}, we have
\[
\Pr\left[ \max_{0\leq t\leq T}\left|\sum_{i=t_0 + 1}^{(t_0 + t)\wedge \tau_{v_0}} \frac{A_i}{H^{i-t_0}}\right| \geq  \frac{|v_0|}{2} \right] < \delta
\]
because the deviation term is $O(\sqrt{\frac{\log\frac{1}{\delta}\eta\lambda_1\epsilon^{1-2a}}{\lambda_1 - \lambda_2}}) = O\left( \epsilon^{1-a}\right) \leq \frac{|v_0|}{4}$ and the summation of the conditional expectation terms is $O(\frac{\eta\lambda_1\epsilon^{-\frac{a}{2}}}{\lambda_1 - \lambda_2}) = O\left( \epsilon^{1-a}\right)\leq \frac{|v_0|}{4}$.
By stopping time and~\autoref{lem: local bound}, we have
\[
\sum_{i=t_0 + 1}^{(t_0 + T)\wedge \tau_{t_0}} \frac{B_i}{H^{i-t_0}} \geq - O\left(\eta^2\frac{\epsilon^{-\frac{a}{2}}}{\eta(\lambda_1 - \lambda_2)}\right)\geq - O(\epsilon^{1-\frac{a}{2}}) \geq -\frac{v_0}{2} \, .
\]
By combining both inequalities, we get
\[
\Pr\left[\min_{1\leq t \leq T} \sum_{i=t_0 + 1}^{(t_0 + t)\wedge \tau_{v_0}} \frac{A_i+ B_i}{H^{i-t_0}} \leq v_0\right] < \delta \, .
\]
\end{proof}

We are going to pull out the stopping time $\tau_{t_0}$ in~\autoref{lem:discrete phase 2 stopped concentration}.
The following lemma shows that under a certain \textit{chain} condition, it is possible to pull out the stopping time without introducing additional failure probability.
\begin{lemma}[Pulling out the stopping time using a chain condition]\label{lem:discrete phase 2 stopping time}
Let $\{M_t\}_{t\in\Nz}$ be an adapted stochastic process and $\tau$ be a stopping time. Let $\{M^*_t\}_{t\in\Nz}$ be the maximal process of $\{M_t\}_{t\in\Nz}$ where $M^*_t = \max_{1\leq t'\leq t} M_t$. For any $t\in\N$, $a\in\Real$, and $\delta\in(0,1)$, suppose
\begin{enumerate}
\item $\Pr[M_{t\wedge \tau}^* \geq a]< \delta$ and
\item For any $1\leq t'<t$, $\Pr[\tau \geq t'+1\ |\ M_{t'}^* < a] = 1$.
\end{enumerate}
Then, we have
\[
\Pr[M_{t}^* \geq a]< \delta \, .
\]
\end{lemma}

\begin{proof}[Proof of~\autoref{lem:discrete phase 2 stopping time}]
The key idea is to introduce another stopping time which helps us partition the probability space.
Let $\tau'$ be the stopping time for the event $\{M_{t\wedge \tau}^* \geq a\}$. The following claim shows that if $\tau$ stopped before time $t$, then $\tau'$ should stop earlier than $\tau$.
\begin{claim}\label{claim:discrete phase 2 disjoint}
Let $\tau$ and $\tau'$ be stopping times as defined above. Suppose the conditions in~\autoref{lem:discrete phase 2 stopping time} hold. Then we have
\[
\Pr[\tau < t, \tau' > \tau] = 0 \, .
\]
\end{claim}
\begin{proof}[Proof of~\autoref{claim:discrete phase 2 disjoint}]
The claim can be proved by contradiction as follows. Suppose both $\tau<t$ and $\tau'>\tau$. By the definition of $\tau'$, we know that $M^*_\tau<a$ since $\tau<\tau'$. However, by the second condition of the lemma, we then have
\[
\Pr[\tau \geq \tau+1\ |\ M_\tau^* < a] = 1 \, ,
\]
which is a contradiction. 
\end{proof}

Next, we will show that $\Pr[M^*_t\geq a]\leq\Pr[M^*_{t\wedge\tau}\geq a]$. The idea is partitioning the probability space as follows. We have
\begin{align*}
\Pr[M^*_t \geq a]&= \Pr[M_t^* \geq a,\tau \geq t] + \Pr[M^*_t \geq a,\tau < t,\tau' \leq \tau] + \Pr[M^*_t \geq a,\tau < t,\tau' > \tau] \, . 
\intertext{By~\autoref{claim:discrete phase 2 disjoint}, we have $\Pr[M^*_t \geq a,\tau < t,\tau' > \tau] = 0$. We have}
&= \Pr[M_t^* \geq a,\tau \geq t] + \Pr[M^*_t \geq a,\tau < t,\tau' \leq \tau] \, .
\intertext{For the first term, when $\tau\geq t$, we have $t=t\wedge\tau$ and thus $M^*_t=M^*_{t\wedge\tau}$. As for the second term, when $\tau'\leq\tau<t$, we have both $M^*_t,M^*_{t\wedge\tau} \geq a$. Thus, the equation becomes}
&= \Pr[M^*_{t\wedge \tau} \geq a,\ \tau \geq t] + \Pr[M^*_{t\wedge \tau} \geq a,\ \tau < t,\ \tau' \leq \tau]\\
&\leq \Pr[M^*_{t\wedge \tau} \geq a]\, .
\end{align*}
Thus, we conclude that $\Pr[M^*_t\geq a] \leq \Pr[M^*_{t\wedge \tau} \geq a] <\delta$ as desired.
\end{proof}

By applying the above~\autoref{lem:discrete phase 2 stopping time} on~\autoref{lem:discrete phase 2 stopped concentration}, we can pull out the stopping time and show concentration on the original process in~\autoref{cor: local concentration}.
\localconcentration*
\begin{proof}[Proof of~\autoref{cor: local concentration}]
Let $\tau_{v_0}$ be the stopping time $\lbrace \tilde{\bw}_t < 2v_0 \rbrace$ such that $t > t_0$. We want to apply~\autoref{lem:discrete phase 2 stopping time} with $M_t = -\sum_{i=t_0 + 1}^{t_0 + t} \frac{A_i+ B_i}{H^{i-t_0}}$, $a=-v_0$ and $\tau = \tau_{v_0} - t_0$. First condition is satisfied by~\autoref{lem:discrete phase 2 stopped concentration}. So it is suffice to check that
\[
\Pr\left[\tau_{v_0} \geq t' + t_0 + 1 \middle| \min_{1\leq t \leq t'} \sum_{i=t_0 + 1}^{t_0 + t} \frac{A_i+ B_i}{H^{i-t_0}} \leq v_0\right] = 1 \, .
\]
And indeed we have by~\autoref{cor:discrete phase 2 ODE}
\[
\tilde{\bw}_{t_0 + t'} \geq H^{t'} \cdot\left(\tilde{\bw}_{t_0} + \sum_{i=t_0+1}^{t_0 + t'}\frac{A_i+ B_i}{H^{i-t_0}} \right) > H^{t'} \cdot\left(v_0 +v_0\right)\geq 2v_0 \, .
\]
This implies that $\tau_{v_0} \geq t' + t_0 + 1 $ as desired. 
\end{proof}
\subsection{Interval Analysis}\label{sec: local interval}
Given $\epsilon\in (0, 1)$, let $\tilde{\epsilon} = \frac{\epsilon}{8}$. The goal of this section is to prove the local convergence of Oja's rule (\autoref{thm:discrete phase 2}) with the following interval scheme that shows the improvement of $\tilde{\bw}_t$
\[
-\frac{1}{3}\rightarrow -\tilde{\epsilon}^{1 - \frac{1}{2}}\rightarrow -\tilde{\epsilon}^{1-\frac{1}{4}}\rightarrow \dotsb \rightarrow-\tilde{\epsilon}^{1-\frac{1}{\log\frac{1}{\tilde{\epsilon}}}} \, .
\]
\begin{proof}[Proof of~\autoref{thm:discrete phase 2}]
Let $\tilde{\epsilon} = \frac{\epsilon}{8}$ and $v_0 = -\frac{1}{3}, l = \log\log\frac{1}{\tilde{\epsilon}}$. For $1\leq i\leq l$, choose $T_i\in \N$ such that $\frac{1}{2}\tilde{\epsilon}^{\frac{1}{2^i}}\geq H^{T_i}\geq \frac{1}{4}\tilde{\epsilon}^{\frac{1}{2^i}}$ and $v_i = -\tilde{\epsilon}^{1 - \frac{1}{2^i}}$. Let $S_j = \sum_{i=1}^{j} T_i$ and let $T = S_l$. Notice that by~\autoref{lem:approx log}, we have $v_l = -\frac{\epsilon}{4}$.

We are going to show that for all $1\leq j\leq l$, we have
\begin{equation}\label{eq: local main}
\Pr\left[\tilde{\bw}_{S_j} \leq v_j \middle| \tilde{\bw}_{S_{j-1}} \geq v_{j-1}  \right] < \frac{\delta}{l} \, .
\end{equation}
Then by union bounding over $j$, we have $\Pr\left[\tilde{\bw}_{T} \leq -\frac{\epsilon}{4} \right] < \delta$ and
\[
\frac{1}{4}^{l}\frac{\epsilon}{4}\leq H^{T} \leq \frac{1}{2}^{l}\frac{\epsilon}{4}\Rightarrow T = \Theta\left(\frac{\log\log\frac{1}{\epsilon} + \log\frac{1}{\epsilon}}{\eta(\lambda_1 - \lambda_2)}\right) = \Theta\left(\frac{\log\frac{1}{\epsilon}}{\eta(\lambda_1 - \lambda_2)}\right)
\]
as desired. What remains to be shown is~\autoref{eq: local main}. Now by~\autoref{cor: local concentration}, for $\eta = \Theta\left(\frac{\epsilon(\lambda_1 - \lambda_2)}{\lambda_1\log\frac{\log\log\frac{1}{\epsilon}}{\delta}}\right)$, we have for all $1\leq j\leq l$
\[
\Pr\left[\min_{1\leq t \leq S_j} \sum_{i=S_{j-1} + 1}^{S_{j-1} + t} \frac{A_i+ B_i}{H^{i-S_{j-1}}} \leq v_j \middle| \tilde{\bw}_{S_{j-1}} \geq v_{j-1}  \right] < \frac{\delta}{l} \, .
\]
Now by~\autoref{cor:discrete phase 2 ODE}, the following is true with probability $1-\delta$
\[
\tilde{\bw}_{S_{j}} \geq H^{T_{j}} \cdot\left(\tilde{\bw}_{S_{j-1}} + \sum_{i=S_{j-1}+1}^{S_{j-1}+t}\frac{A_i+ B_i}{H^{i-S_{j-1}}} \right) \geq \frac{1}{2}\tilde{\epsilon}^{\frac{1}{2^j}}\cdot 2v_{j-1}\geq v_j \, .
\]
This shows that
\[
\Pr\left[\bw_{T, 1}^2 \leq 1 - \epsilon \right] \leq\Pr\left[\tilde{\bw}_{T} \leq -\frac{\epsilon}{4} \right] < \delta
\]
as desired.
\end{proof}

\subsection{Continual Learning}\label{sec: finite continual learning}
One of the most remarkable aspects of the biological learning system is its ability to function indefinitely and continuously adapt. In previous sections, we have only been looking at the convergence of Oja's rule at a time point. However, the sensory system needs to function for a long period of time or even for all time. In this section, we explore the capacity of Oja's rule for continual learning as an application of the previous techniques. In~\autoref{thm: finite continual learning}, we show that Oja's rule can maintain its convergence for any finite time while in~\autoref{thm: continual learning}, we show that Oja's rule can maintain its convergence for all time with a slowly diminishing learning rate that scales like $\Omega(\frac{1}{\log t})$. This shows that even if the animal switches to a new environment after a period of time, the learning rate is still large enough to allow efficient continual learning. Notice that the Kushner-Clark theorem requires $\sum_t\eta_t^2 < \infty$ where the learning rate is commonly set as $\eta_t = O(\frac{1}{t})$. In comparison, our slowly diminishing learning rate can achieve $\sum_t\eta_t^2 = \infty$ and thus enables efficient continual learning.

First, we have the following finite continual learning theorem. By applying the diagonal reduction argument in~\autoref{sec:diagonal}, we prove the finite continual learning part of~\autoref{thm:main stay} as a corollary. 
\begin{theorem}[Finite continual learning]\label{thm: finite continual learning}
Let $n, l\in\N$, $\epsilon,\delta\in(0,1)$. Suppose $\bw_{0,1}^2\geq 1-\frac{\epsilon}{2}$. Let
\[
\eta = \Theta\left(\frac{\epsilon(\lambda_1 - \lambda_2)}{\lambda_1\log\frac{l}{\delta}}\right).
\]  Choose $t'$ such that $\frac{1}{4}\geq H^{t'}\geq \frac{1}{8}$. Then
\[
\Pr\left[\exists 1\leq t\leq lt',\ \bw_{t, 1}^2 < 1-\epsilon\right] < \delta \, .
\]
\end{theorem}
\begin{proof}[Proof of~\autoref{thm: finite continual learning}]
Given any $1\leq j\leq l$, by~\autoref{cor: local concentration}, we have
\[
\Pr\left[\min_{1\leq t \leq t'} \sum\nolimits_{i=(j-1)t' + 1}^{jt' + t} \frac{A_i+ B_i}{H^{i-(j-1)t'}} \leq -\frac{\epsilon}{2} \, \middle|\, \tilde{\bw}_{(j-1)t'} \geq -\frac{\epsilon}{2}  \right] < \frac{\delta}{l} \, .
\]
Notice conditioned on $\tilde{\bw}_{(j-1)t'} \geq -\frac{\epsilon}{2}$ and $\min_{1\leq t \leq t'} \sum_{i=(j-1)t' + 1}^{jt' + t} \frac{A_i+ B_i}{H^{i-(j-1)t'}} > -\frac{\epsilon}{2}$, we have for $1\leq t\leq t'$ by~\autoref{cor:discrete phase 2 ODE}
\[
\tilde{\bw}_{(j-1)t' + t} \geq H^{t} \cdot\left(\tilde{\bw}_{(j-1)t'} + \sum\nolimits_{i=(j-1)t'+1}^{(j-1)t'+t}\frac{A_i+ B_i}{H^{i-(j-1)t'}} \right) \geq  H^{t}(-\frac{\epsilon}{2}-\frac{\epsilon}{2} )\geq  -H^{t}\epsilon \, .
\]
In particular, $\tilde{\bw}_{jt'}\geq -\frac{\epsilon}{2}$. This implies that
\[
\Pr\left[\left(\exists 0\leq t\leq t',\ \tilde{\bw}_{(j-1)t' + t} < -\epsilon\right)\cup \left(\tilde{\bw}_{jt'} < -\frac{\epsilon}{2}\right) \ \middle|\ \tilde{\bw}_{(j-1)t'} \geq -\frac{\epsilon}{2} \right] < \frac{\delta}{l} \, .
\]
Union bound over $1\leq j\leq l$, we get 
\[
\Pr\left[\exists 1\leq t\leq lt',\ \bw_{t, 1}^2 < 1-\epsilon\right] < \delta \, 
\]
as desired.
\end{proof}

As a corollary of the above finite continual learning theorem, we can obtain the following for-all-time continual learning theorem. By applying the diagonal reduction argument in~\autoref{sec:diagonal}, we prove the for-all-time continual learning part of~\autoref{thm:main stay} as a corollary. 
\begin{theorem}[For-all-time continual learning]\label{thm: continual learning}
Let $n, t_0\in\N$, $\epsilon,\delta\in(0,1)$. Suppose $\bw_{0,1}^2\geq 1-\frac{\epsilon}{2}$. There is 
\[
\eta_t \geq \Theta\left(\frac{\epsilon(\lambda_1 - \lambda_2)}{\lambda_1\log\frac{t}{\delta}}\right)
\] 
such that
\[
\Pr\left[\exists t\in\N,\ \bw_{t, 1}^2 < 1-\epsilon\right] < \delta \, .
\]
\end{theorem}
\begin{proof}[Proof of~\autoref{thm: continual learning}]
The proof proceeds by recursively choosing $\eta_t$ in intervals and apply~\autoref{thm: finite continual learning} repetitively. Let $\delta_i = \frac{\delta}{2i^2}$. Then notice that $\sum_{i=1}^\infty \delta_i < \delta$. Now apply~\autoref{thm: finite continual learning} with $t_0 = 1$ with failure probability $\delta_1$ to get the corresponding $\eta, t'$ and denote them as $\eta_{(1)}, t'_{(1)}$. Now for $1\leq j\leq  t'_{(1)}$, define $\eta_j = \eta_{(1)}$. By~\autoref{thm: finite continual learning}, this shows that
\[
\Pr\left[\exists 1\leq t\leq t'_{(1)},\ \bw_{t, 1}^2 < 1-\epsilon\right] < \delta \, .
\]
For the $i$th interval, we apply~\autoref{thm: finite continual learning} with $t_0 = 1$ with failure probability $\delta_i$ to get the corresponding $\eta, t'$ and denote them as $\eta_{(i)}, t'_{(i)}$. Now for $t'_{(i-1)}\leq j\leq  t'_{(i)}$, define $\eta_j = \eta_{(i)}$. Notice that the above recursive scheme ensures that $\eta_t \geq \Theta(\frac{\epsilon(\lambda_1 - \lambda_2)}{\lambda_1\log\frac{t}{\delta}})$. And by union bound, we get 
\[
\Pr\left[\exists t\in\N,\ \bw_{t, 1}^2 < 1-\epsilon\right] < \delta \, .
\]
\end{proof}

\section{Global Convergence: Starting From Random Initialization}\label{sec: global convergence}
For the global convergence result, the synaptic weight $\bw_0$ starts from a random initialization. Specifically, we suppose that $\bw_0$ is uniformly sampled from the unit sphere of $\Real^n$. The main theorem in this section states the convergence of Oja's rule starting from a random initialization for the diagonal case. By applying the diagonal reduction argument in~\autoref{sec:diagonal}, we prove the global convergence part of~\autoref{thm:main} as a corollary. 
The following theorem is the main theorem of this section.
\begin{theorem}[Global convergence of the diagonal case]\label{thm:discrete global main}
Suppose $\bw_0$ is uniformly sampled from the unit sphere of $\Real^n$. For any $n\in\N$, $\delta\in(0,1), \epsilon\in (0, \frac{1}{4})$, let
\[
\eta = \Theta\left(\frac{\lambda_1 - \lambda_2}{\lambda_1}\cdot\left(\frac{\epsilon}{\log\frac{\log\frac{n}{\epsilon}}{\delta}} \bigwedge \frac{\delta^2}{\log^2 \frac{\lambda_1 n}{\delta(\lambda_1 - \lambda_2)^2}}\right)\right),\ T = \Theta\left(\frac{\log \frac{1}{\epsilon} + \log\frac{n}{\delta}}{\eta (\lambda_1 - \lambda_2)}\right).
\]
Then
\[
\Pr\left[\bw_{T, 1}^2 < 1-\epsilon\right] < \delta \, .
\]
\end{theorem}

\paragraph{Proof overview and organization.}
The main difficulty in the global convergence is that at the beginning the bounded differences of the noise in the linearization (see~\autoref{lem: global bound}) cannot be controlled directly. To be precise, the \textit{cross term} $|y_t|$ at the worst case needs to be bounded by $O(\sqrt{n}|\bw_{t, 1}|)$. This will introduce a polynomial dependency on $n$, which makes the convergence inefficient. To deal with this issue, in~\autoref{sec: init stopping time}, we provide an initialization lemma and the definition of the stopping time $\xi_{p, \delta}$ that controls the bounded difference of $|y_t|$. Next, in~\autoref{sec: D main skip}, we show that we can control the bounded difference of $|y_t|$ by showing that the stopping time $\xi_{p, \delta}$ is large with high probability in~\autoref{thm: D main}. The details of the proof is delayed to~\autoref{sec: D main}.

Next, in~\autoref{sec:phase 1 linearization} we derive a linearization for $\bw_{t,1}^2$ using a center at 0 instead of 1 based on the intuition from the continuous dynamic in~\autoref{sec:continuous Oja}. Furthermore, we use the ODE trick to write down the dynamic in a closed form with respect to the linearization.

Similar to the local convergence, in~\autoref{sec:phase 1 noise}, we show that the noise from the ODE trick can be controlled with the stopping time and we can pull out the stopping time carefully to bound the original noise. In~\autoref{sec:phase 1 wrap up}, we prove that $\bw_{t, 1}^2$ is greater than $2/3$ efficiently with high probability in an interval analysis in~\autoref{thm: global efficient stopping}. Finally, in~\autoref{sec: combine}, by combining~\autoref{thm: global efficient stopping}, the local convergence~\autoref{thm:discrete phase 2} and the finite continual learning~\autoref{thm: finite continual learning}, we prove the efficient global convergence in~\autoref{thm:discrete global main}.

\subsection{Initialization and the main stopping time}\label{sec: init stopping time}
In this section, we begin with~\autoref{def: auxiliary stopping}, which introduces the auxiliary stochastic processes that we study in~\autoref{sec: D main skip} for controlling $|y_t|$. Then we give an initialization lemma for the auxiliary processes in~\autoref{lem:D init}, which guarantees that the processes perform well with good probability at the first time step. 

\begin{definition}[Auxiliary processes]\label{def: auxiliary stopping}
For each $2\leq j\leq n$, $t\in[T]$, and $\bw\in\Real^n$, define
\[
f_{t,j}(\bw) = \frac{\sum_{i=2}^j\bx_{t,i}\bw_i}{\bw_{1}} \, .
\]
\end{definition}
We cite the initialization lemma in~{\cite[Lemma~5.1]{AL17}} with some straightforward modification below.

\begin{lemma}[Initialization lemma in~{\cite[Lemma~5.1]{AL17}}]\label{lem:D init}
For any $n,T\in\N$, and $\mathcal{D}$ a distribution over unit vectors in $\Real^n$. Let $\bw_{0}\in\Real^n$ be a random unit vector, then for any $j\in[n]$ and $p,\delta\in(0,1)$, there exists
\[
\Lambda_{p, \delta} = \Theta\left(\frac{1}{p}\sqrt{\log\frac{nT}{\delta}}\right), \Lambda'_{p} =  \Theta\left(\frac{n}{p^2}\log\frac{n}{p}\right)
\]
such that let $\mathcal{C}^{p, \delta}_{0}(\bw_0)$ denote the event $\{\exists j\in[n],\ t\in[T],\ \left|f_{t,j}(\bw_0)\right|> \Lambda_{p, \delta}\}$ and let $\mathcal{C}^{p, \delta}_{init}$ denote the event $\lbrace\bw_{0, 1}^2\geq \frac{1}{\Lambda_p'}\rbrace\cap\lbrace \Pr_{\bx_1,\dots,\bx_T\sim\mathcal{D}}[\mathcal{C}^{p, \delta}_{0}(\bw_0) ] < \delta \rbrace$, we have 
\[
\Pr_{\bw_0}\left[ \mathcal{C}^{p, \delta}_{init}\right]\geq1-p-\delta \, .
\]
In particular, by the definition of $\mathcal{C}^{p, \delta}_{init}$ we have $\Pr[\mathcal{C}^{p, \delta}_{0}\ |\ \mathcal{C}^{p, \delta}_{init}] <\delta$. For convenience, we denote $\mathcal{C}^{p, \delta}_{0}(\bw_0)$ as $\mathcal{C}^{p, \delta}_{0}$ when there is no confusion.
\end{lemma}

\subsection{$|y_t|$ is small with high probability}\label{sec: D main skip}
As we said at the beginning of the section, in order to keep the bounded differences of the noise in the global convergence small, we need to make $|y_t|$ small with high probability. To achieve this, we introduce the following stopping time.

\begin{definition}[Stopping times for controlling $|y_t|$]\label{def:stopping time yt}
Given $p,\delta\in(0,1)$ in~\autoref{lem:D init}, we define the stopping time $\psi_{p, \delta}$ to be the first time $t$ such that $\bw_{t, 1}^2 < 1/(2\Lambda_{p, \delta}')$ and the stopping time $\xi_{p, \delta}$ to be the first time $t$ such that $|f_{t,n}(\bw_{(t-1)\wedge \psi_{p,\delta}})|>2\Lambda_{p, \delta}$.
\end{definition}
When there is no confusion, we will abbreviate $\psi_{p,\delta},\xi_{p, \delta}, \Lambda_{p, \delta}, \Lambda_{p}'$ as $\psi, \xi, \Lambda, \Lambda'$ respectively. Intuitively, bounded difference and moments in the noise of the auxiliary process will be small if $\psi$ and $\xi$ are large with high probability.

To control $|y_t|$, we would like to show that show that $\xi_{p, \delta}$ is large with high probability. Specifically, we want to show the following~\autoref{thm: D main}.

\begin{restatable}[$\xi$ is large with high probability]{theorem}{stoppingmain}\label{thm: D main}
Let $T\in \N$ and $p,\delta\in (0, 1)$. Let $\eta = \Theta\left(\frac{(\lambda_1 - \lambda_2)}{\lambda_1\Lambda_{p,\delta/4n^2T}^2\log\frac{nT}{\delta}}\right)$. If we have $T = \Omega\left(\frac{1}{\eta\lambda_1}\right)$ and $p \leq \delta$, then we have
\[
\forall t\in [T],\ \Pr\left[\xi = t\ \middle|\ \mathcal{C}_{init}^{p,\frac{\delta}{4n^2T}}\right] < \frac{\delta}{2n^2T} \, .
\]
In particular we have
\[
\forall t\in [T],\ \Pr\left[\xi = t \ \middle|\ \xi \geq t, \mathcal{C}_{init}^{p,\frac{\delta}{4nT}}\right]  \leq \frac{\delta}{n^2T}
\]
and
\[
\Pr\left[\xi \leq T \ \middle|\ \mathcal{C}_{init}^{p,\frac{\delta}{4nT}}\right] < \frac{\delta}{2n^2} \, .
\]
\end{restatable}
The proof of~\autoref{thm: D main} requires a careful analysis using stopping time and pull-out technique on $(f_{t,2},\dots,f_{t,n})$. We describe the rough strategy as follows and defer the details to~\autoref{sec: D main}. For all $2\leq i\leq n$, we consider the processes $\lbrace f_{t,i}(\bw_s)\rbrace_{s\in [t-1]}$ and, specifically, the shifted stopped process $f_{t,i}(\bw_{s\wedge\psi\star\xi})$ to make sure the bounded differece of the noise is small. Now~\autoref{lem:D init} makes sure that $|f_{t,i}(\bw_{0})|$ is under controlled with high probability and by our framework to analyze stochastic processes we can control $|f_{t,i}(\bw_{t-1\wedge\psi\star\xi})|$ with high probability and therefore show that $\xi$ is large.

Notice that here we use the shifted stopped process in~\autoref{def:star stopped} with stopping time $\xi$. Recall the definition of the shifted stopped process.
\starstopping*
Intuitively, we can consider the shifted stopped process as an original process which moves one step back if it stops. When we consider to apply the stopping time $\xi$, we should use the shifted stopped process instead of the normal stopped process. To given an intuition, we have by~\autoref{lem: stop difference} 
\[
\bw_{t\wedge\xi} - \bw_{(t-1)\wedge\xi} =  \mathbf{1}_{\xi \geq t}\eta y_{t}(\bx_t - y_{t}\bw_{t-1}) \, .
\]
However, $\mathbf{1}_{\xi \geq t}$ only ensures $f_{t-1, n}(\bw_{t-2})$ is bounded and hence $y_{t-1}$ is bounded, but we need $y_t$ to be bounded instead. Instead if we use the shifted stopped process, by~\autoref{lem: star difference}, 
\[
\bw_{t\star\xi} - \bw_{(t-1)\star\xi} =  \mathbf{1}_{\xi > t}\eta y_{t}(\bx_t - y_{t}\bw_{t-1}) \, .
\]
This makes sure that we have a good control on $|y_t|$ because $\xi > t$.

However, the shifted star process creates one caveat when we calculate the moments for the differences term. To be precise, the difference of a shifted star process contains $\mathbf{1}_{\xi > s}$ and $\xi > s$ is not $\mathcal{F}_{s-1}$ measurable anymore. So we need to do an addition conditioning. Observe that when we condition on the event $\xi>s$, the conditional expectation might slightly change. The following lemma shows that this is not significant.

\begin{restatable}[Conditional expectation on $\xi > s$]{lemma}{conditiononxi}\label{lem: conditional stopped help}
Let $T\in \N$, $\xi$ be the stopping time specified before, $\delta'\in(0,0.5)$, and $t\in [T]$. For $s < t$, suppose
\[
\Pr\left[\xi = s \ \middle|\ \xi \geq s, \mathcal{F}_{s-1}^{(t)}\right] < \delta' \, ,
\]
we have
\[
\left|\Exp\left[\bx_{s, i}\bx_{s, j} \ \middle|\  \mathcal{F}_{s-1}^{(t)}, \xi > s\right] - \Exp\left[\bx_{s, i}\bx_{s, j} \ \middle|\  \mathcal{F}_{s-1}^{(t)}\right]\right| < \frac{2\delta'}{1-\delta'}
\]
and furthermore
\[
\left\|\Exp\left[\bx_s\bx_s^\top \ \middle|\  \mathcal{F}_{s-1}^{(t)}, \xi > s\right] - \Exp\left[\bx_s\bx_s^\top \ \middle|\  \mathcal{F}_{s-1}^{(t)}\right]\right\|_2 < \frac{2\delta'}{1-\delta'} \, .
\]
\end{restatable}
\begin{proof}[Proof of~\autoref{lem: conditional stopped help}]
By the laws of total expectation and rearranging the terms we have
\begin{align*}
\Exp[\bx_{s, i}\bx_{s, j}\, |\, \mathcal{F}_{s-1}^{(t)}]&=\Exp[\bx_{s, i}\bx_{s, j}\, |\, \mathcal{F}_{s-1}^{(t)},\ \xi\geq s]\\
&=\Pr[\xi>s\, |\, \mathcal{F}_{s-1}^{(t)}]\cdot\Exp[\bx_{s, i}\bx_{s, j}\, |\, \mathcal{F}_{s-1}^{(t)},\ \xi>s]\\
&+\Pr[\xi=s\, |\, \mathcal{F}_{s-1}^{(t)},\ \xi\geq s]\cdot\Exp[\bx_{s, i}\bx_{s, j}\, |\, \mathcal{F}_{s-1}^{(t)},\ \xi=s]   \, .    
\end{align*}
The first equality is due to the fact that $\{\xi\geq s\}\in\mathcal{F}_{s-1}^{(t)}$. Now, the first inequality are then immediate consequences of the condition of the lemma and the fact that $|\bx_{s,i}\bx_{s,j}|\leq1$ almost surety. The second inequality can be obtained similarly.
\end{proof}

\subsection{Linearization and ODE trick centered at $0$}\label{sec:phase 1 linearization}
In the analysis of the global convergence, we use a linearization with a center at $0$ instead of $1$. The idea is inspired from the analysis of the continuous dynamic as explained in~\autoref{sec:continuous Oja}. However, unlike in the local convergence case, the bounded differences here can only be controlled after applying the stopping time $\xi_{p, \delta}$ from the last section. For the rest of the section, we set a stopping time and an initialization event from~\autoref{sec: D main skip} to be $\xi_T = \xi_{ \delta/4, \delta/8n^2T}$ and $\mathcal{C}_{init}^T = \mathcal{C}_{init}^{\delta/4, \delta/8n^2T}$. In particular by~\autoref{lem:D init} and~\autoref{thm: D main}, we have $\forall t\in [T]$,
\begin{equation}\label{eq: global D stopping init}
\Pr[\mathcal{C}_{init}^T ] \geq 1 - \frac{\delta}{2},\ \Pr[\xi_T < T\ |\ \mathcal{C}^{T}_{init}]<\frac{\delta}{4n^2},\ \Pr[\xi_T = t\, |\, \xi_T\geq t, \mathcal{C}^{T}_{init}]\leq \frac{\delta}{2n^2T} \, .
\end{equation}
We abbreviate the corresponding $\Lambda_{ \delta/4, \delta/8n^2T}, \Lambda_{\delta/4}'$ as $\Lambda, \Lambda'$ for the rest of the section. As in the local convergence, we will use stopping times, $\psi, \xi_T$, to make sure the bounded difference of the noise in the global convergence is small. As explained in the~\autoref{sec: D main skip}, we need to use the shifted stopped process with $\xi_T$. Specifically, we are going to look at $\bw_{t\wedge\psi\star\xi_T, 1}^2$.

We first derive the linearization with a center at $0$.
\begin{lemma}[Linearization at $0$]\label{lem:discrete phase 1 linearization}
Let $\bz_t = \bx_ty_t - y_t^2\bw_{t-1}$. For any $t\in\N$ and $\eta\in(0, 1)$, we have
\[
\bw_{t, 1}^2 \geq  H\cdot\bw_{t-1, 1}^2 + A_{t} + B_{t}
\]
almost surely, where 
\begin{align*}
H &=1+\frac{2}{3}(\lambda_1 - \lambda_2)\eta \, ,\\
A_t &= 2\eta\bz_{t,1}\bw_{t-1, 1} + \eta^2\bz_{t,1}^2 - \Exp\left[2\eta\bz_{t,1}\bw_{t-1, 1}\, |\, \mathcal{F}_{t-1} \right] + 2\eta\lambda_2(1-\|\bw_{t-1}\|^2)\bw_{t-1, 1}^2 \, , \text{ and }\\
B_t &= 2\eta(\lambda_1 - \lambda_2)\bw_{t-1, 1}^2(1 - \bw_{t-1, 1}^2 - \frac{1}{3}) \, .
\end{align*}
\end{lemma}
\begin{proof}[Proof of~\autoref{lem:discrete phase 1 linearization}]
By~\autoref{eq: linear 1}, we have
\begin{align*}
\bw^2_{t,1} &\geq \bw_{t-1,1}^2  + 2\eta(\lambda_1 - \lambda_2)\bw_{t-1, 1}^2(1 - \bw_{t-1, 1}^2) + A_{t} \\
&= \bw_{t-1,1}^2  + H\cdot\bw_{t-1, 1}^2 + A_{t} + B_t
\end{align*}
as desired.
\end{proof}

We apply the ODE trick (see~\autoref{lem:ODE trick}) on~\autoref{lem:discrete phase 1 linearization} and get the following corollary. 
\begin{corollary}[ODE trick]\label{cor:discrete phase 1 ODE}
For any $t_0\in\Nz$, $t\in\N$, and $\eta\in(0,1)$, we have
\[
\bw_{t_0+t, 1}^2 \geq H^{t} \cdot\left(\bw_{t_0, 1}^2 + \sum_{i=t_0+1}^{t_0+t}\frac{A_i+ B_i}{H^{i-t_0}} \right) \, .
\]
\end{corollary}

To control the noise term, we need to have bounds on the bounded differences and the moments of $A_i, B_i$. 

\begin{lemma}\label{lem: global bound}
Let $A_t, B_t$ be defined as in~\autoref{lem:discrete phase 2 linearization}. Let $\eta = O\left(\frac{\lambda_1 - \lambda_2}{\lambda_1\Lambda^2\log\frac{nT}{\delta}}\right)$. If $T = \Omega(\frac{1}{\eta\lambda_1})$, then for any $t\in[T]$ we have $A_t,B_t$ satisfy the following properties:
\begin{itemize}
\item (Bounded difference) $|\mathbf{1}_{\xi_T > t, \psi\geq t}A_{t}|=O\left(\eta\Lambda\bw_{t-1, 1}^2\right)$ almost surely. If $\bw^2_{t-1,1}\leq \frac{2}{3}$, then $B_{t}\geq 0$ almost surely.
\item (Conditional expectation) $\Exp[\mathbf{1}_{\xi_T > t, \psi\geq t}A_{t}\ |\ \mathcal{F}_{t-1}, \mathcal{C}_{init}^T ]=O\left(\lambda_1\eta^2\Lambda^2\bw_{t-1, 1}^2\right)$.
\item (Conditional variance) $\Var\left[\mathbf{1}_{\xi_T > t, \psi\geq t}A_{t}\ |\ \mathcal{F}_{t-1}, \mathcal{C}_{init}^T \right]=O\left(\lambda_1\eta^2\Lambda^2\bw_{t-1, 1}^4\right)$.
\end{itemize}
\end{lemma}
\begin{proof}[Proof of~\autoref{lem: global bound}]
First by the definition of $\xi_T$, we have $|\mathbf{1}_{\xi_T > t,  \psi\geq t}y_{t}| = O(\Lambda|\bw_{t-1, 1}|)$ and $|\mathbf{1}_{\xi_T > t,  \psi\geq t}\bz_{t\star, 1}| = O(\Lambda|\bw_{t-1, 1}|)$. Combining above and~\autoref{lem:oja w norm ub}, we have
\begin{equation*}
|\mathbf{1}_{\xi_T > t,  \psi\geq t}A_{t}| = O\left(\left(\eta\Lambda + \eta^2\Lambda^2 + \eta^2\right)\bw_{t-1, 1}^2\right)  = O(\eta\Lambda\bw_{t-1, 1}^2) \, .
\end{equation*}
And for $\bw_{t-1, 1}^2\leq \frac{2}{3}$, we have $B_{t} \geq 0$ because $1 - \bw_{t-1, 1}^2 - \frac{1}{3} > 0$. For the conditional expectation, we have
\begin{align*}
\Exp\left[\mathbf{1}_{\xi_T > t,  \psi\geq t}\bz_{t, 1}^2\ \middle|\ \mathcal{F}_{t-1}, \mathcal{C}_{init}^T \right] &= \Exp\left[\mathbf{1}_{\xi_T > t,  \psi\geq t}y_t^2(x_{t, 1} - y_t\bw_{t-1, 1})^2\ \middle|\ \mathcal{F}_{t-1}, \mathcal{C}_{init}^T \right] \, .\\
\intertext{By~\autoref{lem: conditional stopped help},~\autoref{thm: D main} and definition of $\xi_T$, we have}
 &\leq O(\lambda_1\Lambda^2\bw_{t-1, 1}^2) \, . 
\end{align*}
Given a random variable $v$, we denote $\Exp[\mathbf{1}_{\xi_T > t,  \psi\geq t}v| \mathcal{F}_{t-1}, \mathcal{C}_{init}^T ] - \Exp[v\, |\, \mathcal{F}_{t-1}]$ as $\bar{v}$. Now we also have
\begin{align*}
\bar{\bz}_{t, 1}&= \bw_{t-1}^\top \overline{\bx_{t}\bx_{t, 1}} -  \bw_{t-1}^\top \overline{\bx_t\bx_t^\top } \bw_{t-1}\bw_{t-1, 1} \, .
\intertext{By applying~\autoref{lem: conditional stopped help} with~\autoref{eq: global D stopping init} and Cauchy-Schwarz, we have}
&= O\left(\|\bw_{t-1}\|_2\frac{\sqrt{n}}{n^2T} + \|\bw_{t-1}\|_2^3\frac{1}{n^2T}\right) \, .
\intertext{Conditioning on $\psi\geq t$ we have $\frac{1}{n} = O(\Lambda^2\bw_{t-1, 1}^2)$ by the definition of $\Lambda, \Lambda'$. We have}
&= O(\eta\lambda_1\Lambda^2\bw_{t-1, 1}^2) \, .
\end{align*}
So combining above we have
\begin{equation*}
\Exp[\mathbf{1}_{\xi_T > t,  \psi\geq t}A_{t}|\mathcal{F}_{t-1}] = O(\eta^2\lambda_1\Lambda^2\bw_{t-1, 1}^2) \, .
\end{equation*}
And similarly applying~\autoref{lem: conditional stopped help}, we obtain that the conditional variance is
\begin{align*}
\Var\left[\mathbf{1}_{\xi_T > t,  \psi\geq t}A_t|\mathcal{F}_{t-1},\mathcal{C}_{init}^T \right] &= O\left(\eta^2\Exp\left[\mathbf{1}_{\xi_T > t,  \psi\geq t}\bz_{t, 1}^2|\mathcal{F}_{t-1},\mathcal{C}_{init}^T \right] \bw_{t-1, 1}^2 + \eta^4\Exp\left[\mathbf{1}_{\xi_T > t,  \psi\geq t}\bz_{t, 1}^4|\mathcal{F}_{t-1},\mathcal{C}_{init}^T \right] \right) \\
&= O(\eta^2\lambda_1\Lambda^2\bw_{t-1, 1}^4)
\end{align*}
as desired.
\end{proof}

\subsection{Concentration of noise in an interval}\label{sec:phase 1 noise}

In this subsection, we want to show that the noise term in~\autoref{cor:discrete phase 1 ODE} is small. As in the local analysis, we are going to use a stopping time to control good bounded differences. Specifically, we have the following stopped concentration follow from~\autoref{lem: global bound}.

\begin{lemma}[Concentration of the stopped noise term in an interval]\label{lem:discrete phase 1 stopped concentration}
Let $t_0, T, t'\in \N, \delta, \delta'\in (0, 1)$ and $a\in (0, \frac{2}{3})$. Suppose $\bw_{t_0\wedge\psi\star\xi_T, 1}^2 \geq \frac{a}{2}$. Let $\tau_{a}$ be the stopping time $\lbrace \bw_{t\wedge\psi\star\xi_T,1}^2 \geq a\rbrace$. Let $\eta = \Theta\left(\frac{(\lambda_1 - \lambda_2)}{\lambda_1\Lambda^2\log\frac{1}{\delta'}}\right)$. If  $\delta' = O(\frac{\delta}{nT})$, $8\geq H^{t'} \geq 4$ and $T = \Omega(\frac{1}{\eta\lambda_1})$, then
\[
\Pr\left[\min_{1\leq t \leq t'} \sum\nolimits_{i=t_0 + 1}^{(t_0 + t)\wedge\psi\star\xi_T\wedge \tau_{a}} \frac{A_i + B_i}{H^{i-t_0}} \leq - \frac{a}{2}\ \middle|\ \mathcal{C}_{init}^T \right] < \delta'.
\]
\end{lemma}
\begin{proof}[Proof of~\autoref{lem:discrete phase 1 stopped concentration}]
For notational convenience, we denote $\mathbf{1}_{\tau_a,\psi\geq t, \xi > t}A_t$ as $\bar{A}_t$. The proof is based on upper bounding the three moment quantities of the stopped noise term and applying martingale concentration inequality (\textit{i.e.,}~\autoref{lem:freedman}). This is very similar to the proof of~\autoref{lem:discrete phase 2 stopped concentration}. To analyze the moment quantities, first recall from~\autoref{lem: stop difference} and~\autoref{lem: star difference} that the martingale difference of the stopped noise term at time $t$ is $\mathbf{1}_{\tau_a,\psi\geq t, \xi > t}\frac{A_t+B_t}{H^{t-t_0}}$. Next, by the definition of $\tau_a$, we could think of $\bw_{t-1,1}^2$ as $\bw_{t-1,1}^2\leq a$. Now by using~\autoref{lem: global bound} to properly bound the moment quantities of $\mathbf{1}_{\tau_a,\psi\geq t, \xi > t}\frac{A_t}{H^{t-t_0}}$ and apply the bounds to~\autoref{lem:freedman}, we have
\[
\Pr\left[\max_{1\leq t \leq t'} \left|\sum_{i=t_0 + 1}^{(t_0 + t)\wedge\psi\star\xi_T\wedge \tau_{a}} \frac{A_i}{H^{i-t_0}}\right| \geq \frac{a}{2} \ \middle|\ \mathcal{C}_{init}^T \right] < \delta'
\]
because the deviation term is $\sqrt{\frac{\log\frac{1}{\delta'}\eta\lambda_1\Lambda^2 a^2}{\lambda_1 - \lambda_2}} = O\left( a\right) \leq \frac{a}{4}$ and the summation of conditional expectation term is $\frac{\lambda_1\eta\Lambda^2a}{\lambda_1 - \lambda_2} = O\left( a\right)\leq \frac{a}{4}$.
By stopping time $\tau_a$ and~\autoref{lem: global bound}, we have
\[
\sum_{i=t_0 + 1}^{(t_0 + t)\wedge\psi\star\xi_T\wedge \tau_{a}} \frac{B_i}{H^{i-t_0}} \geq 0 \, .
\]
Combine the above two inequalities we get the desiring concentration of the stopped noise term.

\end{proof}
Now we will pull out the stopping time $\psi, \tau_a$ and $\xi_T$ together to show that $\bw_{t, 1}^2$ doubles itself efficiently with high probability.
\begin{lemma}[Pull out stopping time in an interval]\label{lem:discrete phase 1 small stopping time}
Let $t_0, T, t'\in \N, \delta, \delta'\in (0, 1)$ and $a\in (0, \frac{2}{3})$. Suppose $\bw_{t_0\wedge\psi\star\xi_T, 1}^2 \geq \frac{a}{2}$. Let $\tau$ be the stopping time of $\lbrace \bw_{t,1}^2 \geq a\rbrace$. Let $\eta = \Theta\left(\frac{(\lambda_1 - \lambda_2)}{\lambda_1\Lambda^2\log\frac{1}{\delta'}}\right)$. If $\delta' = O(\frac{\delta}{nT})$, $8\geq H^{t'} \geq 4$, $t_0 + t'\leq T$ and $T = \Omega(\frac{1}{\eta\lambda_1})$, then
\[
\Pr[\tau > t_0 + t'\, |\, \xi > T, \mathcal{C}_{init}^T ] < \delta' \, .
\]
\end{lemma}
\begin{proof}[Proof of~\autoref{lem:discrete phase 1 small stopping time}]

Let $\tau_{a}$ be the stopping time $\lbrace \bw_{t\wedge\psi\star\xi_T,1}^2 \geq a\rbrace$. Notice that now we only have controls on $\bw_{t\wedge\psi\star\xi_T\wedge\tau_a, 1}^2$ via the moment information in~\autoref{lem: global bound}. To conclude a statement about $\tau$, we need to pull out $\psi, \tau_a, \xi_T$ from $\bw_{t,1}^2$. $\xi_T$ will be pull out by paying union bounds in the conditioning. $\tau_a$ and $\psi$ will be pulled out similar to~\autoref{lem:discrete phase 2 stopping time}.

The main goal is to upper bound the probability of the event $\tau>t_0+t'$. First, observe that this event implies $\tau_a>t_0+t'$ and thus we have
\[
\Pr[\tau > t_0 + t'\, |\, \xi > T, \mathcal{C}_{init}^T ] \leq\Pr[\tau_a > t_0 + t'\, |\, \xi > T, \mathcal{C}_{init}^T ] \, .
\]
Next, we further partition the probability space with the event $\psi>t_0+t'$ and its complement. Namely,
\begin{align*}
\Pr[\tau_a > t_0 + t'\, |\, \xi > T, \mathcal{C}_{init}^T ]&\leq\Pr[\tau_a > t_0 + t'\, ,\ \psi>t_0+t'\, |\, \xi > T, \mathcal{C}_{init}^T ]\\
&+\Pr[\tau_a > t_0 + t'\, ,\ \psi\leq t_0+t'\, |\, \xi > T, \mathcal{C}_{init}^T ]\\
&=\circled{A}+\circled{B}\, .
\end{align*}
In the following, we show that both \circled{A} and \circled{B} are at most $\delta'/2$ and thus complete the proof. Let us start with defining two error events on the (stopped) noise term.
\[
\mathcal{A} := \left\{\min_{1\leq t \leq t'} \sum_{i=t_0 + 1}^{(t_0 + t)\wedge\psi\star\xi\wedge \tau_{a}} \frac{A_i + B_i}{H^{i-t_0}} \leq - \frac{a}{2}\right\} \text{ and }\mathcal{B} :=\left\{\min_{1\leq t \leq t'} \sum_{i=t_0 + 1}^{t_0 + t} \frac{A_i + B_i}{H^{i-t_0}} \leq - \frac{a}{2}\right\} \, .
\]
Now, we partition the probability space again using $\mathcal{A},\mathcal{B}$ and their complements as follows.
\begin{align*}
\circled{A}&=\Pr\left[\tau_a , \psi > t_0 + t', \mathcal{A}  \ \middle|\ \xi_T > T, \mathcal{C}_{init}^T \right]  + \Pr\left[\tau_a, \psi > t_0 + t', \neg\mathcal{A}  \ \middle|\ \xi_T > T , \mathcal{C}_{init}^T \right] \, ,\\
\circled{B}&=\Pr\left[\tau_a  > t_0 + t',\psi \leq t_0 + t', \mathcal{B} \middle| \xi_T > T,  \mathcal{C}_{init}^T \right]+ \Pr\left[\tau_a  > t_0 + t',\psi \leq t_0 + t',  \neg\mathcal{B}  \middle| \xi_T > T,  \mathcal{C}_{init}^T \right] \, .
\end{align*}
We show that for both \circled{A} and \circled{B}, the first term is small and the second term is $0$.

\begin{claim}[The first term is small]\label{claim:global pull out 1}
\[
\Pr[\mathcal{A}\, |\, \mathcal{C}_{init}^T]<\frac{\delta'}{4} \text{ and } \Pr\left[ \mathcal{B}, \tau_a > t_0 + t', \xi_T > T \ \middle|\ \mathcal{C}_{init}^T \right] < \frac{\delta'}{4} \, .
\]
\end{claim}
\begin{proof}[Proof of~\autoref{claim:global pull out 1}]
Note that by~\autoref{lem:discrete phase 2 stopped concentration} we know that $\mathcal{A}$ happens with small probability when conditioning on $\mathcal{C}_{init}^T$ as desired.

As for $\mathcal{B}$, we need to apply the pull-out lemma (see~\autoref{lem:discrete phase 2 stopping time}) as follows. It suffices to check that if $\neg\mathcal{B}$ is true, then we have $\psi > t_0 + t$. By~\autoref{cor:discrete phase 1 ODE} we have for all $t\in [t']$
\[
\bw_{t_0 + t,1}^2 \geq H^{t}\left( \bw_{t_0,1}^2 + \sum_{i=t_0 + 1}^{t_0 + t} \frac{A_i + B_i}{H^{i-t_0}}\right)\geq H^{t}(a - \frac{a}{2}) \geq \frac{a}{2}
\]
as desired. By~\autoref{lem:discrete phase 2 stopping time}, this shows that $\Pr\left[ \mathcal{B}, \tau_a > t_0 + t', \xi_T > T \ \middle|\ \mathcal{C}_{init}^T \right] < \delta'/4$ as desired.
\end{proof}

\begin{claim}[The second term is $0$]\label{claim:global pull out 2}
\[
\Pr[\tau_a > t_0 + t' , \xi_T > T, \psi > t_0 + t', \neg\mathcal{A}] = 0 \text{ and }\Pr[\tau_a > t_0 + t' , \xi_T > T, \psi \leq t_0 + t', \neg\mathcal{B}] = 0 \, .
\]
\end{claim}
\begin{proof}[Proof of~\autoref{claim:global pull out 2}]
For the first equality, for the sake of contradiction assuming all the events $\tau_a > t_0 + t', \xi_T > T, \psi> t_0 + t'$ and $\neg\mathcal{A}$ are happening. By~\autoref{cor:discrete phase 1 ODE}, we have
\begin{align*}
\bw_{(t_0 + t')\wedge\psi\star\xi,1}^2 &\geq H^{t'\wedge(\psi - t_0)\star(\xi - t_0)}\left( \bw_{t_0,1}^2 + \sum_{i=t_0 + 1}^{(t_0 + t)\wedge\psi\star\xi} \frac{A_i + B_i}{H^{i-t_0}}\right)\\
&=H^{t'}\left( \bw_{t_0,1}^2 + \sum_{i=t_0 + 1}^{(t_0 + t)\wedge\psi\star\xi\wedge\tau_a} \frac{A_i + B_i}{H^{i-t_0}}\right)\geq 4(a - \frac{a}{2}) = 2a
\end{align*}
which contradicts to $\tau_a > t_0 + t'$. Thus, these events cannot happen simultaneously.

Similarly, for the second equality, notice that when all the events $\tau_a > t_0 + t', \xi_T > T, \tau\leq t_0 + t'$ and $\neg\mathcal{B}$ are happening, by the second condition of~\autoref{lem:discrete phase 2 stopping time} we checked in~\autoref{claim:global pull out 1}, we have $\tau > t_0 + t'$ which contradicts to $\tau\leq t_0 + t'$.
\end{proof}

To wrap up, by~\autoref{claim:global pull out 2} we know that the second term of \circled{A} vanishes and thus
\[
\circled{A}=\Pr\left[\tau_a , \psi > t_0 + t', \mathcal{A}  \ \middle|\ \xi_T > T, \mathcal{C}_{init}^T \right]\leq\Pr\left[\mathcal{A}  \ \middle|\ \xi_T > T, \mathcal{C}_{init}^T \right]\leq\frac{\Pr[\mathcal{A}\, |\, \mathcal{C}_{init}^T]}{\Pr[\xi_T>T\, |\, \mathcal{C}_{init}^T]} \, .
\]
As $\Pr[\mathcal{A}\, |\, \mathcal{C}_{init}^T]\leq\delta'/4$ by~\autoref{claim:global pull out 1} and $\Pr\left[\xi_T > T \ \middle|\  \mathcal{C}_{init}^T \right] \geq \frac{1}{2}$ from~\autoref{eq: global D stopping init}, we have $\circled{A}\leq\delta'/2$ as desired. Similarly, as the second term of \circled{B} vanishes by~\autoref{claim:global pull out 2}, we have
\[
\circled{B}=\Pr\left[\tau_a  > t_0 + t',\psi \leq t_0 + t', \mathcal{B}  \middle| \xi_T > T,  \mathcal{C}_{init}^T \right]= \frac{\Pr\left[\tau_a  > t_0 + t', \xi_T > T,\psi \leq t_0 + t', \mathcal{B}  \middle|  \mathcal{C}_{init}^T \right] }{\Pr[\xi_T > T\, |\,  \mathcal{C}_{init}^T ]} \, .
\]
As $\Pr\left[ \mathcal{B}, \tau_a > t_0 + t', \xi_T > T \ \middle|\ \mathcal{C}_{init}^T \right]<\delta'/4$ by~\autoref{claim:global pull out 1} and $\Pr\left[\xi_T > T \ \middle|\  \mathcal{C}_{init}^T \right] \geq \frac{1}{2}$ from~\autoref{eq: global D stopping init}, we have $\circled{B}\leq\delta'/2$ as desired. In conclusion, we have
\[
\Pr[\tau > t_0 + t'\, |\, \xi > T, \mathcal{C}_{init}^T ]\leq\circled{A}+\circled{B}<\delta' \, .
\]

\end{proof}

\subsection{Interval Analysis: From Global to Local}\label{sec:phase 1 wrap up}
In this section, we proceed with the following interval scheme to show the improvement of $\bw_{t, 1}^2$ from the global regime to the local regime
\[
\frac{1}{\Lambda'} \rightarrow 2\frac{1}{\Lambda'} \rightarrow \dotsb \rightarrow 2^{\lfloor \log \frac{2\Lambda'}{3} \rfloor}\frac{1}{\Lambda'} \rightarrow \frac{2}{3} \, .
\]
We first show in~\autoref{lem: param global} on how to choose the learning rate without dependency on $T$ and then show that $\bw_{t, 1}^2$ is going to reach $2/3$ efficiently.

\begin{lemma}[Choice of parameters in global convergence]\label{lem: param global}
Given $t'$ such that $8\geq H^{t'}\geq 4$, there exists 
\[
T =\Theta\left(\frac{\lambda_1\log\frac{n}{\delta}\log^2 \frac{n\lambda_1}{\delta(\lambda_1 - \lambda_2)^2}}{\delta^2(\lambda_1 - \lambda_2)^2}\right)
\]
such that
\[
\eta = \Theta\left(\frac{\lambda_1 - \lambda_2}{\lambda_1\Lambda_T^2 \log \frac{nT}{\delta}}\right),\ T \geq t'\log\Lambda'.
\]
\end{lemma}
\begin{proof}[Proof of~\autoref{lem: param global}]
Since $8\geq H^{t'}\geq 4$, we have that $t' = \Theta(1/\eta(\lambda_1 - \lambda_2))$. Now
\[
t'\log\Lambda' = \Theta\left(\frac{\lambda_1\log\Lambda'\log^2\frac{nT}{\delta}}{\delta^2(\lambda_1 - \lambda_2)^2}\right) \, .
\]
For notational convenience we let $A = \frac{\lambda_1\log\Lambda'}{\delta^2(\lambda_1 - \lambda_2)^2}$. Then we need $T\geq A\log^2\frac{nT}{\delta}$ and 
\[
T = \Theta(A \log^2 nA) = \Theta\left(\frac{\lambda_1\log\frac{n}{\delta}\log^2 \frac{n\lambda_1}{\delta(\lambda_1 - \lambda_2)^2}}{\delta^2(\lambda_1 - \lambda_2)^2}\right)
\]
satisfied the requirement as desired.
\end{proof}

\begin{theorem}[From global to local]\label{thm: global efficient stopping}
Let $n\in\N, \epsilon,\delta\in(0, 1)$. Let $T=\Theta\left(\frac{\lambda_1\log\frac{n}{\delta}\log^2 \frac{n\lambda_1}{\delta(\lambda_1 - \lambda_2)^2}}{\delta^2(\lambda_1 - \lambda_2)^2}\right)$. Let $\tau$ be the stopping time of $\bw_{t, 1}^2 \geq \frac{2}{3}$. Let 
\[
\eta = O\left(\frac{(\lambda_1 - \lambda_2)}{\lambda_1\Lambda^2\log\frac{nT}{\delta}}\right), T_0 = \frac{\lceil \log \frac{2}{3\Lambda'} \rceil + 1}{\eta(\lambda_1 - \lambda_2)} \, .
\]
Then we have 
\[
\Pr\left[\tau > T_0\right] < \delta \, .
\]
\end{theorem}
\begin{proof}[Proof of~\autoref{thm: global efficient stopping}]
Choose $t'$ such that $8\geq H^{t'}\geq 4$ and let $m = \lceil \log \frac{2}{3\Lambda'} \rceil + 1$, $v_i = \frac{1}{\Lambda'}2^i$ for $i = 0,\dotsb, m-1$ and $v_{m} = \frac{2}{3}$. Let $\tau_{v_i}$ be the stopping time of $\lbrace \bw_{t, 1}^2 \geq v_i \rbrace$ and let $T_0 = mt'$. We will apply~\autoref{lem:discrete phase 1 small stopping time} with $\delta' = \frac{\delta}{4m}$. Notice that since $\log \Lambda' = O(nT)$, we can choose $\delta'$ this way. Now we have
\begin{align*}
\Pr[\tau > T_0] &= \Pr[\tau_{v_m} > mt']\\
&\leq \Pr[\tau_{v_m} > mt'\, |\, \xi_T > T, \mathcal{C}_{init}^T ] + \Pr[\xi \leq T\, |\, \mathcal{C}_{init}^T ] +  \Pr[\neg\mathcal{C}_{init}^T ] \, .
\intertext{By~\autoref{eq: global D stopping init} and union bound, we have}
&< \sum_{i=1}^m\Pr[\tau_{v_i}>it', \tau_{v_{i-1}}\leq (i-1)t'\, |\, \xi_T > T, \mathcal{C}_{init}^T ] + \frac{\delta}{4n^2} + \frac{\delta}{2}\\
&\leq \sum_{i=1}^m\Pr[\tau_{v_i}>\tau_{v_{i-1}} + t'\, |\, \xi_T > T, \mathcal{C}_{init}^T ] + \frac{\delta}{4n^2} + \frac{\delta}{2} \, .
\intertext{By~\autoref{lem:discrete phase 1 small stopping time}, each summand can be bounded by $\frac{\delta}{4m}$, we have}
&< \frac{\delta m}{4m} + \frac{\delta}{4n^2} + \frac{\delta}{2}\leq \delta
\end{align*}
as desired.
\end{proof}
\subsection{Combining Theorem \ref{thm: global efficient stopping} with the local analysis}\label{sec: combine}
In this section, since we have shown that $\bw_{t, 1}^2$ efficiently reaches $2/3$ in~\autoref{thm: global efficient stopping}, by combining~\autoref{thm: global efficient stopping}, the local convergence (\autoref{thm:discrete phase 2}) and the finite continual learning (\autoref{thm: finite continual learning}), we derive~\autoref{thm:discrete global main}. 
\begin{proof}[Proof of~\autoref{thm:discrete global main}]
Let $\tau$ to be the hitting time of $\bw_{t, 1}^2 > 1-\frac{\epsilon}{2}$. With \[
\eta = \Theta\left(\frac{\lambda_1 - \lambda_2}{\lambda_1}\cdot\left(\frac{\epsilon}{\log\frac{\log\frac{n}{\epsilon}}{\delta}} \bigwedge \frac{\delta^2}{\log^2 \frac{\lambda_1 n}{\delta(\lambda_1 - \lambda_2)^2}}\right)\right) \, ,
\]
we can apply~\autoref{thm: global efficient stopping},~\autoref{thm:discrete phase 2} to get that
\[
\Pr[\tau > T] < \frac{\delta}{2}
\]
where $T = \Theta(\frac{\log\frac{1}{\epsilon} + \log\Lambda'}{\eta(\lambda_1 - \lambda_2)}) =\Theta(\frac{\log\frac{1}{\epsilon} + \log\frac{n}{\delta}}{\eta(\lambda_1 - \lambda_2)})$. Now we initialize~\autoref{thm: finite continual learning} with $t_0 = \Theta(\log\frac{1}{\epsilon} + \log\frac{n}{\delta})$ with failure probability $\frac{\delta}{2}$ to get
\[
\Pr[\exists 1\leq t\leq T, \bw^2_{\tau + t, 1} < 1-\epsilon] < \frac{\delta}{2} \, .
\]
Since $T\in [\tau, \tau + T]$ if $\tau \leq T$, now by union bounding two inequalities, we have
\[
\Pr[\bw^2_{T, 1} < 1-\epsilon] < \delta \, .
\]
\end{proof}

\section{The Cross Term is Small}\label{sec: D main}

In this section, we will prove~\autoref{thm: D main} to finish the proof of the global convergence. We recall the motivation again. In order to keep the bounded differences of the noise in the global convergence small, we need to make $f_{t, n}(\bw_{t-1})$ small with high probability. Concretely, the stopping time $\xi_{p,\delta}$ stops whenever $|f_{t,n}(\bw_{(t-1)\wedge\psi_{p,\delta}})|>2\Lambda_{p,\delta}$ (see~\autoref{def:stopping time yt}). Therefore, the main goal of this section is to show that $\xi_{p, \delta}$ is large with high probability in~\autoref{thm: D main}. 

\stoppingmain*

In order to prove~\autoref{thm: D main}, we consider the auxiliary processes $(f_{t,2},\dots,f_{t,n})$ (see~\autoref{def: auxiliary stopping}) as a vector and use a vector linearization as well as the ODE trick in~\autoref{cor:D linearization vector} and~\autoref{cor:D ODE}. We then obtain the concentration on the stopped processes in~\autoref{lem:D stopped concentration}. Finally, to prove the main theorem by induction, we prove the induction step in~\autoref{lem: D main} by carefully pulling out the stopping time to finish the proof.





To bound the stopping time $\xi$, we need to show the concentration of $f_{t,j}(\bw_{(t-1)\wedge\psi})$ and as before the linearization and the ODE trick would be our main tools. 

\subsection{Linearization and ODE trick for controlling $|y_t|$}
Let us start with the linearization and the ODE trick for function $f_{t,j}$ in this subsection.

\begin{lemma}[Linearization]\label{lem:D linearization}
Let $t\in[T]$, $s\in[t-1]$. Let $\bw_s=\bw_{s-1}+\eta\bz_s$ where $\bz_s=y_s(\bx_s-y_s\bw_{s-1})$. Then there exists $\overline{\bw}_{s-1}=\bw_{s-1}+c\eta\bz_s$ for some $c\in[0,1]$ such that for all $j,\, 2\leq j\leq n$,
\[
f_{t,j}(\bw_s) = (1-\eta(\lambda_1-\lambda_j))f_{t,j}(\bw_{s-1}) + \eta \sum_{i=2}^{j-1}(\lambda_i-\lambda_{i+1})f_{t,i}(\bw_{s-1}) +  A_{s,j}^{(t)} 
\]
where
\[
A_{s,j}^{(t)} = \eta \nabla f_{t,j}(\bw_{s-1})^\top \left(\bz_{s}-\Exp[\bz_{s}\ |\ \mathcal{F}_{s-1}]\right)  + \eta^2 \bz_{s}^\top \nabla^2 f_{t,j}(\overline{\bw}_{s-1})\bz_{s} \, .
\]
\end{lemma}
\begin{proof}[Proof of~\autoref{lem:D linearization}]
This is a direct application of Taylor expansion. Concretely, there exists $\overline{\bw}_{s-1}=\bw_{s-1}+c\eta\bz_s$ for some $c\in[0,1]$ such that 
\begin{align*}
f_{t,j}(\bw_s) &= f_{t,j}(\bw_{s-1}) + \eta\nabla f_{t,j}(\bw_{s-1})^\top \bz_{s}  + \eta^2 \bz_{s}^\top \nabla^2 f_{t,j}(\overline{\bw}_{s-1})\bz_{s} \, .
\intertext{Note that $\frac{\partial f_{t,j}(\bw)}{\partial\bw_1}=-\frac{f_{t,j}(\bw)}{\bw_{1}}$ and $\frac{\partial f_{t,j}(\bw)}{\partial\bw_i}=\frac{\mathbf{1}_{i\leq j}\cdot\bx_{t,i}}{\bw_1}$ for $i=2,\dots,n$.We have}
&=f_{t,j}(\bw_{s-1})-\eta\frac{f_{t,j}(\bw_{s-1})}{\bw_{s-1,1}}\cdot\bz_{s,1} + \eta\frac{\sum_{i=2}^j\bx_{s,i}\bz_{s,i}}{\bw_{s-1,1}} + \eta^2 \bz_{s}^\top \nabla^2 f_{t,j}(\overline{\bw}_{s-1})\bz_{s} \, .
\intertext{Next, recall that $\Exp[\bz_{s,i}\ |\ \mathcal{F}_{s-1}]=(\lambda_i-\bw_{s-1}^\top\diag(\lambda)\bw_{s-1})\cdot\bw_{s-1,i}$. By adding and subtracting the expectations, the equation becomes}
&=f_{t,j}(\bw_{s-1}) -\eta\lambda_1f_{t,j}(\bw_{s-1}) + \eta\frac{\sum_{i=2}^j\lambda_i\bx_{s,i}\bw_{s-1,i}}{\bw_{s-1,1}}\\
&+\eta\left(\bw_{s-1}^\top\diag(\lambda)\bw_{s-1}\right)\cdot\left(f_{t,j}(\bw_{s-1}) -\frac{\sum_{i=2}^j\bx_{s,i}\bw_{s-1,i}}{\bw_{s-1,1}}\right) + A_{s, j}^{(t)} \, .
\intertext{Observe that the two terms in the parenthesis becomes $0$ after cancelling out with each other. Finally, by adding and subtracting $\eta\lambda_if_{t,i}(\bw_{s-1})$ for each $i=2,3,\dots,j$, we have}
&= \left(1-\eta(\lambda_1-\lambda_j)\right)\cdot f_{t,j}(\bw_{s-1}) + \eta\sum_{i=2}^{j-1}(\lambda_i-\lambda_{i+1})f_{t,i}(\bw_{s-1}) + A_{s, j}^{(t)} 
\end{align*}
as desired.
\end{proof}

We can write the above lemma in a vector form. For any $t\in[T]$, let $\mathbf{f}_t(\bw),\bA_{s}^{(t)}\in\Real^{n-1}$ be $(n-1)$-dimensional vectors where the $i^\text{th}$ coordinates of them are $f_{t,i+1}(\bw)$, $A_{s,i+1}^{(t)}$ respectively. The following is an immediate corollary of~\autoref{lem:D linearization} by rewriting everything into a vector form.

\begin{corollary}[Linearization in a vector form]\label{cor:D linearization vector}
For any $t\in[T]$ and $s\in[t-1]$, we have
\[
\bff_t(\bw_s) = H\mathbf{f}_t(\bw_{s-1})+\bA_{s}^{(t)}
\]
where
\[
H=\begin{pmatrix}
1-\eta(\lambda_1-\lambda_2)&0&0&\cdots&0\\
\eta(\lambda_2-\lambda_3)&1-\eta(\lambda_1-\lambda_3)&0&\cdots&0\\
\eta(\lambda_2-\lambda_3)&\eta(\lambda_3-\lambda_4)&1-\eta(\lambda_1-\lambda_4)&\cdots&0\\
\vdots&\vdots&\vdots&\ddots&\vdots\\
\eta(\lambda_2-\lambda_3)&\eta(\lambda_3-\lambda_4)&\eta(\lambda_4-\lambda_5)&\cdots&1-\eta(\lambda_1-\lambda_n)
\end{pmatrix} \, .
\]
\end{corollary}
By the ODE trick for vector (see~\autoref{lem:ODE trick matrix}), we immediately have the following corollary for a closed form solution to $\bff_{t}(\bw_s)$.
\begin{corollary}[ODE trick]\label{cor:D ODE}
For any $t\in[T]$,  $s\in[t -1]$, we have
\[
\bff_{t}(\bw_{s})=H^s\bff_{t}(\bw_{0}) + \sum_{s'= 1}^{s}H^{s - s'}\bA_{s'}^{(t)} \, .
\]
\end{corollary}

\subsection{Concentration of the noise term}\label{sec:D concentration}
We want to control the noise term in~\autoref{cor:D ODE}. However, same as the situation before, we cannot get the concentration for the noise terms of the ODE trick directly. As a consequence, we have to introduce a new stopping time $\tau_t$ to make sure the bounded difference of the stopped processes are small enough for the martingale concentration inequality.

For a fixed $t\in[T]$, we define a stopping time $\tau_t$ for the noise terms from $s=1,2,\dots,t-1$ as follows.
First, we work on a slightly different filtration $\{\mathcal{F}_s^{(t)}\}_{s\in[t-1]}$ than the natural filtration $\{\mathcal{F}_s\}_{s\in[t-1]}$. The key idea is that the stopping time can depend on $\bx_t$ since we only look at the noise term up to $t-1$. Concretely, for each $s\in[t-1]$, let $\mathcal{F}_{s}^{(t)}$ be the $\sigma$-algebra generated by $\{\bx_1,\bx_2,\dots,\bx_s\}\cup\{\bx_t\}$. Note that $\{\mathcal{F}_s^{(t)}\}_{s\in[t-1]}$ is well-defined and $\{A_{t,s,j}\}_{s\in[t-1]}$ is an adapted random process with respect to $\{\mathcal{F}_s^{(t)}\}_{s\in[t-1]}$, \textit{i.e.,} $A_{t,s,j}$ lies in $\mathcal{F}_s^{(t)}$ for all $s\in[t-1]$.
Also, note that $\Exp[\bz_s\ |\ \mathcal{F}_{s-1}]=\Exp[\bz_s\ |\ \mathcal{F}_{s-1}^{(t)}]$. That is, the conditional expectation and conditional variance of $\bz$ are the same with respect to $\{\mathcal{F}_s\}$ and $\{\mathcal{F}_s^{(t)}\}$. The following is the definition of the stopping time $\tau_t$.

\begin{definition}[Stopping time for $\bff_t$]\label{def:stopping time ft}
Let $T\in\Nz,p,\delta\in(0,1)$. For every $t\in[T]$, let  $\Lambda_{p,\delta}$ be the parameter specified in~\autoref{lem:D init} and $\xi,\psi$ be the stopping times specified in~\autoref{def:stopping time yt}. Define $\tau_t$ to be the stopping time for the first $s$ such that $\{\|\bff_t(\bw_{s\wedge\psi\star\xi})\|_\infty>2\Lambda_{p, \delta} \}$.
\end{definition}

The goal of this subsection is to show the concentration of the stopped noise vector as follows.

\begin{lemma}[Concentration for the stopped noise term]\label{lem:D stopped concentration}
Let $T\in \N_{\geq 0}, p,\delta, \delta'\in(0, 1), t\in [T]$. Let $\Lambda_{p,\delta}$ be the parameter specified before and $\xi,\psi,\tau_t$ be the stopping times as chosen before.
Let $\eta = \Theta\left(\frac{(\lambda_1 - \lambda_2)}{\lambda_1\Lambda_{p,\delta}^2\log\frac{1}{\delta'}}\right)$. If $T = \Omega(\frac{1}{\eta\lambda_1}), p \leq\delta$ and the following condition is true
\begin{equation*}
\forall 1\leq t' \leq t-1,\ \Pr[\xi = t'\, |\, \xi \geq t', \mathcal{C}_{init}^{p,\delta}] \leq \frac{1}{n^2T} \, ,
\end{equation*}
then for all $\overline{s}\in[t-1]$,
\[
\Pr\left[\exists i\in [n-1],\ \sum\nolimits_{s=1}^{\overline{s}\wedge\psi\star\xi_{p, \delta}\wedge\tau_t}\left(H^{\overline{s}-s}\bA_{t,s}\right)_i\geq \Lambda_{p,\delta}\ |\  \mathcal{C}_{init}^{p,\delta} \right] < n\delta' \, .
\]
\end{lemma}

To enable martingale concentration, we have to bound the three moment quantities of the martingale difference of $M_{t,s}\in\Real^{n-1}$ where $M_{t,s,j}$ is the $j^\text{th}$ entry of $\sum_{s'= 1}^{s\wedge\psi\star\xi\wedge\tau_t}H^{\overline{s} - s'}\bA_{t,s'}$ for every $s\in[\overline{s}]$ and $j\in[n-1]$.
By~\autoref{lem: stop difference} and~\autoref{lem: star difference}, the difference can be rewritten as
\[
M_{t,s}-M_{t,s-1}=\mathbf{1}_{\tau_t \geq s\wedge\psi\star\xi, \xi > s\wedge \psi, \psi \geq s}H^{\overline{s} - s}\bA_{t, s}=\mathbf{1}_{\tau_t,\psi \geq s, \xi > s}H^{\overline{s} - s}\bA_{t, s} \, .
\]

Before we bound the bounded differences and the moments for $\mathbf{1}_{\tau_t, \psi \geq s, \xi > s}H^{\overline{s}-s}\bA_{t,s}$, recall that in~\autoref{lem: conditional stopped help} we showed that the conditional expectation of $\bx_{s,i}\bx_{s,j}$ would not change by too much when conditioning on the event $\xi>s$.
Let us start with bounding the three moment quantities of the stopped process of $\bA_{t,s}$ in~\autoref{lem:D diff var} and then extend to that of $H^{\overline{s}-s}\bA_{t,s}$ in~\autoref{lem:D diff var vector}.
\begin{lemma}[Structure of the stopped $\bA_{t,s}$]\label{lem:D diff var}
Let $T\in\N,\eta\in (0, 1), t\in [T]$ and $s\in[t-1]$. Let $\Lambda$ be the parameter specified before and $\xi,\psi,\tau_t$ be the stopping times as chosen before. 
If $\eta = O\left(\frac{1}{\Lambda}\right), T = \Omega(\frac{1}{\eta\lambda_1}),p \leq\delta$ and the following condition holds
\begin{equation}
\forall 1\leq t' \leq t-1,\ \Pr[\xi = t'\, |\, \xi \geq t', \mathcal{C}_{init}^{p,\delta}] \leq \frac{1}{n^2T}
\end{equation}
then the following holds almost surely.
\begin{align*}
\intertext{$\bullet$ (Bounded difference) We have}
&\left\|\mathbf{1}_{\tau_t,\psi \geq s, \xi > s}\bA_{s}^{(t)}\right\|_{\infty}=O(\eta\Lambda^2) \, .
\intertext{$\bullet$ (Conditional expectation) We have}
&\left\|\Exp[\mathbf{1}_{\tau_t,\psi \geq s, \xi > s}\bA_{s}^{(t)}\ |\ \mathcal{F}_{s-1}^{(t)},  \mathcal{C}_{init}^{p,\delta}]\right\|_{\infty}=O(\eta^2\lambda_1\Lambda^3) \, .
\intertext{$\bullet$ (Conditional variance) We have}
&\left\|\Exp[\mathbf{1}_{\tau_t,\psi \geq s, \xi > s}\bA_{s}^{(t)}{\bA_{s}^{(t)}}^\top \ |\ \mathcal{F}_{s-1}^{(t)},  \mathcal{C}_{init}^{p,\delta}]\right\|_{max}=O(\eta^2\lambda_1\Lambda^4) \, .
\end{align*}
where the $\|\cdot\|_{max}$ is the entrywise maximum of a matrix. 
\end{lemma}
\begin{proof}[Proof of~\autoref{lem:D diff var}]
The proof is basically direct verification using the definition of stopping time and~\autoref{lem: conditional stopped help}. We postpone the proof to~\autoref{sec: annoying diff var}.
\end{proof}

Now note that given $\overline{s}\in [t-1]$ the stopped process $\left\{\sum_{s'= 1}^{s\wedge\psi\star\xi\wedge\tau_t}H^{\overline{s} - s'}\bA_{t,s'}\right\}_{s\in[\overline{s}]}$ is an adapted stochastic process with respect to $\{\mathcal{F}_{s}^{(t)}\}_{s\in[\bar{s}]}$. Furthermore, it has small bounded difference and  moments. Concretely we have the following.

\begin{lemma}[Structure of the stopped $H^{\overline{s}-s}\bA_{t,s}$]\label{lem:D diff var vector}
Let $T\in \N,\eta,\delta\in (0, 1), t\in [T], \bar{s}\in [t-1]$. Let $\Lambda$ be the parameter specified before and $\xi,\tau_t$ be the stopping times as chosen before.
For any $s\in[\bar{s}]$ and $j\in[n-1]$, let $M_{t,s,j}$ be the $j^\text{th}$ entry of $\sum_{s'= 1}^{s\wedge\psi\star\xi\wedge\tau_t}H^{\overline{s} - s'}\bA_{t,s'}$. If $\eta = O\left(\frac{1}{\Lambda}\right), T = \Omega\left(\frac{1}{\eta\lambda_1}\right),p \leq\delta$ and the following condition is true
\begin{equation*}
\forall 1\leq t' \leq t-1,\ \Pr[\xi = t'\, |\, \xi \geq t', \mathcal{C}_{init}^{p,\delta}] \leq \frac{1}{n^2T} \, ,
\end{equation*}
then the following holds.
\begin{align*}
\intertext{$\bullet$ (Bounded difference) For any $j\in[n-1]$, we have}
&\max_{s\in [\bar{s}]}\left|M_{t,s,j}-M_{t,s-1,j}\right|=O(\eta\Lambda^2 )\, \text{almost surely}.
\intertext{$\bullet$ (Conditional expectation) For any $j\in[n-1]$, we have}
&\sum_{s=1}^{\overline{s}}\Exp\left[M_{t,s,j}-M_{t,s-1,j}\ \middle| \ \mathcal{F}_{s-1}^{(t)},  \mathcal{C}_{init}^{p,\delta}\right]=O\left(\frac{\eta\lambda_1\Lambda^3}{\lambda_1-\lambda_2}\right) \, .
\intertext{$\bullet$ (Conditional variance) For any $j\in[n-1]$, we have}
&\sum_{s=1}^{\overline{s}}\Var\left[M_{t,s,j} -  M_{t,s-1,j}\ \middle| \ \mathcal{F}_{s-1}^{(t)},  \mathcal{C}_{init}^{p,\delta}\right]=O\left(\frac{\eta\lambda_1\Lambda^4}{\lambda_1-\lambda_2}\right) \, .
\end{align*}

\end{lemma}

\begin{proof}[Proof of~\autoref{lem:D diff var vector}]
For notational convenience, given a matrix $A$, we will denote its $j^\text{th}$ row as $A_{(j)}$ for the rest of the proof. First, notice that the multiplier matrix $H=VDV^{-1}$ is invertible where
\[
V=\begin{pmatrix}
1&0&0&\cdots&0\\
1&1&0&\cdots&0\\
1&1&1&\cdots&0\\
\vdots&\vdots&\vdots&\ddots&\vdots\\
1&1&1&\cdots&1
\end{pmatrix}\ \ ,\ \ V^{-1}=\begin{pmatrix}
1&0&0&\cdots&0\\
-1&1&0&\cdots&0\\
0&-1&1&\cdots&0\\
\vdots&\vdots&\vdots&\ddots&\vdots\\
0&0&0&\cdots&1
\end{pmatrix}
\]
and $D=\diag(d_1,d_2,\dots,d_{n-1})$ such that $d_i=1-\eta(\lambda_1-\lambda_{i+1})$ for every $i=1,2,\dots,n-1$. To see the above, observe that for any diagonal matrix $D'=\diag(d_1',d_2',\dots,d_{n-1}')$, we have
\[
VD'V^{-1} = \begin{pmatrix}
d_1'&0&0&\cdots&0\\
d_1'-d_2'&d_2'&0&\cdots&0\\
d_1'-d_2'&d_2'-d_3'&d_3'&\cdots&0\\
\vdots&\vdots&\vdots&\ddots&\vdots\\
d_1'-d_2'&d_2'-d_3'&d_3'-d_4'&\cdots&d_{n-1}'
\end{pmatrix} \, .
\]
Note that if $d_1'\geq d_2'\geq\cdots\geq d_{n-1}'\geq 0$, then we have
\begin{equation}\label{eq: l1 norm matrix}
\|\left(VD'V^{-1}\right)_{(i)}\|_1 = d_i' + \sum_{j = 1}^{i-1}d_j' - d_{j+1}' = d_1' \, .
\end{equation}
\begin{itemize}
\item (Bounded difference) Fixed $j\in [n]$. First we have for all $s\in [\bar{s}]$,
\[
|M_{t, s, j} - M_{t, s-1, j}| = |\mathbf{1}_{\tau_t ,\psi\geq s, \xi > s}H^{\overline{s} - s}\bA_{t, s}|   \leq O(\eta\Lambda^2)
\] 
by~\autoref{eq: l1 norm matrix}.

\item (Conditional expectation) Similarly, we have
\begin{align*}
\Exp\left[M_{t,s,j}-M_{t,s-1,j}\ \middle| \ \mathcal{F}_{s-1}^{(t)},  \mathcal{C}_{init}^{p,\delta}\right] &= \Exp\left[\mathbf{1}_{\tau_t ,\psi\geq s, \xi > s}H^{\overline{s} - s}\bA_{t, s}\ \middle| \ \mathcal{F}_{s-1}^{(t)},  \mathcal{C}_{init}^{p,\delta}\right]\\
&\leq \left\|H^{\overline{s} - s}_{(j)}\right\|_1\left\|\Exp\left[\mathbf{1}_{\tau_t ,\psi\geq s, \xi > s}\bA_{t, s} \ \middle| \ \mathcal{F}_{s-1}^{(t)},  \mathcal{C}_{init}^{p,\delta}\right]\right\|_\infty \, .
\intertext{By~\autoref{eq: l1 norm matrix} and~\autoref{lem:D diff var}, we have}
&\leq (1-\eta(\lambda_1 - \lambda_2))^{\overline{s} - s}\cdot O(\eta^2\lambda_1\Lambda^3)   \, .
\end{align*}
So by geometric series, we have
\[
\sum_{s=1}^{\overline{s}}\Exp\left[M_{t,s,j}-M_{t,s-1,j}\ \middle| \ \mathcal{F}_{s-1}^{(t)},  \mathcal{C}_{init}^{p,\delta}\right]=O\left(\frac{\eta\lambda_1\Lambda^3}{\lambda_1-\lambda_2}\right).
\]

\item (Conditional variance) Similarly, we have
\begin{align*}
&\Var\left[M_{t,s,j} -  M_{t,s-1,j}\ \middle| \ \mathcal{F}_{s-1}^{(t)},  \mathcal{C}_{init}^{p,\delta}\right]\\
=&\ \Exp\left[\mathbf{1}_{\tau_t ,\psi\geq s, \xi > s}(H^{\overline{s} - s}\bA_{t, s})^2\ \middle| \ \mathcal{F}_{s-1}^{(t)}, \mathcal{C}_{init}^{p,\delta}\right]\\
=&\ (H^{\overline{s} - s}_{(j)})^\top \Exp\left[\mathbf{1}_{\tau_t ,\psi\geq s, \xi > s}\bA_{t, s}\bA_{t, s}^\top \ \middle| \ \mathcal{F}_{s-1}^{(t)},  \mathcal{C}_{init}^{p,\delta}\right](H^{\overline{s} - s}_{(j)})\\
\leq&\ \left\|H^{\overline{s} - s}_{(j)}\right\|_1\cdot\left\|\Exp\left[\mathbf{1}_{\tau_t ,\psi\geq s, \xi > s}\bA_{t, s}\bA_{t, s}^\top \ \middle| \ \mathcal{F}_{s-1}^{(t)},  \mathcal{C}_{init}^{p,\delta}\right]\right\|_{max}\cdot\left\|H^{\overline{s} - s}_{(j)}\right\|_1 \, .
\intertext{By~\autoref{eq: l1 norm matrix} and~\autoref{lem:D diff var}, we have}
&\leq  (1-\eta(\lambda_1 - \lambda_2))^{2(\overline{s} - s)}\cdot O(\eta^2\lambda_1\Lambda^4) \, .
\end{align*}
So by geometric series, we have
\[
\sum_{s=1}^{\overline{s}}\Var\left[M_{t,s,j} -  M_{t,s-1,j}\ \middle| \ \mathcal{F}_{s-1}^{(t)},  \mathcal{C}_{init}^{p,\delta}\right]=O\left(\frac{\eta\lambda_1\Lambda^4}{\lambda_1-\lambda_2}\right) \, .
\]
\end{itemize}
\end{proof}

As a consequence of~\autoref{lem:D diff var vector}, we are able to prove the concentration for the stopped noise term.

\begin{proof}[Proof of~\autoref{lem:D stopped concentration}]
The proof is based on applying the corollary of Freedman's inequality (see~\autoref{cor:freedman}) on each coordinate using~\autoref{lem:D diff var vector}. We have
\[
\Pr\left[\sum\nolimits_{s=1}^{\bar{s}\wedge\psi\star\xi_{p, \delta}\wedge\tau_t}\left(H^{\overline{s}-s}\bA_{t,s}\right)_i\geq \Lambda_{p, \delta}\ |\  \mathcal{C}_{init}^{p,\delta}\right] < \delta' \, .
\]
by noticing that the deviation term is $O(\sqrt{\frac{\eta\lambda_1\Lambda_{p, \delta}^4\log\frac{1}{\delta'}}{\lambda_1-\lambda_2}}) <\frac{\Lambda_{p, \delta}}{2}$ and the sum of conditional expectation term is $O(\frac{\eta\lambda_1\Lambda_{p, \delta}^3}{\lambda_1-\lambda_2}) <\frac{\Lambda_{p, \delta}}{2}$. Now we obtain the desired inequality by union bounding over $i\in [n-1]$.
\end{proof}

\subsection{Wrap up}\label{sec:D wrap up}
First fix $\delta, \delta'$ in the~\autoref{lem:D stopped concentration} as $\frac{\delta}{4n^2T}, \frac{\delta}{4n^3T^2}$ respectively. The following lemma proves the inductive step toward the main theorem.

\begin{lemma}[Inductive step]\label{lem: D main}
Let $T\in \N_{\geq 0}, p,\delta\in(0, 1)$ be the parameters and $\xi$ be the stopping times as chosen before. Let $\eta = \Theta\left(\frac{(\lambda_1 - \lambda_2)}{\lambda_1\Lambda_{p,\delta/4n^2T}^2\log\frac{nT}{\delta}}\right)$. If $T = \Omega(\frac{1}{\eta\lambda_1})$, $t\in[T]$, $p\leq \delta$ and for every $t'\in[t-1]$, 
\[
\Pr\left[\xi = t'\, \middle|\, \mathcal{C}_{init}^{p,\frac{\delta}{4n^2T}}\right] < \frac{\delta}{2n^2T} \, ,
\]
then
\[
\Pr\left[\xi = t\, \middle|\, \mathcal{C}_{init}^{p,\frac{\delta}{4n^2T}}\right] < \frac{\delta}{2n^2T} \, .
\]
\end{lemma}
\begin{proof}[Proof of~\autoref{lem: D main}]
Fix $T\in\N$ such that $T=\Omega(\frac{1}{\eta\lambda_1})$ and $t\in[T]$, we have
\begin{align*}
\Pr\left[\xi = t\, \middle|\, \mathcal{C}_{init}^{p,\frac{\delta}{4n^2T}}\right]&= \Pr\left[|f_{t, n}(\bw_{(t-1)\wedge\psi\star\xi})| > 2\Lambda, \xi\geq t\, \middle|\, \mathcal{C}_{init}^{p,\frac{\delta}{4n^2T}}\right]\\
&\leq \Pr\left[\tau_t < t \, \middle|\,  \mathcal{C}_{init}^{p,\frac{\delta}{4n^2T}}\right]
\end{align*}
where the last inequality is due to~\autoref{def:stopping time ft}. Intuitively, when the noise term of the linearization for $\bff_t$ is small, then we expect $\tau_t$ would not stop by the ODE trick. Formally, denote the event where the noise term is large as $\mathcal{A}_{s_0} = \lbrace \exists i\in [n], \sum_{s=1}^{s_0\wedge\psi\star\xi}\left(H^{s_0-s}\bA_{t,s}\right)_i\geq \Lambda\rbrace$ and denote its stopped version as $\mathcal{A}_{s_0}^{\tau_t} = \lbrace \exists i\in [n], \sum_{s=1}^{s_0\wedge\psi\star\xi\wedge\tau_t}\left(H^{s_0-s}\bA_{t,s}\right)_i\geq \Lambda\rbrace$. Recall from~\autoref{lem:D init} that $\mathcal{C}_0^{p,\frac{\delta}{4n^2T}}$ is the event
\[
\left\lbrace\exists j\in[n],\ t\in[T],\ \left|f_{t,j}(\bw_0)\right|> \Lambda_{p, \delta}\right\rbrace \, .
\]
Now, we partition the probability space and rewrite the error probability as follows.
\begin{align*}
\Pr\left[\tau_t < t \ \middle| \ \mathcal{C}^{p,\frac{\delta}{4n^2T}}_{init}\right] & \leq \Pr\left[\tau_t < t,  \neg\mathcal{C}_0^{p,\frac{\delta}{4n^2T}},\neg\mathcal{A}_{t-1} \ \middle| \ \mathcal{C}^{p,\frac{\delta}{4n^2T}}_{init}\right] + \Pr\left[\mathcal{A}_{t-1}\cup\mathcal{C}_0^{p,\frac{\delta}{4n^2T}}\ \middle| \ \mathcal{C}^{p,\frac{\delta}{4n^2T}}_{init} \right]\\
&=\circled{A}+\circled{B} \, .
\end{align*}
Notice that the concentration analysis in the previous subsection only works for the \textit{stopped version}, \textit{i.e.,} $\mathcal{A}_{s_0}^{\tau_t}$. To go from $\mathcal{A}_{s_0}$ to $\mathcal{A}_{s_0}^{\tau_t}$, observe that the events $\neg\mathcal{C}_0^{p,\frac{\delta}{4n^2T}}$ and $\neg\mathcal{A}_{s_0}$ imply the noise term to be small and thus the ODE trick (see~\autoref{cor:D ODE}) gives the following useful equality for every $1\leq s_0\leq t-1$
\begin{equation}\label{eq: D stopping time}
\Pr\left[ \tau_t \geq s_0 + 1 \ \middle| \ \neg{\mathcal{C}}_0^{p,\frac{\delta}{4n^2T}},\neg \mathcal{A}_{s_0}\right] = 1 \, .
\end{equation}

From~\autoref{eq: D stopping time}, we immediately have $\circled{A}=0$. As for \circled{B}, by the chain rule of expectation, we have
\begin{align*}
\circled{B} &= \sum_{s = 1}^{t-1}\Pr\left[\mathcal{A}_{s}, \neg\mathcal{A}_{s-1}, \neg\mathcal{C}^{p,\frac{\delta}{4n^2T}}_0 \ \middle| \  \mathcal{C}^{p,\frac{\delta}{4n^2T}}_{init}  \right] + \Pr\left[\mathcal{C}^{p,\frac{\delta}{4n^2T}}_0 \ \middle| \  \mathcal{C}^{p,\frac{\delta}{4n^2T}}_{init}  \right] \, .
\end{align*}
Notice that the last term of the above equation can be upper bounded by $\delta/(4n^2T)$ due to~\autoref{lem:D init}. As for each term in the summation, by~\autoref{eq: D stopping time}, we have
\[
\Pr\left[\mathcal{A}_{s}, \neg\mathcal{A}_{s-1}, \neg\mathcal{C}^{p,\frac{\delta}{4n^2T}}_0 \ \middle| \  \mathcal{C}^{p,\frac{\delta}{4n^2T}}_{init}  \right]=\Pr\left[\mathcal{A}_{s}^{\tau_t}, \neg\mathcal{A}_{s-1}, \neg\mathcal{C}^{p,\frac{\delta}{4n^2T}}_0 \ \middle| \  \mathcal{C}^{p,\frac{\delta}{4n^2T}}_{init}  \right]\leq\Pr\left[\mathcal{A}_{s}^{\tau_t}\ \middle| \  \mathcal{C}^{p,\frac{\delta}{4n^2T}}_{init}  \right] \, .
\]
Finally, to invoke the concentration we proved in the previous subsection (\textit{i.e.}~\autoref{lem:D stopped concentration}), we have to verify the condition as follows. For every $1\leq t'\leq t-1$,
\[
\Pr[\xi = t'\, |\, \xi \geq t', \mathcal{C}_{init}^{p,\frac{\delta}{4n^2T}}] \leq \frac{\Pr[\xi = t'\, |\, \mathcal{C}_{init}^{p,\frac{\delta}{2nT}}]}{1-\Pr[\xi < t'\, |\, \mathcal{C}_{init}^{p,\frac{\delta}{4n^2T}}]} < \frac{\frac{\delta}{4n^2T}}{1- \frac{(t-1)\delta}{2nT}} \leq \frac{1}{n^2T}
\]
where the inequalities hold due to the induction hypothesis in the lemma statement. As a result,~\autoref{lem:D stopped concentration} gives $\Pr[\mathcal{A}_{s}^{\tau_t}\, | \,  \mathcal{C}^{p,\frac{\delta}{4n^2T}}_{init}  ]<(t-1)n\delta/(4n^3T^2)$. This gives us
\[
\circled{B}\leq\sum_{s=1}^{t-1}\frac{(t-1)n\delta}{4n^3T^2}+\frac{\delta}{4n^2T}<\frac{\delta}{2n^2T} \, .
\]
Recall that~\autoref{eq: D stopping time} implies $\circled{A}=0$, we conclude that
\[
\Pr\left[\xi = t\, \middle|\, \mathcal{C}_{init}^{p,\frac{\delta}{4n^2T}}\right]\leq\circled{A}+\circled{B}<\frac{\delta}{2n^2T}
\]
as desired.

\end{proof}

Now the main theorem can be derived as a corollary.
\stoppingmain*
\begin{proof}[Proof of~\autoref{thm: D main}]
The proof proceed by induction. For the base case, we have
\[
\Pr[\xi = 1\, |\, \mathcal{C}_{init}^{p,\frac{\delta}{4n^2T}}] = \Pr\left[ |f_{1,j}(\bw_0)|> 2\Lambda \ \middle|\  \mathcal{C}_{init}^{p,\frac{\delta}{4n^2T}}\right] < \frac{\delta}{2n^2T} \, .
\]
The induction step is exactly~\autoref{lem: D main} and this gives us the first conclusion. The second conclusion can be obtained from union bounding over $T$.
\end{proof}


\section{Conclusions and Future Directions}

In this work, our contributions are two-fold. First, we give the first convergence rate analysis for the biological Oja's rule in solving streaming PCA. In particular, the rate we show is nearly optimal. This provides theoretical evidences for a biologically-plausible streaming PCA mechanism to converge in a biologically-realistic time scale in the retina-optical nerve pathway. Second, we introduce a novel one-shot framework to analyze stochastic dynamics. The framework is simple and captures the intrinsic behaviors of the dynamics. Finally, as a byproduct, the convergence rate we get for the biological Oja's rule outperforms the state-of-the-art upper bound for streaming PCA (using ML Oja's rule).

In this section, we discuss the biological and algorithmic significance of our results and point out potential future directions.

\subsection{Biological aspects}
\paragraph{Spiking Oja's Rule.} In this paper, we simplify the biological dynamic using a rate-based model. It would be interesting to design a spiking version of the learning rule to solve streaming PCA. On the other hand, it has been shown that Spike Timing Dependent Plasticity (STDP) has self-normalizing behaviors \cite{Abbott2000}, so the higher-order terms in biological Oja's rule might not be needed for the normalization in the spiking version.

\paragraph{Convergence rate analysis for other biological-plausible learning rules.}
As mentioned in~\autoref{sec:related works}, there are plenty of Hebbian-type learning rules that had been proposed to solve some computational problems~\cite{S77,BCM82,Sanger89,XOS92,OBL00,Aparin12,PHC15}. Nevertheless, most of them do not have an efficiency guarantee and we think it would be of interest to use our frameworks to systematically analyze the convergence rates of these update rules. This is not only a natural theoretical question but also could potentially provide insights on how these biologically-plausible algorithms are different from standard algorithms.
\paragraph{Convergence rate analysis for biologically-plausible learning rules for online $k$-PCA}
In this work, we focus on biological Oja's rule in finding the top eigenvector of the covariance matrix. It is a natural question to ask: whether there is a biologically-plausible algorithm for finding top $k$ eigenvectors (a.k.a. the $k$-PCA problem)? In the setting of ML Oja's rule, this can be achieved by \textit{QR decomposition}~\cite{AL17}. As mentioned in~\autoref{sec:related works}, computational neuroscientists have proposed several variants of biological Oja's rule to solve streaming $k$-PCA~\cite{Oja92,Sanger89,Foldiak1989,Leen1991,Rubner1989,KDT94,PHC15}. Some networks use feedforward connections only but the learning rules are not local~\cite{Oja92,Sanger89} while some use Hebbian learning on the feedforward connection and use anti-Hebbian learning on the recurrent connection to decorrelate the outputs~\cite{Foldiak1989,Leen1991,Rubner1989,KDT94,PHC15}. However, there is no convergence rate analysis for these networks and even the results on the global convergence in the limit are not known for most of these networks. Therefore, it will be interesting to apply our framework to derive a convergence rate analysis for these biologically-plausible learning rules to solve online $k$-PCA.

\subsection{Algorithmic aspects}
\paragraph{Improving the guarantees for biological Oja's rule.} In this paper, we mainly focus on the situation when $\lambda_1 > \lambda_2$ while some of the previous works also considered the gap-free setting. We believe our framework can be easily extended to the gap-free setting and leave it as future work. Also, there are some logarithmic terms (e.q. additive $\log\log\log(1/\epsilon)$ in the local convergence) in the convergence rate and do not seem to be inherent. It would be interesting to find out the optimal logarithmic dependency.

On the other hand, we suspect the $\log(1/\epsilon)$ term in the convergence rate of biological Oja's rule might be necessary. Thus, showing a lower bound with $\log(1/\epsilon)$ would be of great interest. Note that there exists (non-streaming) algorithm which solves PCA using only $O\left(\lambda_1\epsilon^{-1}\gap^{-2}\right)$ samples so the lower bound should be tailored to the dynamic.

\paragraph{Tighter analysis for ML Oja's rule.} 
Using the objective function from \cite{AL17}, one can also easily generalize our framework to ML Oja's rule and tighten the bounds for both the local and global convergence rates. 

\paragraph{Other Stochastic Dynamics.} There are many stochastic optimization problems in machine learning where the optimal analysis still remains elusive, \textit{e.g.,} stochastic gradient dynamics of matrix completion, low-rank approximation, nonnegative matrix factorization, etc. It is of great interest to apply our \textit{one-shot} framework to analyze other important stochastic dynamics.

\subsection*{Acknowledgements}
We thank Nancy Lynch for valuable comments on the presentation of this paper, thank Kai-Min Chung for helpful discussions, and thank Rohit Agrawal, Boaz Barak, Chun-Hsiang Chan, Yanlin Chen, and Santhoshini Velusamy for comments on the draft of this paper. We thank Nancy Lynch again for organizing a brain algorithm reading group at MIT in Spring 2019 and all the participants for the inspiring conversation.

\bibliography{mybib}

\newcommand{\etalchar}[1]{$^{#1}$}
\newcommand{\noopsort}[1]{} \newcommand{\printfirst}[2]{#1}
  \newcommand{\singleletter}[1]{#1} \newcommand{\switchargs}[2]{#2#1}
\begin{thebibliography}{HMRAR13}

\bibitem[ACS13]{ACS13}
Raman Arora, Andy Cotter, and Nati Srebro.
\newblock Stochastic optimization of pca with capped msg.
\newblock In {\em Advances in Neural Information Processing Systems}, pages
  1815--1823, 2013.

\bibitem[AN00]{Abbott2000}
Larry~F. Abbott and Sacha~B. Nelson.
\newblock {Synaptic plasticity: taming the beast}.
\newblock {\em Nature Neuroscience}, 3:1178--1183, 2000.

\bibitem[Apa12]{Aparin12}
Vladimir Aparin.
\newblock Simple modification of oja rule limits $l_1$-norm of weight vector
  and leads to sparse connectivity.
\newblock {\em Neural computation}, 24(3):724--743, 2012.

\bibitem[AR90]{Atick1990}
Joseph~J. Atick and A.~Norman Redlich.
\newblock Towards a theory of early visual processing.
\newblock {\em Neural Computation}, 2:308--320, 1990.

\bibitem[AR92]{Atick1992}
Joseph~J. Atick and A.~Norman Redlich.
\newblock What does the retina know about natural scenes?
\newblock {\em Neural Computation}, 4:196--210, 1992.

\bibitem[AZL17]{AL17}
Zeyuan Allen-Zhu and Yuanzhi Li.
\newblock First efficient convergence for streaming k-pca: a global, gap-free,
  and near-optimal rate.
\newblock In {\em 2017 IEEE 58th Annual Symposium on Foundations of Computer
  Science (FOCS)}, pages 487--492. IEEE, 2017.

\bibitem[Azu67]{Azuma67}
Kazuoki Azuma.
\newblock Weighted sums of certain dependent random variables.
\newblock {\em Tohoku Mathematical Journal, Second Series}, 19(3):357--367,
  1967.

\bibitem[BCM82]{BCM82}
Elie~L. Bienenstock, Leon~N. Cooper, and Paul~W. Munro.
\newblock Theory for the development of neuron selectivity: orientation
  specificity and binocular interaction in visual cortex.
\newblock {\em Journal of Neuroscience}, 2(1):32--48, 1982.

\bibitem[BDWY16]{BDWY16}
Maria-Florina Balcan, Simon~Shaolei Du, Yining Wang, and Adams~Wei Yu.
\newblock An improved gap-dependency analysis of the noisy power method.
\newblock In {\em Conference on Learning Theory}, pages 284--309, 2016.

\bibitem[BM02]{Baccus2002}
Stephen~A. Baccus and Markus Meister.
\newblock Fast and slow contrast adaptation in retinal circuitry.
\newblock {\em Neuron}, 36:909--919, 2002.

\bibitem[CCL19]{CCL19}
Chi{-}Ning Chou, Kai{-}Min Chung, and Chi{-}Jen Lu.
\newblock On the algorithmic power of spiking neural networks.
\newblock In {\em 10th Innovations in Theoretical Computer Science Conference,
  {ITCS} 2019, January 10-12, 2019, San Diego, California, {USA}}, pages
  26:1--26:20, 2019.

\bibitem[CG90]{CG90}
Pierre Comon and Gene~H Golub.
\newblock Tracking a few extreme singular values and vectors in signal
  processing.
\newblock {\em Proceedings of the IEEE}, 78(8):1327--1343, 1990.

\bibitem[CKS96]{CKS96}
Andrzej Cichocki, Wlodzimierz Kasprzak, and Wladyslaw Skarbek.
\newblock Adaptive learning algorithm for principal component analysis with
  partial data.
\newblock {\em Cybernetics and Systems Research}, pages 1014--1019, 1996.

\bibitem[CL94]{CL94}
Hong Chen and Ruey-Wen Lin.
\newblock An online unsupervised learning machine for adaptive feature
  extraction.
\newblock {\em IEEE Transactions on Circuits and Systems II: Analog and Digital
  Signal Processing}, 41(2):87--98, 1994.

\bibitem[CWY20]{Chou2020}
Chi-Ning Chou, Mien~Brabeeba Wang, and Tiancheng Yu.
\newblock A general framework for analyzing stochastic dynamics in learning
  algorithms.
\newblock {\em arXiv: 2006.06171}, 2020.

\bibitem[DK96]{DK96}
Konstantinos~I. Diamantaras and Sun~Yuan Kung.
\newblock {\em Principal component neural networks: theory and applications}.
\newblock John Wiley \& Sons, Inc., 1996.

\bibitem[DSOR15]{DOR15}
Christopher De~Sa, Kunle Olukotun, and Christopher R{\'e}.
\newblock Global convergence of stochastic gradient descent for some non-convex
  matrix problems.
\newblock In {\em Proceedings of the 32Nd International Conference on
  International Conference on Machine Learning - Volume 37}, ICML'15, pages
  2332--2341. JMLR.org, 2015.

\bibitem[Duf13]{Duflo13}
Marie Duflo.
\newblock {\em Random iterative models}, volume~34.
\newblock Springer Science \& Business Media, 2013.

\bibitem[Fre75]{F75}
David~A. Freedman.
\newblock On tail probabilities for martingales.
\newblock {\em The Annals of Probability}, pages 100--118, 1975.

\bibitem[Fö89]{Foldiak1989}
Peter Földiák.
\newblock Adaptive network for optimal linear feature extraction.
\newblock {\em International 1989 Joint Conference on Neural Networks}, pages
  401--405, 1989.

\bibitem[GS12]{Ganguli2012}
Surya Ganguli and Haim Sompolinsky.
\newblock {Compressed sensing, sparsity and neural data}.
\newblock {\em Annual review of neuroscience}, 35(1):463--483, 2012.

\bibitem[HBM05]{Hosoya2005}
Toshihiko Hosoya, Stephen~A. Baccus, and Markus Meister.
\newblock Dynamic predictive coding by the retina.
\newblock {\em Nature}, 436:71--77, 2005.

\bibitem[Heb49]{Hebb49}
Donald~O. Hebb.
\newblock {\em The organization of behavior: {A} neuropsychological theory}.
\newblock Wiley, New York, June 1949.

\bibitem[HK92]{Hornik1992}
Kurt Hornik and Chung-Ming Kuan.
\newblock Convergence analysis of local feature extraction algorithms.
\newblock {\em Neural Networks}, 5:229--240, 1992.

\bibitem[HKP91]{HKP91}
John Hertz, Anders Krogh, and Richard~G. Palmer.
\newblock Introduction to the theory of neural computation.
\newblock {\em Santa Fe Institute Studies in the Sciences of Complexity;
  Lecture Notes, Redwood City, Ca.: Addison-Wesley, 1991}, 1991.

\bibitem[HMRAR13]{HMRR13}
Michel Habib, Colin McDiarmid, Jorge Ramirez-Alfonsin, and Bruce Reed.
\newblock {\em Probabilistic methods for algorithmic discrete mathematics},
  volume~16.
\newblock Springer Science \& Business Media, 2013.

\bibitem[Hot33]{Hotelling33}
Harold Hotelling.
\newblock Analysis of a complex of statistical variables into principal
  components.
\newblock {\em Journal of educational psychology}, 24(6):417, 1933.

\bibitem[HP94]{HP94}
George~F. Harpur and Richard~W. Prager.
\newblock {\em Experiments with simple Hebbian-based learning rules in pattern
  classification tasks}.
\newblock Citeseer, 1994.

\bibitem[HP14]{HP14}
Moritz Hardt and Eric Price.
\newblock The noisy power method: A meta algorithm with applications.
\newblock In {\em Advances in Neural Information Processing Systems}, pages
  2861--2869, 2014.

\bibitem[HP19]{Merav2019}
Yael Hitron and Merav Parter.
\newblock Counting to ten with two fingers: Compressed counting with spiking
  neurons.
\newblock {\em 27th Annual European Symposium on Algorithms}, 2019.

\bibitem[HvH98]{Hemmen1998}
M.~Haft and J.~Leo van Hemmen.
\newblock {Theory and implementation of infomax filters for the retina}.
\newblock {\em Network: Computation in Neural Systems}, 9(1):39--71, 1998.

\bibitem[JJK{\etalchar{+}}16]{JJKNS16}
Prateek Jain, Chi Jin, Sham~M Kakade, Praneeth Netrapalli, and Aaron Sidford.
\newblock Streaming pca: Matching matrix bernstein and near-optimal finite
  sample guarantees for oja’s algorithm.
\newblock In {\em Conference on learning theory}, pages 1147--1164, 2016.

\bibitem[Kar96]{Karayiannis96}
Nicolaos~B. Karayiannis.
\newblock Accelerating the training of feedforward neural networks using
  generalized hebbian rules for initializing the internal representations.
\newblock {\em IEEE transactions on neural networks}, 7(2):419--426, 1996.

\bibitem[KC78]{Kushner1978}
Harold.~J. Kushner and Dean~S. Clark.
\newblock {\em Stochastic approximataon for constrained and unconstrained
  systems}.
\newblock Springer, Berlin, 1978.

\bibitem[KDT94]{KDT94}
Sun-Yuan Kung, Konstantinos~I. Diamantaras, and Jin-Shiuh Taur.
\newblock Adaptive principal component extraction (apex) and applications.
\newblock {\em IEEE Transactions on Signal Processing}, 42(5):1202--1217, 1994.

\bibitem[Lee91]{Leen1991}
Tood~K. Leen.
\newblock Dynamics of learning in linear feature-discovery networks.
\newblock {\em Network}, 2(1):85--105, 1991.

\bibitem[LG16]{L16}
Jean-Fran{\c{c}}ois Le~Gall.
\newblock {\em Brownian motion, martingales, and stochastic calculus}, volume
  274.
\newblock Springer, 2016.

\bibitem[LM18]{LM18}
Nancy~A. Lynch and Cameron Musco.
\newblock A basic compositional model for spiking neural networks.
\newblock {\em arXiv preprint arXiv:1808.03884}, 2018.

\bibitem[LMP17a]{LMP17ITCS}
Nancy~A. Lynch, Cameron Musco, and Merav Parter.
\newblock Computational tradeoffs in biological neural networks:
  Self-stabilizing winner-take-all networks.
\newblock In {\em 8th Innovations in Theoretical Computer Science Conference,
  {ITCS} 2017, January 9-11, 2017, Berkeley, CA, {USA}}, pages 15:1--15:44,
  2017.

\bibitem[LMP17b]{LMP17DISC}
Nancy~A. Lynch, Cameron Musco, and Merav Parter.
\newblock Neuro-ram unit with applications to similarity testing and
  compression in spiking neural networks.
\newblock In {\em 31st International Symposium on Distributed Computing, {DISC}
  2017, October 16-20, 2017, Vienna, Austria}, pages 33:1--33:16, 2017.

\bibitem[LMP17c]{LMP17BDA}
Nancy~A. Lynch, Cameron Musco, and Merav Parter.
\newblock Spiking neural networks: An algorithmic perspective.
\newblock In {\em Workshop on Biological Distributed Algorithms (BDA), July
  28th, 2017, Washington DC, USA}, 2017.

\bibitem[LMPV18]{LMPV18}
Robert~A. Legenstein, Wolfgang Maass, Christos~H. Papadimitriou, and
  Santosh~Srinivas Vempala.
\newblock Long term memory and the densest k-subgraph problem.
\newblock In {\em 9th Innovations in Theoretical Computer Science Conference,
  {ITCS} 2018, January 11-14, 2018, Cambridge, MA, {USA}}, pages 57:1--57:15,
  2018.

\bibitem[LMT19]{LM19}
Nancy~A. Lynch and Frederik Mallmann-Trenn.
\newblock Learning hierarchically structured concepts.
\newblock {\em arXiv preprint arXiv:1909.04559}, 2019.

\bibitem[LTYH09]{LTYH09}
Jian~Cheng Lv, Kok~Kiong Tan, Zhang Yi, and Sunan Huang.
\newblock A family of fuzzy learning algorithms for robust principal component
  analysis neural networks.
\newblock {\em IEEE Transactions on Fuzzy Systems}, 18(1):217--226, 2009.

\bibitem[LW19]{WL19}
Nancy~A. Lynch and Mien~Brabeeba Wang.
\newblock Integrating temporal information to spatial information in a neural
  circuit.
\newblock {\em arXiv preprint arXiv:1903.01217}, 2019.

\bibitem[LWLZ18]{LWLZ18}
Chris~Junchi Li, Mengdi Wang, Han Liu, and Tong Zhang.
\newblock Near-optimal stochastic approximation for online principal component
  estimation.
\newblock {\em Mathematical Programming}, 167(1):75--97, 2018.

\bibitem[OBL00]{OBL00}
Shan Ouyang, Zheng Bao, and Gui-Sheng Liao.
\newblock Robust recursive least squares learning algorithm for principal
  component analysis.
\newblock {\em IEEE Transactions on Neural Networks}, 11(1):215--221, 2000.

\bibitem[Oja82]{Oja82}
Erkki Oja.
\newblock Simplified neuron model as a principal component analyzer.
\newblock {\em Journal of mathematical biology}, 15(3):267--273, 1982.

\bibitem[Oja92]{Oja92}
Erkki Oja.
\newblock Principal components, minor components, and linear neural networks.
\newblock {\em Neural networks}, 5(6):927--935, 1992.

\bibitem[OK85]{OK85}
Erkki Oja and Juha Karhunen.
\newblock On stochastic approximation of the eigenvectors and eigenvalues of
  the expectation of a random matrix.
\newblock {\em Journal of mathematical analysis and applications},
  106(1):69--84, 1985.

\bibitem[Pea01]{Pearson01}
Karl Pearson.
\newblock On lines and planes of closest fit to systems of points in space.
\newblock {\em Philosophical Magazine}, 2:559--572, 1901.

\bibitem[Peh19]{Pehlevan19}
Cengiz Pehlevan.
\newblock A spiking neural network with local learning rules derived from
  nonnegative similarity matching.
\newblock In {\em {IEEE} International Conference on Acoustics, Speech and
  Signal Processing, {ICASSP} 2019, Brighton, United Kingdom, May 12-17, 2019},
  pages 7958--7962, 2019.

\bibitem[PHC15]{PHC15}
Cengiz Pehlevan, Tao Hu, and Dmitri~B. Chklovskii.
\newblock A hebbian/anti-hebbian neural network for linear subspace learning: A
  derivation from multidimensional scaling of streaming data.
\newblock {\em Neural computation}, 27(7):1461--1495, 2015.

\bibitem[Plu95]{Plumbley1995}
Mark~D. Plumbley.
\newblock Lyapunov functions for convergence of principal component algorithms.
\newblock {\em Neural Networks}, 8(1):11--23, 1995.

\bibitem[PV18]{PV19}
Christos~H. Papadimitriou and Santosh~S. Vempala.
\newblock Random projection in the brain and computation with assemblies of
  neurons.
\newblock In {\em 10th Innovations in Theoretical Computer Science Conference
  (ITCS 2019)}. Schloss Dagstuhl-Leibniz-Zentrum fuer Informatik, 2018.

\bibitem[RT89]{Rubner1989}
Jeanne Rubner and Paul Tavan.
\newblock A self-organizing network for principal-component analysis.
\newblock {\em Europhysics Letters ({EPL})}, 10(7):693--698, 1989.

\bibitem[SA06]{SA06}
Christian~D. Swinehart and Larry~F. Abbott.
\newblock Dimensional reduction for reward-based learning.
\newblock {\em Network: Computation in Neural Systems}, 17(3):235--252, 2006.

\bibitem[San89]{Sanger89}
Terence~D. Sanger.
\newblock Optimal unsupervised learning in a single-layer linear feedforward
  neural network.
\newblock {\em Neural networks}, 2(6):459--473, 1989.

\bibitem[SBW{\etalchar{+}}97]{Smirnakis1997}
Stellos~M. Smirnakis, Michael~J. Berry, David~K. Warland, William Bialek, and
  Markus Meister.
\newblock Adaptation of retinal processing to image contrast and spatial scale.
\newblock {\em Nature}, 386:69--73, 1997.

\bibitem[SCL19]{SCL19}
Lili Su, Chia-Jung Chang, and Nancy Lynch.
\newblock Spike-based winner-take-all computation: Fundamental limits and
  order-optimal circuits.
\newblock {\em arXiv preprint arXiv:1904.10399}, 2019.

\bibitem[SEC84]{Shapley1984}
Robert Shapley and Christina Enroth-Cugell.
\newblock Visual adaptation and retinal gain controls.
\newblock {\em Progress in Retinal Research}, 3:263--346, 1984.

\bibitem[Sej77]{S77}
Terrence~J. Sejnowski.
\newblock Storing covariance with nonlinearly interacting neurons.
\newblock {\em Journal of mathematical biology}, 4:303–--321, 1977.

\bibitem[Sha16]{Shamir16}
Ohad Shamir.
\newblock Convergence of stochastic gradient descent for pca.
\newblock In {\em International Conference on Machine Learning}, pages
  257--265, 2016.

\bibitem[SLS77]{Snyder1977}
Allan~W. Snyder, Simon~B. Laughlin, and Doekele~G. Stavenga.
\newblock {Information capacity of the eye}.
\newblock {\em Vision Research}, 17:1163--–75, 1977.

\bibitem[SLY06]{SLY06}
Lifeng Shang, Jian~Cheng Lv, and Zhang Yi.
\newblock Rigid medical image registration using pca neural network.
\newblock {\em Neurocomputing}, 69(13-15):1717--1722, 2006.

\bibitem[Wan95]{Wandell1995}
Brian~A. Wandell.
\newblock {\em Foundations of vision}.
\newblock Sunderland, MA: Sinauer, 1995.

\bibitem[WC72]{WC72}
Hugh~R. Wilson and Jack~D. Cowan.
\newblock Excitatory and inhibitory interactions in localized populations of
  model neurons.
\newblock {\em Biophysical journal}, 12(1):1--24, 1972.

\bibitem[WC73]{WC73}
Hugh~R. Wilson and Jack~D. Cowan.
\newblock A mathematical theory of the functional dynamics of cortical and
  thalamic nervous tissue.
\newblock {\em Kybernetik}, 13(2):55--80, 1973.

\bibitem[Wil91]{W91}
David Williams.
\newblock {\em Probability with martingales}.
\newblock Cambridge university press, 1991.

\bibitem[XOS92]{XOS92}
Lei Xu, Erkki Oja, and Ching~Y. Suen.
\newblock Modified hebbian learning for curve and surface fitting.
\newblock {\em Neural Networks}, 5(3):441--457, 1992.

\bibitem[Yan98]{Yan1998}
Wei-Yong Yan.
\newblock Stability and convergence of principal component learning algorithms.
\newblock {\em SIAM Journal on Matrix Analysis and Applications},
  19(4):933--955, 1998.

\bibitem[YHM94]{Yan1994}
Wei-Yong Yan, Uwe Helmke, and John~B. Moore.
\newblock Global analysis of oja’s flow for neural networks.
\newblock {\em IEEE Transactions on Neural Networks}, 5:674--683, 1994.

\bibitem[YYLT05]{YYLT05}
Zhang Yi, Mao Ye, Jian~Cheng Lv, and Kok~Kiong Tan.
\newblock Convergence analysis of a deterministic discrete time system of oja's
  pca learning algorithm.
\newblock {\em IEEE Transactions on Neural Networks}, 16(6):1318--1328, 2005.

\bibitem[Zuf02]{Zufiria02}
Pedro~J. Zufiria.
\newblock On the discrete-time dynamics of the basic hebbian neural network
  node.
\newblock {\em IEEE Transactions on Neural Networks}, 13(6):1342--1352, 2002.

\end{thebibliography}
\bibliographystyle{alpha}

\appendix
\section{Oja's derivation for the biological Oja's rule}\label{sec:oja derivation}
Recall that Oja wanted to use the following normalized update rule to solve the streaming PCA problem.
\begin{equation}\label{eq:oja normalized app}
\bw_t = \frac{\left(I+\eta_t\bx_t\bx_t^\top\right)\bw_{t-1}}{\|\left(I+\eta_t\bx_t\bx_t^\top\right)\bw_{t-1}\|_2} \, .
\end{equation}
Oja applied \textit{Taylor's expansion} on the normalization term and truncated the higher-order term of $\eta_t$. Concretely, we have
\begin{align}
\|\left(I+\eta_t\bx_t\bx_t^\top\right)\bw_{t-1}\|_2^{-1} &= \left(\sum_{i=1}^n\left(\bw_{t-1,i}+\eta_ty_t\bx_{t,i}\right)^2\right)^{-1/2}\nonumber\\
&=\left(\sum_{i=1}^n\bw_{t-1,i}^2+2\eta_ty_t\bx_{t,i}\bw_{t-1,i}+O(\eta_t^2)\right)^{-1/2} \, .\nonumber
\intertext{As $\ y_t=\bx_t^\top\bw_{t-1}$ and $\|\bw_{t-1}\|_2$ is expected to be $1$, the equation approximately becomes}
&=\left(1+2\eta_ty^2_t+O(\eta_t^2)\right)^{-1/2}=1-\eta_ty_t^2+O(\eta_t^2) \, .\label{eq:oja taylor app}
\end{align}
Replace the denominator of~\autoref{eq:oja normalized app} with~\autoref{eq:oja taylor app} and truncate the $O(\eta_t^2)$ term, one gets exactly~\autoref{eq:oja}.

\section{Details of the Linearizations in Continuous Oja's Rule}\label{sec:continuous oja linearization}
Recall that the dynamic of the continuous Oja's rule is the following.
\[
\frac{d\bw_t}{dt} = \diag(\lambda)\bw_t-\bw_t^\top\diag(\lambda)\bw_t\bw_t \, .
\]
Before proving the two convergence theorems of continuous Oja's rule using different linearizations, let us first prove the following lemma on some basic properties. 

\begin{lemma}[Properties of continuous Oja's rule]\label{lem:continuous properties}
Let $\bw_0\in\Real^n$ such that $\|\bw_0\|_2=1$ and $\bw_{0,1}>0$.
For any $t\geq0$, we have
\begin{enumerate}
\item $\|\bw_t\|_2=1$,
\item $\frac{d\bw_{t,1}}{dt}\geq(\lambda_1-\lambda_2)\bw_{t,1}(1-\bw_{t,1}^2)$, and
\item $\bw_{t,1}$ is non-decreasing
\end{enumerate}
almost surely.
\end{lemma}
\begin{proof}[Proof of~\autoref{lem:continuous properties}]
In the following, everything holds almost surely so we would not mention this condition every time.
First, consider
\begin{align*}
\frac{d\|\bw_t\|_2^2}{dt} &= 2\bw_t^\top\frac{d\bw_t}{dt} = 2\bw_t^\top\left(\diag(\lambda)\bw_t-\bw_t^\top\diag(\lambda)\bw_t\bw_t\right)\\
&= 2\bw_t^\top\diag(\lambda)\bw_t\cdot\left(1-\|\bw_t\|_2^2\right) \, .
\end{align*}
As $1-\|\bw_0\|_2^2=0$, by induction, we have $\|\bw_t\|_2=1$ for all $t\geq0$.

For the second item of the lemma, we have
\begin{align*}
\frac{d\bw_{t,1}}{dt}&=\left(\lambda_1-\left(\sum_{i\in[n]}\lambda_i\bw_{t,i}^2\right)\right)\bw_{t,1} \geq (\lambda_1 - \lambda_2)\bw_{t, 1}(1 - \bw_{t, 1}^2)\\
&=\lambda_1(\bw_{t,1}-\bw_{t,1}^3) - \sum_{i=2}^n\lambda_i\bw_{t,i}^2\bw_{t,1} \, .
\intertext{From the first item, we have $\sum_{i=2}^n\bw_{t,i}^2=1-\bw_{t,1}^2$. Thus, we have}
&\geq\lambda_1(\bw_{t,1}-\bw_{t,1}^3) - \lambda_2(1-\bw_{t,1}^2)\bw_{t,1}=(\lambda_1-\lambda_2)\bw_{t,1}(1-\bw_{t,1}^2) \, .
\end{align*}

The last item of the lemma is then an immediate corollary of the first two items.
\end{proof}

Now, we restate and prove~\autoref{thm:continuous lin at 0} as follows.
\contlinzero*
\begin{proof}[Proof of~\autoref{thm:continuous lin at 0}]
Observe that for any $t\geq0$ such that $\bw_{t, 1}^2\leq 1 -\epsilon$, by the second item of~\autoref{lem:continuous properties}, we have
\[
\frac{d\bw_{t,1}}{dt}\geq (\lambda_1 - \lambda_2)\bw_{t, 1}(1 - \bw_{t, 1}^2) \geq \epsilon(\lambda_1 - \lambda_2) \bw_{t, 1} \, .
\]
Let $\tau=\frac{10\log(1/\bw_{0,1}^2)}{\epsilon(\lambda_1-\lambda_2)}$ and assume $\bw_{\tau,1}^2\leq1-\epsilon$ for the sake of contradiction. From the above linearization and $\bw_{t,1}$ being non-decreasing (the third item of~\autoref{lem:continuous properties}), we have
\[
\bw_{\tau,1}\geq e^{\epsilon(\lambda_1-\lambda_2)\tau}\cdot\bw_{0,1}>1 \, ,
\]
which is a contradiction to the first item of~\autoref{lem:continuous properties}. Thus, we conclude that for any $t=\Omega\left(\frac{\log(1/\bw_{0,1}^2)}{\epsilon(\lambda_1-\lambda_2)}\right)$, $\bw_{t,1}^2>1-\epsilon$.
\end{proof}

Now, we restate and prove~\autoref{thm:continuous lin at 1} as follows.
\contlinone*
\begin{proof}[Proof of~\autoref{thm:continuous lin at 1}]
Observe that for any $t\geq0$, by the second item of~\autoref{lem:continuous properties}, we have
\begin{align*}
\frac{d(\bw_{t, 1}-1)}{dt} &\geq (\lambda_1 - \lambda_2)\bw_{t, 1}(1 - \bw_{t, 1}^2)\\
&= -(\lambda_1 - \lambda_2)(\bw_{t, 1} - 1)(\bw_{t, 1} + \bw_{t, 1}^2) \, .
\intertext{As $\bw_{t,1}$ is non-decreasing (the third item of~\autoref{lem:continuous properties}) and at most 1, we have}
&\geq -(\lambda_1 - \lambda_2)\bw_{0,1}(\bw_{t, 1} - 1) \, .
\end{align*}
By solvign the linear ODE, we have
\[
\bw_{t, 1}-1\geq(\bw_{0,1}-1)\cdot e^{-(\lambda_1-\lambda_2)\bw_{0,1}t} \, .
\]
Thus, for any $t\geq \Omega\left(\frac{\log(1/\epsilon)}{\bw_{0,1}(\lambda_1 - \lambda_2)}\right)$, we have $\bw_{t, 1}^2 > 1 - \epsilon$ \, . 
\end{proof}

\section{Why the Analysis of ML Oja's Rule Cannot be Applied to Biological Oja's Rule?}\label{sec:compare bio ML}
In this section, we would like to discuss what makes biological Oja's rule much harder to analyze comparing to the previous approaches for ML Oja's rule. We study this problem through the lens of their corresponding continuous dynamics. Observe that, to study ML Oja's rule, it suffices to study the following dynamic
\[
\frac{d\bw_t}{dt} = \diag(\lambda)\bw_t \, .
\]
The dynamic of the objective function $\sum_{i=2}^n\bw_{t,i}^2/\bw_{t,1}^2$ would be
\begin{align*}
\frac{d\frac{\sum_{i=2}^n\bw_{t,i}^2}{\bw_{t,1}^2}}{dt} &= \frac{-2\sum_{i=2}^n\bw_{t,i}^2}{\bw_{t,1}^3}\lambda_1\bw_{t,1} + \sum_{i=2}^n\frac{2\bw_{t,i}}{\bw_{t,1}^2}\lambda_i\bw_{t,i}\\
&\leq -2(\lambda_1 - \lambda_2)\frac{\sum_{i=2}^n\bw_{t,i}^2}{\bw_{t,1}^2} \, .
\end{align*}
Namely, the continuous dynamic is just a linear ODE with \textit{slope} being independent to the value of $\bw_t$. In comparison, the dynamic of the biological Oja's rule is the following.
\[
 \frac{d\bw_{t,1}}{dt}\geq(\lambda_1-\lambda_2)\bw_{t,1}(1-\bw_{t,1}^2)
\]
where you have to use at least two objective functions with different linearizations to get tight analysis. Furthermore, for any linearization, there exist some values of $\bw_t$ that make the improvement extremely small or even vanishing. It is also not obvious to choose which two objective functions to analyze unless you are guided by the continuous dynamics. 

We remark that the discussion here only suggests the difficulty of applying previous techniques of ML Oja's rule to biological Oja's rule. It might still be the case that the two dynamics are coupled but we argue here that even this is the case, previous techniques cannot show this.

\section{Proof of Lemma~\ref{lem:D diff var}}\label{sec: annoying diff var}

\begin{proof}[Proof of~\autoref{lem:D diff var}] 
The proof is basically direct verification using the definition of $\xi, \tau_t, \psi$ and~\autoref{lem: conditional stopped help}. Let's first describe $\nabla f_{t, j}(\bw_{s-1})$ and $\nabla^2 f_{t, j}(\bw_{s-1})$ and give their corresponding bounds. We have 
\[
(\nabla f_{t, j}(\bw_{s-1}))_1 = \frac{- f_{t, j}(\bw_{s-1})}{\bw_{s-1, 1}}, \ \forall 1 < i \leq j, (\nabla f_{t, j}(\bw_{s-1}))_i = \frac{\bx_{t,i}}{\bw_{s-1,1}}
\]
and all other coordinates are zero. In particular, conditioning on $\tau_t, \psi\geq s$, we have
\begin{equation}\label{eq: first derivative bound}
\|\nabla f_{t, j}(\bw_{s-1})\|_2 = O(\frac{\Lambda}{\bw_{s-1, 1}^2}) = O(\sqrt{\Lambda'}\Lambda) \, .
\end{equation}
For $\nabla^2 f_{t, j}(\bw_{s-1})$, we have
\[
(\nabla^2 f_{t, j}(\bw_{s-1}))_{1, 1} = \frac{\sum_{i=2}^j\bx_{t, i}\bar{\bw}_{s-1, i}}{\bar{\bw}_{s-1, 1}^3} \, ,
\]
\[
\forall 1 < i \leq j,\ (\nabla^2 f_{t, j}(\bw_{s-1}))_{1, i} = (\nabla^2 f_{t, j}(\bw_{s-1}))_{i, 1} = -\frac{\bx_{t, i}}{\bar{\bw}_{s-1, 1}^2}
\]
and all other coordinates are zero. In particular, we can rewrite it as linear combination of three rank one matrices
\[
\nabla^2 f_{t, j}(\bw_{s-1}) = \alpha_1\bx_{t}^{(j, 1)}{\bx_{t}^{(j, 1)}}^\top + \alpha_2\bx_{t}^{(j, 0)}{\bx_{t}^{(j, 0)}}^\top +  \alpha_3e_1e_1^\top 
\]
where
\[
\alpha_1 = -\frac{1}{\bar{\bw}_{s-1}^2},\, \alpha_2 = \frac{1}{\bar{\bw}_{s-1}^2},\, \alpha_3 = \frac{\sum_{i=2}^j \bx_{t, i}\bar{\bw}_{s-1, i}}{\bar{\bw}_{s-1}^3} + \frac{1}{\bar{\bw}_{s-1}^2}\, \text{, and}
\]
$e_1$ is the basis vector of first coordinate and $\bx_{t, i}^{(j, a)} = \bx_{t, i}$ if $1 <i\leq j$, $\bx_{t, 1}^{(j, a)} = a$ and it is zero at all other coordinates. Now we would like to bound the coefficient. Notice that since $\eta = O(\frac{1}{\Lambda})$, 
\[
\bar{\bw}_{s-1, i} = \bw_{s-1, i} + c\eta\bz_{s, i} = \bw_{s-1, i} + O(\bw_{s-1, 1}\bx_{s, i} + \eta\bw_{s-1, i}) \, .
\]
In particular, $\bar{\bw}_{s-1, i} = O(\bw_{s-1, i} + \bw_{s-1, 1}\bx_{s, i})$. Now we can bound the coefficient $|\alpha_1| = O(\frac{1}{\bw_{s-1, 1}^2})$, $|\alpha_2| = O(\frac{1}{\bw_{s-1, 1}^2})$ and $|\alpha_3| = O\left(\frac{\Lambda}{\bw_{s-1, 1}^2}\right).$ In particular, given any vector $v$, we have
\begin{align}
|v^\top \nabla^2 f_{t, j}(\bw_{s-1})v| &=\left|\alpha_1v^\top \bx_{t}^{(j, 1)}{\bx_{t}^{(j, 1)}}^\top v + \alpha_2v^\top \bx_{t}^{(j, 0)}{\bx_{t}^{(j, 0)}}^\top v +  \alpha_3v^\top e_1e_1^\top v \right|\nonumber\\
&= \left| \alpha_1{\bx_{t}^{(j, 1)}}^\top vv^\top {\bx_{t}^{(j, 1)}} + \alpha_2{\bx_{t}^{(j, 0)}}^\top vv^\top {\bx_{t}^{(j, 0)}} +  \alpha_3e_1^\top vv^\top e_1 \right| \, .\nonumber
\intertext{By combining the bound $\alpha_i =  O\left(\frac{\Lambda}{\bw_{s-1, 1}^2}\right)$, $\|vv^\top \|_2 = \|v\|^2_2$ and $\|e_1\|_2$, $\|\bx_{t}^{(j, 0)}\|_2$, $\|\bx_{t}^{(j, 1)}\|_2\leq 2$, we have}
&\leq O\left(\frac{\Lambda}{\bw_{s-1,1}^2}\|v\|_2^2\right) \, .\label{eq: second derivative bound}
\end{align}
Now we are ready to analyze the bounds on $\bA_{s, j}^{(t)}$. For notational convenience, denote $\bz_{s}-\Exp[\bz_{s}\ |\ \mathcal{F}_{s-1}]$ as $\bar{\bz}_s$ and separate $\bA_{s, j}^{(t)}$ into two terms where $\bA_{s, j}^{(t, 1)} = \eta \nabla f_{t,j}(\bw_{s-1})^\top \bar{\bz}_s$ and $\bA_{s, j}^{(t, 2)} = \eta^2 \bz_{s}^\top \nabla f_{t,j}^2(\overline{\bw}_{s-1})\cdot\bz_{s}.$ By Cauchy-Schwarz and~\autoref{eq: first derivative bound}, We have
\[
|\bA_{s, j}^{(t, 1)}| \leq \eta\|\nabla f_{t, j}(\bw_{s-1})\|_2\|\bar{\bz}_{s}\|_2 = O\left(\eta\cdot \frac{\Lambda}{\bw_{s-1}}\cdot y_s\right)= O(\eta\Lambda^2) \, .
\]
We also have 
\begin{align*}
|\bA_{s, j}^{(t, 2)}| &= |\eta^2\bz_{s}^\top \nabla^2 f_{t, j}(\bw_{s-1})\bz_{s}| \, .\\
\intertext{By~\autoref{eq: second derivative bound}, we have}
&= O\left(\eta^2 \frac{\Lambda}{\bw_{s-1, 1}^2}\|\bz_{s}\|_2^2\right) \, .
\intertext{Because $ \|\bz_{s}\|_2^2 = O(y_s^2)$, we have}
&= O\left(\eta^2 \frac{\Lambda}{\bw_{s-1, 1}^2}y_s^2\right)\\ 
&= O(\eta^2\Lambda^3) = O(\eta\Lambda^2) \, .
\end{align*}
This gives us $|\mathbf{1}_{\psi, \tau_t \geq s, \xi > s} \bA^{(t)}_{s, j}| = O(\eta\Lambda^2)$. 
For conditional expectation, we have 
\begin{align*}
\left|\Exp[\mathbf{1}_{\psi, \tau_t \geq s, \xi > s}\bA_{s, j}^{(t, 1)}\, |\, \mathcal{F}_{s-1}^{(t)}, \mathcal{C}_{init}^{p,\delta}]\right| & = \left|\Exp[\mathbf{1}_{\psi, \tau_t \geq s, \xi > s}\eta \nabla f_{t,j}(\bw_{s-1})^\top \bar{\bz}_{s}| \mathcal{F}_{s-1}^{(t)}, \mathcal{C}_{init}^{p,\delta}]\right| \, . \\
\intertext{Notice that we have $\Exp[\bz_{s}\, |\, \mathcal{F}_{s-1}^{(t)}, \mathcal{C}_{init}^{p,\delta}] = \Exp[\bx_{s}\bx_s^\top \bw_{s-1} - \bw_{s-1}^\top \bx_{s}\bx_s^\top \bw_{s-1}\bw_{s-1}\, |\, \mathcal{F}_{s-1}^{(t)}, \mathcal{C}_{init}^{p,\delta}]$. By~\autoref{lem: conditional stopped help} applying on $\bx_s\bx_s^\top $ and Cauchy-Schwarz, we have}
&\leq O\left(\eta\|\nabla f_{t,j}(\bw_{s-1})\|_2\frac{1}{nT}\|\bx_s\|_2) + \eta\|\nabla f_{t,j}(\bw_{s-1})\|_2\|\bw_{s-1}\|_2^3\frac{1}{nT}\right) \, .
\intertext{By~\autoref{eq: first derivative bound} and $T = \Omega(\frac{1}{\eta\lambda_1})$, we have}
&\leq O\left(\frac{\eta^2\lambda_1\Lambda^2}{\sqrt{n}}\right) = O(\eta^2\lambda_1\Lambda^3) \, .
\end{align*}
For $\bA_{s, j}^{(t, 2)}$, we have
\begin{align*}
\left|\Exp[\mathbf{1}_{\psi, \tau_t \geq s, \xi > s}\bA_{s}^{(t, 2)}\, |\, \mathcal{F}_{s-1}^{(t)}, \mathcal{C}_{init}^{p,\delta}]\right| & = \left|\Exp[\mathbf{1}_{\psi, \tau_t \geq s, \xi > s}\eta^2 \bz_{s}^\top \nabla^2 f_{t,j}(\bw_{s-1})^\top \bz_{s}| \mathcal{F}_{s-1}^{(t)}, \mathcal{C}_{init}^{p,\delta}]\right| \, . \\
\intertext{Notice we have $\|\Exp[\mathbf{1}_{\psi, \tau_t \geq s, \xi > s}\bz_{s}\bz_{s}^\top\, |\, \mathcal{F}_{s-1}^{(t)}]\|_2 = O(y_s^2\lambda_1)$ by~\autoref{lem: conditional stopped help}. Again by~\autoref{eq: second derivative bound}, we have}
&\leq O\left(\eta^2 y_s^2\lambda_1 \frac{\Lambda}{\bw_{s-1, 1}^2}\right) = O(\eta^2\lambda_1\Lambda^3) \, .
\end{align*}
So we have
\[
\left|\Exp[\mathbf{1}_{\psi, \tau_t \geq s, \xi > s}\bA_{s}^{(t)}\, |\, \mathcal{F}_{s-1}^{(t)}, \mathcal{C}_{init}^{p,\delta}]\right| = O(\eta^2\lambda_1\Lambda^3) \, .
\]
For the last moment bound, fix $2\leq j, j'\leq n$. Expanding the definition, we get
\[
\Exp[\mathbf{1}_{\psi, \tau_t \geq s, \xi > s}\bA_{s, j}^{(t, 1)}\bA_{s, j'}^{(t, 1)} + \bA_{s, j}^{(t, 1)}\bA_{s, j'}^{(t, 2)} +\bA_{s, j}^{(t, 2)}\bA_{s, j'}^{(t, 1)} + \bA_{s, j}^{(t, 2)}\bA_{s, j'}^{(t, 2)}  \, |\, \mathcal{F}_{s-1}, \mathcal{C}_{init}^{p,\delta}] \, .
\]
For the first term, we have
\begin{align*}
|\Exp[\mathbf{1}_{\psi, \tau_t \geq s, \xi > s}\bA_{s, j}^{(t, 1)}\bA_{s, j'}^{(t, 1)} \, |\, \mathcal{F}_{s-1}, \mathcal{C}_{init}^{p,\delta}]| &= |\Exp[\mathbf{1}_{\psi, \tau_t \geq s, \xi > s}\eta^2\nabla f_{t, j}(\bw_{s-1})^\top \bar{\bz}_{s}\bar{\bz}_{s}^\top \nabla f_{t, j'}(\bw_{s-1}) \, |\, \mathcal{F}_{s-1}, \mathcal{C}_{init}^{p,\delta}]| \, .
\intertext{Notice we have $\|\Exp[\mathbf{1}_{\psi, \tau_t \geq s, \xi > s}\bar{\bz}_{s}\bar{\bz}_{s}^\top\, |\, \mathcal{F}_{s-1}^{(t)}, \mathcal{C}_{init}^{p,\delta}]\|_2 = O(y_s^2\lambda_1)$ by~\autoref{lem: conditional stopped help}. We have}
&\leq O(\eta^2\|\nabla f_{t, j}(\bw_{s-1})\|_2y_s^2\lambda_1\|\nabla f_{t, j'}(\bw_{s-1})\|_2) \, .
\intertext{Since we know $\|\nabla f_{t, j}(\bw_{s-1})\|_2 = O\left(\frac{\Lambda}{\bw_{s-1,1}^2}\right)$, we have}
&= O(\eta^2\lambda_1\Lambda^4) \, .
\end{align*}
For the second and third term, since they are symmetric, we will only deal with the second term. We have
\begin{align*}
|\Exp[\mathbf{1}_{\psi, \tau_t \geq s, \xi > s}\bA_{s, j}^{(t, 1)}\bA_{s, j'}^{(t, 2)} \, |\, \mathcal{F}_{s-1}, \mathcal{C}_{init}^{p,\delta}]| &= |\Exp[\mathbf{1}_{\psi, \tau_t \geq s, \xi > s}\eta^3\nabla f_{t, j}(\bw_{s-1})^\top \bar{\bz}_{s}\bz_{s}^\top \nabla^2 f_{t, j'}(\bw_{s-1})\bz_{s}^\top \, |\, \mathcal{F}_{s-1}, \mathcal{C}_{init}^{p,\delta}]| \, .
\intertext{By taking the maximum of the $\bA_{s, j}^{(t, 1)}$ and combining with~\autoref{eq: second derivative bound}, we have}
&\leq O(\eta\Lambda^2\cdot \eta^2\lambda_1\Lambda^3)\\
&= O(\eta^3\lambda_1\Lambda^5) = O(\eta^2\lambda_1\Lambda^4) \, .
\end{align*}
For the last term, we can deal with it completely analogously. In particular we have
\[
|\Exp[\mathbf{1}_{\psi, \tau_t \geq s, \xi > s}\bA_{s, j}^{(t, 2)}\bA_{s, j'}^{(t, 2)} \, |\, \mathcal{F}_{s-1}, \mathcal{C}_{init}^{p,\delta}]| \leq O(\eta\Lambda^2\cdot \eta^2\lambda_1\Lambda^3\cdot )= O(\eta^2\lambda_1\Lambda^4) \, .
\]
Combining all the terms, we get
\[
|\Exp[\mathbf{1}_{\psi, \tau_t \geq s, \xi > s}\bA_{s, j}^{(t)}\bA_{s, j'}^{(t)} \, |\, \mathcal{F}_{s-1}, \mathcal{C}_{init}^{p,\delta}]| = O(\eta^2\lambda_1\Lambda^4) \, .
\]
\end{proof}

\end{document}